\documentclass[pra,
aps,
onecolumn,
superscriptaddress,
nofootinbib,
floatfix,
10pt
]{revtex4-2}

\pdfoutput=1

\usepackage{soul}
\usepackage[labelfont=bf,font=small,singlelinecheck=false,justification=justified,format=plain]{caption}
\usepackage{physics}
\usepackage{amsmath,amsthm,mathtools,amssymb}
\usepackage{dsfont}
\usepackage{graphicx}
\usepackage{bm}
\usepackage{bbm}
\usepackage{multirow}
\usepackage{cancel}
\usepackage{graphicx}
\usepackage[table,xcdraw]{xcolor}
\usepackage{enumitem}
\usepackage{tensor}
\usepackage{tikz}
\usepackage{pgfplots}
\pgfplotsset{compat=1.18}
\usetikzlibrary{quantikz2}
\usepgfplotslibrary{fillbetween}
\pgfmathdeclarefunction{h}{1}{\pgfmathparse{0.8*sqrt(#1)}}
\usepackage[colorlinks=true,allcolors=teal]{hyperref}
\usepackage[capitalize]{cleveref}
\usepackage{blindtext}
\usepackage{titlesec}

\DeclareFontFamily{U}{MnSymbolA}{}
\DeclareFontShape{U}{MnSymbolA}{m}{n}{
    <-6>  MnSymbolA5
   <6-7>  MnSymbolA6
   <7-8>  MnSymbolA7
   <8-9>  MnSymbolA8
   <9-10> MnSymbolA9
  <10-12> MnSymbolA10
  <12->   MnSymbolA12}{}
\newcommand{\MnArrow}[1]{{\usefont{U}{MnSymbolA}{m}{n}\symbol{#1}}}

\newcommand{\Ish}{\!\!\rotatebox[origin=c]{135}{\MnArrow{194}}} %
\newcommand{\Dsh}{\!\!\rotatebox[origin=c]{45}{\MnArrow{184}}}  %

\crefname{table}{Table}{Tables}
\crefname{equation}{Eq.}{Eqs.}
\crefname{figure}{Figure}{Figures}
\crefname{section}{Sec.}{Secs.}
\crefname{appendix}{Appendix}{Appendices}
\crefname{proposition}{Prop.}{Props.}
\crefname{theorem}{Theorem}{Theorems}
\crefname{lemma}{Lemma}{Lemmas}
\crefname{corollary}{Corollary}{Corollaries}

\def\id{ {\mathbbm 1} }

\def\E{ {\mathcal E} }

\def\H{ {\mathcal H} }

\def\N{ {\mathcal N} }
\def\O{ {\mathcal O}_{q,q'} }

\def\S{ {\mathcal S} }

\def\U{ {\mathcal U} }

\def\bmo{ {\bm{0}} }

\def\bmd{ {\bm{d}} }

\newcommand{\fmax}{f^\text{max}_i}
\newcommand{\fmin}{f^\text{min}_i}
\newcommand{\fu}{F^U_i}
\newcommand{\fl}{F^L_i}

\newcommand{\qmax}{q^\text{max}_i}
\newcommand{\qmin}{q^\text{min}_i}
\newcommand{\qu}{Q^U_i}
\newcommand{\ql}{Q^L_i}

\newcommand{\amaxp}{A_1^\text{max}}
\newcommand{\amaxn}{A_2^\text{max}}

\newcommand{\concat}{\overset{\frown}{=}}

\newcommand{\declor}[1]{L^{\Dsh}_{#1}(s)}
\newcommand{\inclor}[1]{L^{\Ish}_{#1}(s)}

\newcommand{\qdeclor}[1]{L^{\Dsh}_{#1|q}(s)}
\newcommand{\qinclor}[1]{L^{\Ish}_{#1|q}(s)}
\newcommand{\qpdeclor}[1]{L^{\Dsh}_{#1|q'}(s)}
\newcommand{\qpinclor}[1]{L^{\Ish}_{#1|q'}(s)}

\newcommand{\qqdeclor}[2]{L^{\Dsh}_{#1|#2}(s)}

\newcommand{\qqqdeclor}[3]{L^{\Dsh}_{#1|#2}(#3)}
\newcommand{\qqqinclor}[3]{L^{\Ish}_{#1|#2}(#3)}

\newcommand{\weak}{\succ\!\succ}

\newtheorem{theorem}{Theorem}
\newtheorem{lemma}[theorem]{Lemma}
\newtheorem{corollary}{Corollary}

\newtheorem{definition}{Definition}

\newenvironment{mypsmallmatrix}
  {\left(\begin{smallmatrix}}
  {\end{smallmatrix}\right)}

\begin{document}

\title{Fundamental limits on photonic state preparation from phase-space majorization}

\author{Twesh Upadhyaya}
\email{tweshu@umd.edu}
\affiliation{Joint Center for Quantum Information and Computer Science, University of Maryland and NIST, College Park, MD 20742, USA}
\affiliation{Department of Physics, University of Maryland, College Park, MD 20742, USA}

\author{Oliver Hahn}
\affiliation{Department of Basic Science, The University of Tokyo, 3-8-1 Komaba, Meguro-ku, Tokyo 153-8902, Japan}

\author{Ulysse Chabaud}
\affiliation{QAT team, DIENS, \'Ecole Normale Sup\'erieure, PSL University, CNRS, INRIA, 45 rue d'Ulm, Paris 75005, France}

\author{Nikolaos Koukoulekidis}
\affiliation{Duke Quantum Center and Departments of Physics and Electrical and Computer Engineering, Duke University, Durham, NC, USA}

\date{October 27, 2026}

\begin{abstract}
Preparing complex quantum states of light for quantum information processing tasks often proceeds by converting from simpler states via experimentally accessible operations. We provide fundamental bounds on such probabilistic conversions, both exact and approximate. To do so, we establish a theory of weak, fully relative majorization that compares the disorder of phase-space functions. Our framework is versatile and can be tailored to different systems and sets of allowed operations. Focusing on Wigner functions in multimode photonic systems, we establish rigorous no-go results and tradeoff relations for exact and approximate conversions involving squeezed states, cat states, Gottesman--Kitaev--Preskill states, and Fock states. In doing so, we derive novel bounds and recover known tight ones. Our results demonstrate that phase-space disorder imposes strong limitations on photonic state preparation.
\end{abstract}

\maketitle

\section{Introduction}

Quantum information processing exploits features of quantum systems as resources for computational tasks~\cite{shor1994algorithms,grover1996,Harrow2017}, communication tasks~\cite{wiesnerconjugate1983,bennett2014quantum}, and sensing tasks~\cite{Degen2017} that are inaccessible, or substantially more costly, with classical systems. 
Photonic platforms are particularly appealing in this context: photons propagate rapidly; interact weakly with their environment; can encode information in several degrees of freedom, including continuous-variable ones; and can be generated, manipulated, and detected using increasingly large integrated devices~\cite{o2009photonic,wang2020integrated}. As a result, photonic quantum experiments have been at the forefront of landmark demonstrations of both fundamental and applied quantum protocols, including the violation of Bell inequalities~\cite{aspect1982experimental}, quantum teleportation~\cite{ursin2004quantum,jin2010experimental}, %
and quantum computational speedup~\cite{zhong2020quantum}, among others.

Harnessing the potential of photonic quantum technologies, however, requires the preparation of complex photonic quantum states. For instance, photonic quantum information processing limited to Gaussian states and operations cannot perform entanglement distillation~\cite{giedke2002characterization,eisert2002distilling,fiuravsek2002gaussian} or achieve quantum computational advantage~\cite{Bartlett2002}, let alone universal quantum computing~\cite{Lloyd1999,arzani2025effective}. Moreover, non-Gaussian quantum states such as cat states~\cite{guillaud2019repetition}, and Gottesman--Kitaev--Preskill (GKP) states~\cite{Gottesman2001} are instrumental for fault-tolerant photonic quantum computing~\cite{niset2009no,Baragiola_2019}.

The standard approach to photonic state preparation is to combine elementary input states using simple operations with conditional measurements. Experimentally accessible building blocks include coherent states, squeezed states, low-photon-number Fock states, and small-amplitude cat states. On the side of operations, passive linear optics and, more generally, Gaussian operations can be implemented with high control~\cite{lvovsky2020production}. At the measurement level, photon number measurements, homodyne or heterodyne detection and postselection then supply an effective nonlinearity, leading to a variety of state preparation protocols~\cite{ourjoumtsev2006generating,Ourjoumtsev2007generation,Etesse:14,Arzani2017polynomial,Weigand2018generating,zheng2020gaussian,hahn2022deterministic,pizzimenti2024optical,takase2024generation}.

Constructing a specific protocol to prepare a target quantum state provides a lower bound on the achievable preparation fidelity and success probability, but gives little information when the search for such a protocol fails. This raises the following question:
\begin{center}
   \textit{What are the fundamental limits of photonic state preparation?} 
\end{center}
A useful theory of state preparation should robustly identify impossible transformations and remain applicable when ancillary modes, heralding measurements, and experimental imperfections are included, over a variety of classes of experimentally accessible operations.

Several techniques have been developed to address this problem, each focused on a predetermined set of experimentally accessible operations. 
These include Lie-algebraic invariants for conversions with passive linear optics~\cite{parellada2023no,parellada2024lie,rodari2025observation,draux2025invariants}, and polynomial representation methods for states with few photons~\cite{de2024simple,vanmeter2007general,singh2026rigorous}.
For protocols built from Gaussian operations, the resource theory of quantum non-Gaussianity~\cite{takagi2018convex,albarelli2018resource} provides a different family of obstructions. In order to derive necessary conditions for state conversion using Gaussian operations, many resource measures have been introduced~\cite{takagi2018convex,albarelli2018resource,chabaud2020stellar,Hahn2026assessingnon,Hahn2025classicalsimulation}. Such Gaussian monotones lead to no-go theorems for exact and approximate state preparation using Gaussian protocols. Recently, \emph{majorization} of phase-space distributions~\cite{koukoulekidis2022constraints} has also been used to assess exact, deterministic state preparation~\cite{VanHerstraeten2023continuous,de2024continuous}, including a general framework in Ref.~\cite{Upadhyaya2026}. %

In this work, we develop a powerful methodology for assessing photonic state preparation via phase-space disorder. Unlike many previous works, our framework has the following desirable features, each enabled by a distinct technical contribution of this work: (i) The framework allows the study of several classes of probabilistic protocols, rather than being limited to a fixed set of operations, as we extend the phase-space majorization theory from Ref.~\cite{Upadhyaya2026} to the setting of \textit{weak and fully relative majorization}. (ii) The framework provides analytical and rigorous numerical limitations, as we develop techniques to analytically compute Lorenz curves and divergences and we introduce the concept of \textit{robust Lorenz curves} which provide certified constraints that account for numerical error. (iii) The framework allows one to establish limits on exact and approximate state preparation, as we derive explicit quantitative lower semicontinuity bounds for Wigner--Rényi divergences, that are also extensible to Wigner Lorenz curves.

We apply our framework to a range of linear-optical and Gaussian state-preparation protocols, including cat-state distillation, GKP magic-state preparation, and Fock-state concentration. In doing so, we develop new constraints for photonic state preparation and recover known tight no-goes. Our results are summarized in \cref{fig:summary}.

\begin{figure}
     \centering
     \includegraphics[width=\linewidth]{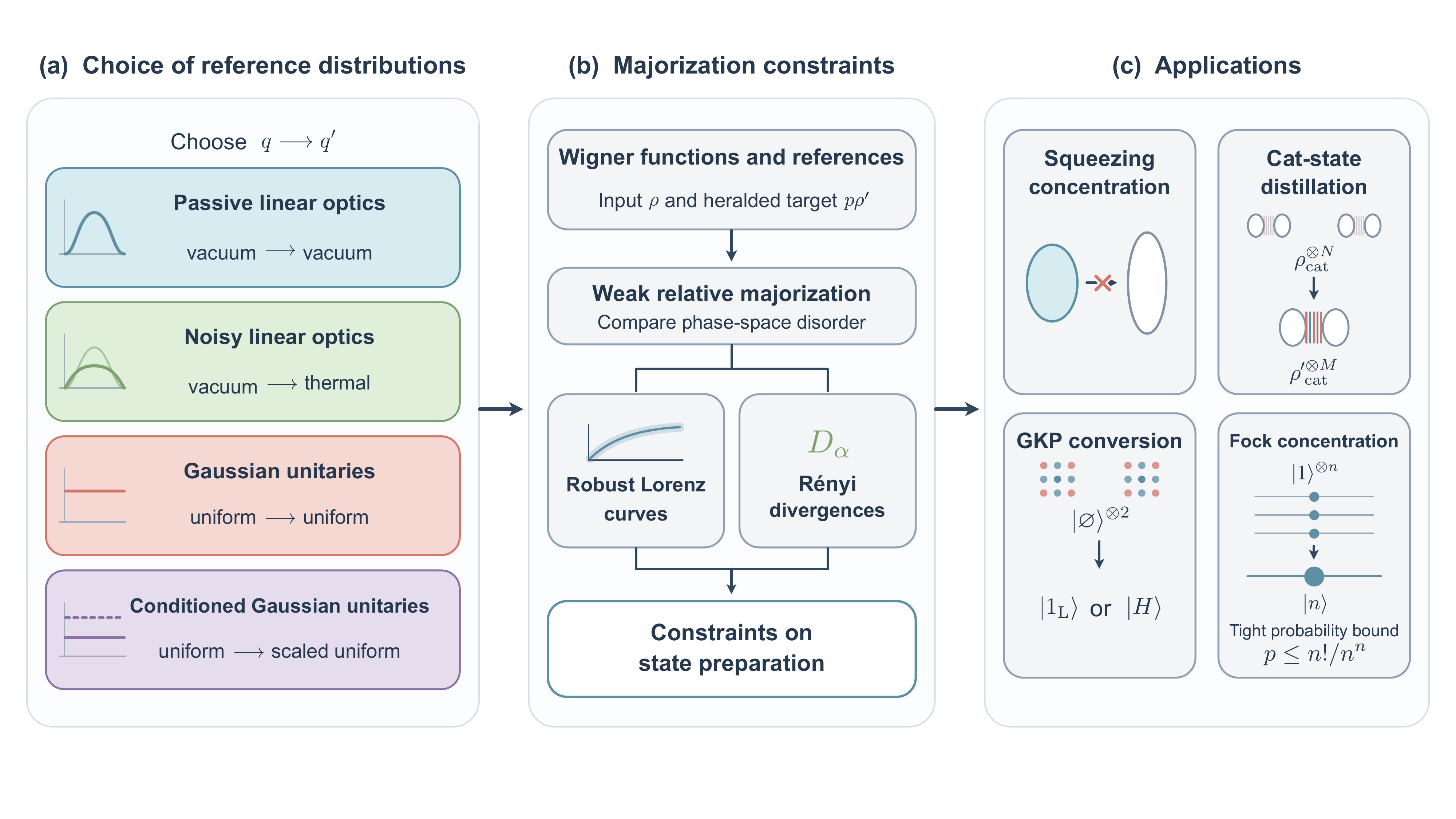}
     \caption{Schematic depiction of our majorization framework for assessing quantum state preparation. (a) We first pick input/output reference distributions $q$ and $q'$, which correspond to a class of protocols. (b) Then, we choose a state preparation scenario specified by a choice of input state $\rho$ and output state $p\rho'$, where $p$ is the success probability, and we apply our relative majorization results; these yield robust Lorenz curves and Wigner--Rényi divergences which allow us to deduce constraints on the state conversion $\rho\to\rho'$, based on the class of protocols analyzed. (c) We illustrate this framework with several applications, including robust bounds on success probability for scenarios involving cat states, GKP states and Fock states and reproducing known no-go results for Fock number and squeezing concentration.
     }
     \label{fig:summary}
\end{figure}

The paper is structured as follows: in \cref{sec:preliminaries}, we provide the technical background needed to appreciate the framework; in \cref{sec:majorization}, we relate majorization on quasiprobabilities to probabilistic exact and approximate state conversion tasks; in \cref{sec:robustcurves}, we introduce robust Lorenz curves; and in \cref{sec:conversion}, we apply the tools we have developed through divergences and Lorenz curves to assess several experimental scenarios of photonic state preparation.

\section{Preliminaries}
\label{sec:preliminaries}

In this section, we provide a short review of bosonic states and channels, including their phase-space representations, and majorization. For more background, we refer the reader to Refs.~\cite{Gerry2004,Weedbrook2012-qu,BUCCO} for the former and Refs.~\cite{marshall2011inequalities,Joe1990,Upadhyaya2026} for the latter.

\subsection{Bosonic quantum information theory}

We consider photonic/bosonic systems consisting of a finite number $N$ of modes. The modes represent spatial or temporal degrees of freedom~\cite{Weedbrook2012-qu}.
The Hilbert space of a single mode is
$\mathcal H= L^2(\mathbb R)$ and admits a countable orthonormal basis. Denote multimode Hilbert spaces by $\mathcal{H}_N \coloneqq {\mathcal H}^{\otimes N}$ and the set of density operators on them by $\S(\H_N)$. Each mode has a pair of quadrature (position and momentum) operators~\cite{Gerry2004}.

For a system of $N$ modes, let the vector of the $N$ quadrature operator pairs be $\mathbf{\hat{r}}=(\hat{x}_1,\ldots,\hat{x}_N,\hat{p}_1,\ldots,\hat{p}_N)^{\rm T}$.
The quadrature operators satisfy the canonical commutation relations,
\begin{equation} \label{eq:ccr}
    \left[\mathbf {\hat{r}} , \mathbf {\hat{r}}^{\rm T} \right] = \frac{i}{2}
    \begin{pmatrix} 
        0 & \id_N \\ 
        -\id_N & 0
    \end{pmatrix}
    \eqqcolon \frac{i}{2} \Omega
    \,,
\end{equation}
where $\Omega$ is the symplectic form. Note that we use the convention $\hbar=1/2$ in this work.
The annihilation and creation operators are
$\hat{a}_j=\hat{x}_j+i\hat{p}_j$ and $\hat{a}_j^\dagger=\hat{x}_j-i\hat{p}_j$ for $j=1,\dots,N$.
The photon number operators are $\hat{a}_j^\dagger \hat{a}_j$.

A bosonic Hamiltonian $H_G$ is called \emph{Gaussian} if it is a polynomial in the creation and annihilation operators of degree at most 2. \emph{Gaussian states} $\rho_G$ are the ground and thermal states of positive definite Gaussian Hamiltonians~\cite{BUCCO},
\begin{align}
    \rho_G= \frac{e^{-\beta H_G}}{\Tr{e^{-\beta H_G}}}.
\end{align}
Gaussian states are fully characterized by their first two moments, i.e. their mean displacement $\bmd$ and covariance $ V$, defined as
\begin{equation}\label{eq:moments}
 d_j \coloneqq  \expval{\hat r_j} \,,\qquad
 V_{jk} \coloneqq 2\left\langle
        \{\hat r_j-d_j,\hat r_k-d_k\}
      \right\rangle,
\end{equation}
where $\{\cdot,\cdot\}$ denotes the anticommutator and $ V$ satisfies the uncertainty relation $ V + i\Omega \geq 0$.
A unitary $U_G$ is called a \emph{Gaussian unitary} if it can be generated by a time-dependent Gaussian Hamiltonian.

A bosonic quantum channel from $N$ modes to $M$ modes is a completely positive, trace-preserving (CPTP) map
$\mathcal E:\mathcal S(\mathcal H_N)\to\mathcal S(\mathcal H_M)$. We will refer to completely positive, trace-non-increasing (CPTNI) maps as \emph{operations}. 
The subset of channels that take Gaussian states to Gaussian states are the  \emph{Gaussian channels}~\cite{lindblad2000cloning,holevo2001evaluating,Caruso2008-dt}. 
They are specified by a rescaling of the input quadratures $ X$, classical noise that is added to the covariance matrix $ Y$, and a phase-space displacement  $\bm \delta$. 
Define the displacement operators as
$D(\boldsymbol\gamma)
=\exp(-2i\boldsymbol\gamma^{\mathsf T}\Omega\hat{\mathbf r})$,
where $\boldsymbol\gamma\in\mathbb R^{2N}$.
The adjoint of a Gaussian channel maps displacement operators to
scalar multiples of displacement operators according to
\begin{equation}\label{eq:dual_channel}
    \mathcal E_G^\dagger\!\left(D_\text{out}(\boldsymbol\gamma)\right)
    =
    D_\text{in}\!\left(-\Omega_\text{in} X^{\mathsf T}\Omega_\text{out}\boldsymbol\gamma\right)
    \exp\!\left(
        -\frac12\boldsymbol\gamma^{\mathsf T}
        \Omega_\text{out}^{\mathsf T}Y\Omega_\text{out}\boldsymbol\gamma
        -2i\boldsymbol\gamma^{\mathsf T}\Omega_\text{out}\boldsymbol\delta
    \right).
\end{equation}
The matrices need to satisfy the condition $Y + i(\Omega_\text{out} - X \Omega_\text{in} X^{\rm T}) \ge 0$ to ensure complete positivity.
Equivalently, a Gaussian channel transforms the displacement vector $\bmd$ and covariance matrix $V$ of the input as  
\begin{equation}\label{eq:xychannel}
    \bmd \mapsto X\bmd + \bm{\delta} \,,\quad V \mapsto X V X^{\rm T} + Y \,.
\end{equation}
For $N=M$, Gaussian unitaries correspond to the case of $Y=0$; $X$ is then a symplectic matrix and we get the same class of unitaries as are generated by Gaussian Hamiltonians.

Gaussian measurements are described by positive operator-valued measures (POVMs) that yield Gaussian outcome distributions for every Gaussian input state~\cite{Weedbrook2012-qu}. Homodyne and heterodyne measurements are the two most common Gaussian measurements. Indeed, any Gaussian measurement can be realized via Gaussian unitaries, ancillary Gaussian states, and homodyne detection~\cite{Weedbrook2012-qu}.
For a single mode, ideal homodyne detection measures the rotated quadrature
\begin{equation}
    \hat x_\theta
    =\hat x\cos\theta+\hat p\sin\theta.
\end{equation}
Its POVM is the projection-valued spectral measure of $\hat x_\theta$. In terms of eigenstates of the rotated quadrature operators, this measure can be
written formally as
\begin{equation}
    M_\theta({\rm d}x)
    =|x;\theta\rangle\langle x;\theta|\,{\rm d}x.
\end{equation}
Introducing the complex outcome
\begin{equation}
    \gamma=y_x+i y_p,
\end{equation}
the heterodyne POVM takes the form
\begin{equation}
    M({\rm d}^2\gamma)
    =\frac{1}{\pi}|\gamma\rangle\langle\gamma|\,
      {\rm d}^2\gamma,
    \qquad
    \frac{1}{\pi}\int_{\mathbb C}
       |\gamma\rangle\langle\gamma|\,{\rm d}^2\gamma=\id ,
\end{equation}
where $\ket{\gamma}$ is a coherent state and ${\rm d}^2\gamma ={\rm d}(\operatorname{Re}\gamma)\,{\rm d}(\operatorname{Im}\gamma)$.
Thus, $\pi^{-1}|\gamma\rangle\langle\gamma|$ is the operator-valued
density of the heterodyne POVM with respect to ${\rm d}^2\gamma$. %

Instead of representing a quantum state as an operator on the Hilbert space $\mathcal{H}^{\otimes N}$, one can equivalently represent it as a real-valued function over a phase space $\mathbbm{R}^{2N}$. Perhaps the most ubiquitous example is the Wigner function~\cite{Wigner1932}. Let $\hat{\Pi}\coloneqq(-1)^{\sum_{i=1}^N \hat{n}_i}$ be the total parity operator and $\hat D(\bm x, \bm p) \coloneqq e^{2i(\bm p \bm{\hat{x}} - \bm x \bm{\hat{p}})}$ displacement operators. Define the displaced and scaled parity operators 
\begin{equation}\label{eq:Delta}
    \hat{\Delta}_{\rm W}(\bm{x},\bm{p})
    \coloneqq
    \left(\frac{2}{\pi}\right)^N \,D(\bm{x},\bm{p})
    \,\hat{\Pi}\,
    D^{\dagger}(\bm{x},\bm{p})
    \,.
\end{equation}
The Wigner function of $\rho$ is then given by the expectation values of displaced parity operators~\cite{Grossmann_1976, royer1977Wigner},
\begin{equation}\label{eq:def_wigner_function}
    W_{\rho}(\bm{x},\bm{p}) \coloneqq \Tr\big[\hat{\Delta}_{\rm W}(\bm{x},\bm{p})\;\rho\big]  %
    \,.
\end{equation}
The mapping is bijective, so that
\begin{equation}
    \rho = \pi^N \int W_{\rho}(\bm{x},\bm{p})  \hat{\Delta}_{\rm W}(\bm{x},\bm{p}) \ {\rm d}^N \bm{x} \ {\rm d}^N  \bm{p}.
\end{equation}
The Wigner function of a quantum state is normalized,
\begin{equation}
    \Tr(\rho) = \int W_{\rho}(\bm{x},\bm{p}) \ {\rm d}^N \bm{x} \ {\rm d}^N  \bm{p}.
\end{equation}
However, it does not always have a finite absolute integral. Our formalism is restricted to the subset of states that do have absolutely integrable Wigner functions, i.e. finite negative volume. These so-called Feichtinger states form a dense subset in trace norm in the space of all quantum states~\cite{gosson}.
The purity can be expressed as
\begin{equation}
  \Tr(\rho^2) = \pi^N \int W_{\rho}(\bm{x},\bm{p})^2 {\rm d}^N \bm{x} \ {\rm d}^N  \bm{p}
  \,.
\end{equation}

Quantum channels, or any CPTNI maps, from $N$ to $M$ modes admit a Wigner kernel,
\begin{equation}\label{wignerkernel}
K_{\Phi} (\bm{x},\bm{p},\bm{x}',\bm{p}') \coloneqq \pi^N \Tr( \hat{\Delta}_{\rm W}(\bm{x},\bm{p}) \Phi(\hat{\Delta}_{\rm W}(\bm{x}',\bm{p}')))
\,.
\end{equation}
This kernel allows us to write the action of the map on a state directly in terms of Wigner functions, 
\begin{equation}
W_{\Phi(\rho)}(\bm{x},\bm{p}) = \int_{\mathbbm{R}^{2N}} K_{\Phi} (\bm{x},\bm{p},\bm{x}',\bm{p}') W_{\rho}(\bm{x}',\bm{p}') {\rm d}^N \bm{x}' \ {\rm d}^N  \bm{p}'
\,.
\end{equation}
We use the shorthand $K_\Phi W_\rho$ to denote the action of the kernel as above.

More generally, one can consider the representation of states and channels in an entire family of $s$-ordered representations~\cite{cahill1969density}, including the Husimi $Q$ function~\cite{husimi1940some} and Sudarshan--Glauber $P$ function~\cite{Glauber_1963,Sudarshan_1963}.

\subsection{Majorization theory}

In this section, we provide the requisite background on majorization theory.  We state all definitions using the Wigner function for clarity of exposition, but we emphasize that the results hold more generally for all real-valued and integrable phase-space representations. The more general statements can be found in \cref{app:majproofs}. Majorization is a preorder on distributions based on how disordered they are~\cite{marshall2011inequalities} that has found many applications in quantum information theory~\cite{nielsen1999conditions,Vidal1999,Torun_Majorization_QRT_Review_2023}.

In the following, we restrict to Wigner functions whose absolute integral converges.
Let $\mathcal{I}$ denote an indicator function and denote the positive and negative parts of $f$ by $f^+\coloneqq\max\{f,0\}$ and $f^-\coloneqq\min\{f,0\}$. Given a Wigner function on $N$ modes, the distribution and codistribution functions of its positive and negative parts are
\begin{align}
    D_{W^+_{\rho}}(t)&\coloneqq \int \mathcal{I}\{W^+_{\rho}(\bm{x},\bm{p}) >t\}  {\rm d}^N \bm{x} \ {\rm d}^N  \bm{p},\\
    C_{W^-_{\rho}}(t)&\coloneqq \int \mathcal{I}\{W^-_{\rho}(\bm{x},\bm{p}) <t\} {\rm d}^N \bm{x} \ {\rm d}^N  \bm{p}.
\end{align}
These equal the area over which the function is greater or less than a certain value.
We can then define the decreasing and increasing rearrangements of a Wigner function as
\begin{align}
    {W^+_{\rho}}^\downarrow(u) &\coloneqq \inf\{ t: D_{W^+_{\rho}}(t)\leq u \} \,, \\
    {W^-_{\rho}}^\uparrow(u) &\coloneqq \sup\{ t: C_{W^-_{\rho}}(t)\leq u \} \,.
\end{align}
Informally these are the reorderings of the function's values in decreasing or increasing order arranged on a 1D line indexed by the accumulated volume, regardless of the original space the function was defined on.
Finally, from the rearrangements we define the positive and negative Lorenz curves,
\begin{align}
    L^{\Dsh}_{W_{\rho}}(s)\coloneqq \int_0^s (W^+_{\rho})^\downarrow(u) {\rm d}u \, , \\ 
    L^{\Ish}_{W_{\rho}}(s)\coloneqq \int_0^s (W^-_{\rho})^\uparrow(u) {\rm d}u \, ,
\end{align}
where the argument $s$ is in $[0,\infty)$. The shape of the Lorenz curves captures how quickly a function accumulates its positive and negative values. Two functions that take the same values just at different points in phase space have identical Lorenz curves. For example, the Wigner Lorenz curves of the vacuum and of any coherent state are the same. In \cref{fig:illus}(a), we plot the Lorenz curves of $W_{\ket{0}}$, $W_{\ket{1}}$, and $W_{\ket{\gamma=1.2}_{\text{cat}}}$.

\begin{figure}[h!]
    \centering
    \includegraphics[width=\linewidth]{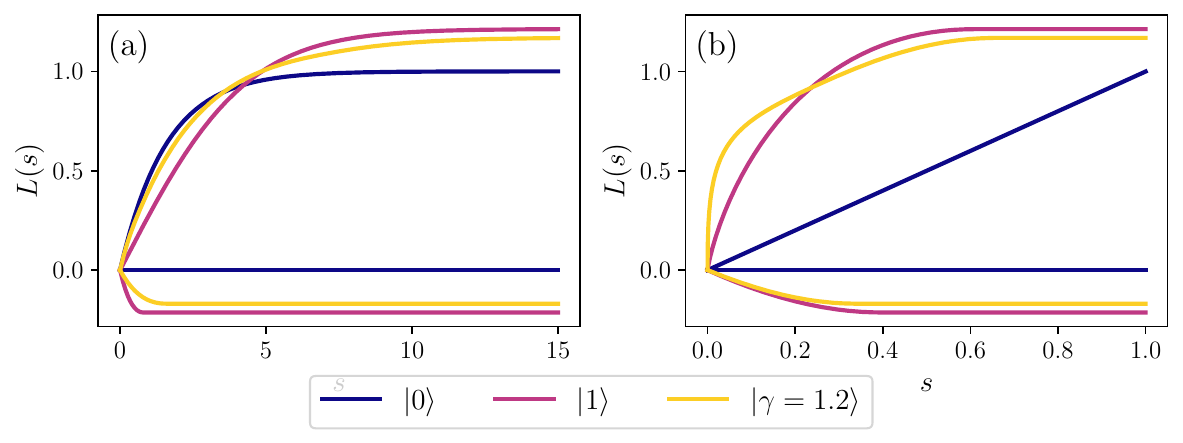}
\captionof{figure}{Wigner Lorenz curves of vacuum, the first Fock state, and a cat state for (a) regular majorization and (b) majorization relative to the Wigner function of vacuum.}
\label{fig:illus}
\end{figure}

Linear maps from functions, Wigner or otherwise, on $\mathbbm{R}^{2N}$ to those on $\mathbbm{R}^{2M}$ can be represented by integration kernels $S$. An integral kernel is positive if
\begin{equation}\label{def:positive}
    S(\bm{x},\bm{p},\bm{x}',\bm{p}') \geq 0
\end{equation}
for all points $(\bm{x},\bm{p},\bm{x}',\bm{p}')$. A positive integral kernel $S$ is substochastic if 
\begin{equation}\label{def:substochastic}
    \int_{\mathbbm{R}^{2M}}  S(\bm{x},\bm{p},\bm{x}',\bm{p}') {\rm d}^M \bm{x} \ {\rm d}^M  \bm{p} \leq 1
\end{equation}
for almost all $(\bm{x}',\bm{p}')$, and stochastic if the same holds with equality. A substochastic integral kernel is doubly substochastic if 
\begin{equation}\label{def:dblsubstochastic}
    \int_{\mathbbm{R}^{2N}}  S(\bm{x},\bm{p},\bm{x}',\bm{p}') {\rm d}^N \bm{x}' \ {\rm d}^N  \bm{p}' \leq 1
\end{equation}
for almost all $(\bm{x},\bm{p})$. All Gaussian channels~\cite{Cerf2007-chap2, Caruso2008-dt,Weedbrook2012-qu} have stochastic Wigner kernels, of the form
\begin{align}
K_{\Phi_G}(\bm x, \bm p,\bm x', \bm p') &= 
\frac{2^M}{\pi^M \sqrt{\det  Y}}\   
\exp\left[
-2\Big(
\begin{mypsmallmatrix}
    \bm{x}\\\bm{p}
\end{mypsmallmatrix}
-{X}
\begin{mypsmallmatrix}
\bm{x}'\\\bm{p}'
\end{mypsmallmatrix}
-\bm{\delta}
\right)^\intercal 
{Y}^{-1}
\left(
\begin{mypsmallmatrix}
    \bm{x}\\\bm{p}
\end{mypsmallmatrix}
-
{X}
\begin{mypsmallmatrix}
    \bm{x}'\\\bm{p}'
\end{mypsmallmatrix}
-\bm{\delta}
\Big)
\right]
\label{eq:gaussian_wigner_kernel}.
\end{align}
For $N=M$, the kernel is doubly substochastic when $\abs{\det X}\geq 1$.

The majorization setting we have thus described is that of regular majorization: distributions that are more spread out, approaching a uniform distribution over phase space, rank lower in the preorder. 
In particular, the comparison of distributions with different normalizations is referred to as \emph{weak} majorization~\cite{marshall2011inequalities}.
A more general majorization setting is that of \emph{relative} majorization. Here, distributions are compared relative to an arbitrary strictly positive reference distribution $q$, which need not be a physical Wigner function. Relative majorization with $q=1$ is equivalent to regular majorization. We now repeat the above definitions for weak, relative majorization. 

For an $N$-mode Wigner function $W_\rho$ and reference distribution $q$, the relative distribution and codistribution functions of its positive and negative parts are
\begin{align}
    D^q_{W^+_{\rho}}(t)&\coloneqq \int \mathcal{I}\{W^+_{\rho}(\bm{x},\bm{p}) >t q(\bm{x},\bm{p}) \} q(\bm{x},\bm{p}) {\rm d}^N \bm{x} \ {\rm d}^N  \bm{p},\\
    C^q_{W^-_{\rho}}(t)&\coloneqq \int \mathcal{I}\{W^-_{\rho}(\bm{x},\bm{p}) <t q(\bm{x},\bm{p})\} q(\bm{x},\bm{p})  {\rm d}^N \bm{x} \ {\rm d}^N  \bm{p}.
\end{align}
The relative decreasing and increasing rearrangements of a Wigner function are then
\begin{align}
    {W^+_{\rho}}^{\downarrow,q}(u) &\coloneqq \inf\{ t: D^q_{W^+_{\rho}}(t)\leq u \} \,, \\
    {W^-_{\rho}}^{\uparrow,q}(u) &\coloneqq \sup\{ t: C^q_{W^-_{\rho}}(t)\leq u \} \,.
\end{align}
Finally, from the rearrangements we define the positive and negative Lorenz curves,
\begin{align}
    L^{\Dsh}_{W_{\rho}|q}(s)\coloneqq \int_0^s (W^+_{\rho})^{\downarrow,q}(u) {\rm d}u \, , \\ 
    L^{\Ish}_{W_{\rho}|q}(s)\coloneqq \int_0^s (W^-_{\rho})^{\uparrow,q}(u) {\rm d}u \, .
\end{align}
$s$ is in the range $[0,\int q(\bm{x},\bm{p}) {\rm d}^N \bm{x} \ {\rm d}^N  \bm{p}]$, to be read as $[0,\infty)$ if the integral diverges. %
In \cref{fig:illus}(b), we plot the Wigner Lorenz curves of the same three states as above but now relative to $q=W_{\ket{0}}$.

A substochastic kernel $S$ is $(q,q')$-substochastic if $(Sq)(\bm{x},\bm{p}) \leq q'(\bm{x},\bm{p})$. %
It is useful to note that Lorenz curves obey the following rescaling relation for $c,p\in(0,1]$,
\begin{align}\label{eq:rescale}
    \qdeclor{W_{\rho}}&=\frac{1}{p} \qqqdeclor{pW_{\rho}}{cq}{cs}\\
    \qinclor{W_{\rho}}&=\frac{1}{p} \qqqinclor{pW_{\rho}}{cq}{cs}.
\end{align}
For subnormalized states $\rho$ and any nonnegative convex function $\phi:\mathbbm{R} \rightarrow \mathbbm{R}_{\geq 0}$ satisfying $\phi(0)=0$, we define the functional 
\begin{equation}
    D(W_\rho\|q)= \int_X q \phi\left(\frac{W_\rho}{q}\right) {\rm d}\mu.
\end{equation}
We refer to such a functional as a generalized divergence. For example, up to additional monotonic transformations on the integral, this construction recovers the R\'enyi entropies for regular majorization and R\'enyi divergences for relative majorization. These convex averages and any other functionals that respect the majorization preorder are referred to as Schur-convex or Schur-concave.

For the consideration of probabilistic conversions, it will be useful to compare Lorenz curves of states between spaces with different total measure. In this case, for $s$ exceeding the total measure, we define the Lorenz curve by its endpoint value, i.e. we extend the positive Lorenz curve by a horizontal plateau at its highest value and the negative one by a horizontal plateau at its lowest value.

\section{Constraints on heralded state conversion from majorization}
\label{sec:majorization}

In this section, we provide our main results on majorization and discuss how its different formulations result in state conversion no-go results, distillation bounds, and tradeoff relations.
The key technical contributions are twofold:
Firstly, we provide equivalent formulations of relative majorization on quasiprobabilities, extending the main result of~\cite{Upadhyaya2026} in three ways: distributions are permitted to (i) be defined on different measure spaces, (ii) have different normalizations, and (iii) be compared relative to different references.
Secondly, we extend the formulation of majorization to the approximate setting.
Therefore, our framework captures heralded, probabilistic exact and approximate conversions.
We present our results via the Wigner representation for concreteness, but they hold for all real-valued and integrable phase space representations. This includes other phase-space representations for bosonic states, but also those for qudits. While we focus on bosonic phase spaces, our framework applies in more general settings, including discrete phase spaces for qudit systems.

\subsection{Bounds on exact state conversions}

Exact conversions demand that the output state matches the target state, albeit with some probability of success.
There are two notions of such probabilistic conversion protocols: A \emph{heralded} conversion is one where there is a classical flag when the conversion succeeds; the map is trace non-increasing and outputs $p\rho'$ on success. An \emph{unheralded conversion} has no such flag; the map is trace-preserving and outputs $p'\rho'+(1-p')\sigma$ for some $\sigma$. In the setting of ~\cite{Renes2016}, weak majorization of eigenvalue distributions provides the same upper bound on $p$ and $p'$. Due to the negativity of quasiprobability distributions, this is not the case in our setting. Our majorization bounds are upper bounds only on success probabilities of heralded conversions; unheralded conversions may be achievable with higher probability.

We present a version of our majorization theorem geared towards heralded, probabilistic photonic state conversions.
\begin{theorem}\label{physmajor}
    Let $\rho \in \S(\H_N)$ and $\rho' \in \S(\H_M)$ be two quantum states with finite Wigner negative volume and let $q, q'$ be two strictly positive functions over the phase spaces of $N$ and $M$ modes, respectively.
    Denote by $\O$ the set of completely positive trace-non-increasing maps that are Wigner $(q,q')$-substochastic (see \cref{sec:recipes} for examples).
    The heralded conversion of $\rho$ to $\rho'$ with probability $p$ and using operations from $\O$, is impossible if either of the following conditions holds:
    \begin{enumerate}
        \item There exists $s>0$ such that $\qdeclor{W_\rho} < p \qpdeclor{W_{\rho'}} $ or $\qinclor{W_\rho} > p\qpinclor{W_{\rho'}}$;
        \item There exists a generalized divergence $D$ such that $D(W_\rho\|q) < D(pW_{\rho'}\|q')$.
    \end{enumerate}
\end{theorem}
If neither of the conditions in \cref{physmajor} holds, then $W_\rho$ is said to (weakly) majorize $W_{p\rho'} = pW_{\rho'}$ relative to $(q,q')$, denoted by $(W_\rho,q) \weak (W_{p \rho'},q')$. 
Note that weak relative majorization does not guarantee that a conversion is necessarily possible with a physical quantum operation.

In \cref{app:majproofs}, we establish the complete theorem from which this result derives. 
Namely, we prove the equivalence of four conditions for weak relative majorization between quasiprobabilities on different measure spaces and with different strictly positive reference distributions.
This is an extension of the relative majorization result of Ref.~\cite{Upadhyaya2026} which is restricted to distributions with the same normalization, defined on the same measure spaces, and compared with respect to the same reference distribution.

Our definition of weak majorization for quasiprobabilities recovers the standard notion of weak submajorization for nonnegative distributions~\cite{marshall2011inequalities}, but is distinct from both submajorization and supermajorization for distributions with negativity. An advantage of our definition over others is that it admits majorization equivalences that parallel those in simpler settings. We discuss this in more detail in \cref{app:majproofs}.

The second statement of \cref{physmajor} allows us to obtain entropic constraints on state transformations.
Concretely, we consider a family of relative entropies in bosonic phase space based on the classical $\alpha$-R\'enyi divergences.
\begin{definition}\label{def:WignerRenyi}
    Given two possibly subnormalized quantum states $\rho, \tau$ with Wigner functions $W_\rho, W_\tau$ such that $W_{\tau}$ is strictly positive on the entire phase space, we define the \emph{Wigner $\alpha$-R\'enyi divergence} of $\rho$ relative to $\tau$ as
    \begin{equation}
        D_\alpha(W_{\rho} \| W_{\tau}) \coloneqq \frac{1}{\alpha-1} \log \int_{\mathbb{R}^{2N}} \abs{W_\rho(\bm{r})}^\alpha W_\tau(\bm{r})^{1-\alpha} \; {\rm d}^{2N}\bm{r} \,, \quad \alpha \in (1, \infty) \,.
    \end{equation}
\end{definition}
Here, $W_\tau$ plays the role of the reference distribution $q$ and the restriction to $\alpha > 1$ is needed for quasiprobabilities~\cite{Upadhyaya2026} to ensure Schur-convexity.
The Wigner $\alpha$-R\'enyi divergence inherits several desirable properties from the classical R\'enyi divergence on probability distributions~\cite{koukoulekidis2022constraints,koukoulekidis2023quasi}:
For normalized $\rho$ and $\tau$, it is nonnegative, equaling zero if and only if $\rho = \tau$.
It is also additive even for subnormalized states,
\begin{equation}
    D_\alpha(W_{\rho_1 \otimes \rho_2} \| W_{\tau_1 \otimes \tau_2}) = D_\alpha(W_{\rho_1} \| W_{\tau_1}) + D_\alpha(W_{\rho_2} \| W_{\tau_2}) \,.
\end{equation}
In particular, $D_\alpha(W_{\rho^{\otimes n}} \| W_{\tau^{\otimes n}}) = nD_\alpha(W_{\rho} \| W_{\tau})$. Finally, for any scalars $p,t \in (0,1]$, we have 
\begin{equation}\label{renyiscale}
    D_\alpha(W_{p\rho} \| W_{\tau})
    = D_\alpha(W_{\rho} \| W_{\tau}) + \frac{\alpha}{\alpha-1} \log p - \log t
    \,,
\end{equation}
leading to the following bound.
\begin{corollary}\label{prop:div}
    Let $\E$ be a Wigner-substochastic operation and $\tau$ be a strictly Wigner-positive state such that $\E(\tau)$ is also strictly Wigner-positive.
    Then, for an arbitrary quantum state $\rho$ and scalar $\alpha > 1$, the Wigner $\alpha$-R\'enyi divergence satisfies the data-processing inequality
    \begin{equation}
        D_\alpha(W_{\rho} \| W_{\tau}) \ge D_\alpha(W_{\E(\rho)} \| W_{\E(\tau)}) \,.
    \end{equation}
    In particular, if $\E(\rho) = p\rho'$ and $\E(\tau) = t\tau'$ for some probabilities $p,t$, then
    \begin{equation}\label{corollaryprobbound}
        p^\alpha t^{1-\alpha} \le \exp( (\alpha-1) \left[ D_\alpha(W_{\rho} \| W_{\tau}) - D_\alpha(W_{\rho'} \| W_{\tau'}) \right] )
        \,.
    \end{equation}
\end{corollary}
\begin{proof}
    The divergence as defined is a Schur-convex function under majorization relative to $(q,q')=(W_{\tau},W_{\E(\tau)})$. 
    Hence, the first inequality follows immediately from the extension of \cref{physmajor} in \cref{app:majproofs}, and the second inequality follows from the first together with \cref{renyiscale}.
\end{proof}
\cref{corollaryprobbound} gives an upper bound on the probability of heralded conversion of $\rho$ to $\rho'$ using $\E$. If we are interested in a class of operations which all have the same action on $\tau$, then \cref{corollaryprobbound} bounds the conversion probability under any of those operations via a single condition.

\subsection{Bounds on approximate state conversions}

In practice, especially in bosonic settings, state conversion protocols do not often output an exact target state, even probabilistically. Instead, they generally output states that are close to the target in trace distance or fidelity. 
Namely, given an input state $\rho$ and a target state $\rho'$, the free operations may output $\E(\rho) = p \tilde{\rho}$ where the state $\tilde{\rho}$ is close to the target, $\frac{1}{2}\|\tilde{\rho}-\rho'\|_1 \le \epsilon$.
In principle, the probability $p$ of exact conversion to a target state may be small or zero, while the probability of generating a state in the neighborhood of the target may be large.

From \cref{corollaryprobbound}, as long as $D_\alpha(W_{\tilde{\rho}} \| W_{\E(\tau)})$ does not get much smaller than $D_\alpha(W_{\rho'} \| W_{\E(\tau)})$ for all $\tilde{\rho}$ close to $\rho'$, the exact conversion bound can be relaxed to approximate conversions. 
We now show that this property, termed lower semicontinuity, does in fact hold for the Wigner--R\'enyi divergences.
The full proof is deferred to \cref{app:wigner}.
\begin{theorem}\label{thm:semicontinuity}
    Let $\tau \in \S(\H_N)$ be a reference state such that $W_\tau > 0$ and $\alpha > 1$ be a scalar.
    Then, the following two statements hold.
    \begin{enumerate}
    \item[(i)] The functional $\rho \mapsto D_\alpha(W_\rho \| W_\tau)$ is lower semicontinuous on $\S(\H_N)$ in the trace-norm topology.
    In particular, for any sequence of states $(\rho_n)_{n \in \mathbb{N}} \subset \S(\H_N)$ converging to $\rho \in \S(\H_N)$ in trace norm, i.e.
    \begin{equation}
        \lim_{n \to \infty} \frac{1}{2} \|\rho_n - \rho\|_1 = 0
        \,,
    \end{equation}
    we have
    \begin{equation}
        \liminf_{n \to \infty} D_\alpha(W_{\rho_n} \| W_\tau) \ge D_\alpha(W_{\rho} \| W_\tau)
        \,.
    \end{equation}
    \item[(ii)] Assume that $\tau$ is a Gaussian state with displacement vector $\bmd$ and covariance matrix $V$, and $\rho$ satisfies the finite tail-margin assumption 
    \begin{equation}\label{eq:Dkappa_assumption}
        D_{\alpha}(W_{\rho} \| q_\kappa) < \infty
        \,, \qquad
        q_\kappa(\bm{r}) \coloneqq \exp(-2\kappa \|\bm{r} - \bmd\|_2^2) \; W_\tau(\bm{r})
        \,,
    \end{equation}
    for some $\kappa > 0$.
    Then, every quantum state $\tilde{\rho} \in \S(\H_N)$ satisfying
    $\frac{1}{2} \|\tilde{\rho}-\rho\|_1 \le \epsilon$ obeys
    \begin{equation}\label{eq:semicontinuity}
        D_\alpha(W_{\tilde{\rho}}\|W_\tau)
        \ge
        D_\alpha(W_{\rho}\|W_\tau) - \frac{\log(1/(1-\Delta\epsilon^c))}{\alpha-1}
        \,, \qquad
        \forall \epsilon \in \left[ 0, \Delta^{-\frac{1}{c}} \right)
        \,,
    \end{equation}
    for constants $\Delta > 0$ and
    \begin{equation}
        c \coloneqq \frac{\kappa \alpha}{\kappa \alpha + \|V^{-1}\|_{\infty}}
        \in (0, 1)
        \,,
    \end{equation}
    that are both independent of $\tilde{\rho}$ and $\epsilon$.
    
    In particular, for all $E \in (0,1)$, the functional $\rho \mapsto D_\alpha(W_\rho \| W_\tau)$ obeys the one-sided H\"older continuity bound
    \begin{equation}\label{eq:semicontinuity2}
        D_\alpha(W_{\tilde{\rho}}\|W_\tau)
        \ge
        D_\alpha(W_{\rho}\|W_\tau) - \frac{\log(1/(1-E))}{(\alpha-1)E}\Delta \epsilon^c
        \,, \qquad
        \forall \epsilon \in \left[ 0, \left(E\Delta^{-1}\right)^{\frac{1}{c}} \right]
        \,.
    \end{equation}
    \end{enumerate}
\end{theorem}
To promote the qualitative lower semicontinuity of part \emph{(i)} to a quantifiable one-sided H\"older continuity bound in part \emph{(ii)}, we have made two additional assumptions.
The first assumption, that the reference state $\tau$ be Gaussian, is not strictly necessary and a quantitative bound can be derived for any $\tau$ with $W_\tau > 0$.
However, the Gaussianity assumption allows us to get closed forms for $\Delta$ and $c$ that depend on the covariance matrix $V$; the explicit expression for $\Delta$ can be found in \cref{app:wigner}.
The parameter $E$ is controllable, and it can be used to trade off a smaller error prefactor against a larger valid range for the error.

The second assumption of \cref{eq:Dkappa_assumption} is a mild strengthening of the assumption that $D_{\alpha}(W_{\rho} \| W_\tau) < \infty$.%
As a concrete example,  in \cref{sec:cat}, we study noisy cat state conversions, and we explicitly calculate the range of $\kappa$ for which the assumption holds in terms of the relevant physical parameters.
As $\kappa \to 0^+$, the prefactor $\Delta \to 1$ remains finite, but the exponent $c \to 0^+$ and the maximum valid error rate $\Delta^{-\frac{1}{c}} \to 0^+$ both vanish, recovering the data-processing inequality for exact state conversions.

The qualitative observation that R\'enyi divergences, and more general types of divergences~\cite{csiszar1967information}, are lower semicontinuous on arbitrary measure spaces is well-known (e.g.~\cite{csiszar1967information,liese2006on,van2014renyi,CaiVerdu2019}), and our result does not imply full continuity.
The constants $\Delta, c$ in \cref{thm:semicontinuity} depend on the state $\rho$ that satisfies the assumption $D_{\alpha}(W_{\rho} \| q_\kappa) < \infty$ for some $\kappa > 0$.
Therefore, assuming that a universal $\kappa > 0$ existed, such that $D_{\alpha}(W_{\tilde{\rho}} \| q_\kappa) < \infty$ also holds for all $\tilde{\rho}$ that are $\epsilon$-close to $\rho$ in trace norm, one would recover H\"older continuity.
However, in \cref{app:wigner} we show that the functional $\rho \mapsto D_\alpha(W_{\rho} \| W_\tau)$ is discontinuous by constructing a sequence of states with bounded moments that approach $\rho$ but have divergences growing arbitrarily large.

The lower semicontinuity of various quantum divergences and entanglement monotones defined in the Hilbert space has been studied (e.g.~\cite{winter2016tight,berta2017variational,ferrari2023asymptotic,mosonyi2023strong,van2025energy}).
Interestingly, it is known that uniform continuity of quantum entropies can be recovered under energy constraints~\cite{winter2016tight}, which contrasts our discontinuity result for the Wigner $\alpha$-R\'enyi divergence under bounded moments.
Moreover, the classical $\alpha$-R\'enyi entropy on discrete, finite-dimensional Wigner functions is known to be continuous~\cite{koukoulekidis2022constraints}.
Using the additional property of Schur-convexity, Ref.~\cite{alexander2023general} uses $\alpha$-R\'enyi divergences to obtain tradeoff relations for approximate conversions in the context of qudit magic state distillation.
As expected, this continuity bound diverges with the local qudit dimension.

\cref{thm:semicontinuity} extends the upper bound of \cref{prop:div} to the setting of approximate heralded conversions.
In particular, consider the transformation $\E(\rho) = p\tilde{\rho}$ where the output state $\tilde{\rho}$ is $\epsilon$-close to the target $\rho'$ in trace distance.
Under the assumptions of \cref{thm:semicontinuity} with $\E(\tau) = t\tau'$ for some probability $t$, the tradeoff relation becomes
\begin{equation}\label{corollaryprobbound_extended}
    p^\alpha t^{1-\alpha} (1-\Delta\epsilon^c) \le \exp( (\alpha-1) \left( D_\alpha(W_{\rho} \| W_{\tau}) - D_\alpha(W_{\rho'} \| W_{\tau'}) \right) )
    \,.
\end{equation}

We note that one can also obtain bounds on approximate conversions using the full Lorenz curves. Suppose $\frac{1}{2} \| \tilde{\rho} - \rho \|_1 \le \epsilon$. 
Then, 
as, e.g. shown in \cref{app:wigner}, the Wigner function obeys $\|W_\rho(\bm r) - W_{\tilde{\rho}}(\bm r)\|\le 2\left(\frac{2}{\pi}\right)^N \epsilon$. 
Consequently, the positive and negative Lorenz curves of $(W_\rho-2\left(\frac{2}{\pi}\right)^N \epsilon)^+$ and $(W_\rho+2\left(\frac{2}{\pi}\right)^N \epsilon)^-$ provide, in the language of \cref{sec:robustcurves}, inner bounds on the Lorenz curves of any state in the trace-norm-ball. 
For regular majorization, the Lorenz curve bounds can be simply expressed:
the positive curve's bound is $\declor{W_\rho}-2s\left(\frac{2}{\pi}\right)^N \epsilon $ until this bound reaches a slope of zero, and a horizontal plateau afterwards;
the negative curve is similarly piecewise equal to $\inclor{W_\rho}+2s\left(\frac{2}{\pi}\right)^N \epsilon $ and a horizontal plateau once zero slope is reached.

\subsection{Majorization recipes: specifying the allowed operations}
\label{sec:recipes}

Our majorization framework allows for the study of conversions between arbitrary, pure or mixed, input and output states. 
One of the strengths of our framework is its flexibility: 
Aside from the ability to choose different phase-space representations, we can also study a broad range of state-conversion scenarios by choosing different reference distributions $q,q'$.
These specify the class of allowed operations to be $\O$, i.e. operations with Wigner $(q,q')$-substochastic kernels.
Unlike some resource monotones which are specific to a particular resource theory, our majorization framework is able to accommodate different scenarios and applications. It does so through one common lens: that of phase-space disorder.
One of the challenges in applying this framework is that it is difficult to enumerate all the operations in $\O$.
For example, it is difficult to classify all channels with a positive Wigner kernel; in fact, it is hard to characterize even the states that are Wigner-positive.
However, in a practical state conversion setting, we have a given set of physically realizable operations, and for constraining state conversions via majorization, it suffices to show that this set of physical operations is included in $\O$. If other operations are also included in $\O$, the no-goes are still valid, but their quantitative statements may be weakened.
Nevertheless, understanding which operations are disallowed in a specific majorization setting is also useful as it can guide the construction of conversion protocols.
In this section, we summarize several relevant ``recipes'' for applying majorization with Wigner functions, i.e. choices of $(q,q')$. We give classes of channels that are and are not in the different $\O$.

We start with a short discussion on the partial trace operation, as understanding its interaction with the majorization preorder is important to understand the following recipes.
Canonically, partial trace is treated as a free operation in most resource theories. This is operationally natural, since discarding a subsystem is usually easy.
Under our framework, when $q$ is integrable, partial trace can be made an allowed operation for an appropriate choice of $q'$. However, if $q$ is not integrable, it is not generally possible to include partial trace in the set of allowed operations.

We give a simple example to show that partial trace can be regarded in certain contexts as resourceful. Consider a Stinespring dilation of a replacement channel, consisting simply of SWAP followed by partial trace. If the partial trace operation is replaced by conditioning on a non-identity POVM element, it is no longer possible to deterministically achieve replacement on every state. Hence, partial trace enables an extra capability.
Partial trace is also resourceful from a regular-majorization perspective because it can reduce the disorder in a distribution, concentrating mass back onto a single value. This effect can be arbitrarily large as Wigner functions can be arbitrarily spread out.

\subsubsection{Measurement-conditioned Gaussian dilations}\label{recipegaussian}

Any Gaussian channel with $M=N$ can be realized by a Gaussian unitary dilation acting on the original $N$-mode state and $2N$ copies of vacuum, followed by partial trace~\cite{Caruso2008-dt}. It is challenging to find a single choice of $(q,q')$ such that all these channels are $(q,q')$-substochastic. For example, the choice $q=1$ is invariant under any initial Gaussian unitary. However, it then becomes arbitrarily concentrated under partial trace, as discussed above. 

Instead, we can consider restricted classes of channels consisting of the dilating Gaussian unitary followed by conditioning on a POVM element with finite trace, as depicted in \cref{fig:gdil}. The Gaussian unitary has a doubly stochastic Wigner kernel, acting as a symplectic coordinate change on the input Wigner function. Thus, we need only analyze the measurement.
\begin{figure}
    \centering
\begin{quantikz}[row sep=0.5cm, column sep=0.6cm]
    \lstick{$\rho$}
    & \gate[3]{U_G}
      \gategroup[wires=3,steps=2,
        style={dashed,rounded corners,inner sep=6pt}]{}
    & \qw & \qw \rstick{$\Phi(\rho)$} \\
    \lstick{$\ket{0}^{\otimes N}$}
    & & \meter[2]{P} & \setwiretype{n} \\
    \lstick{$\ket{0}^{\otimes N}$}
    & & & \setwiretype{n}
\end{quantikz}
\caption{Schematic of measurement-conditioned Gaussian dilation. $U_G$ is an arbitrary Gaussian unitary, while the conditioning POVM elements $P$ and $\id-P$ must be Wigner-positive. }
    \label{fig:gdil}
\end{figure}
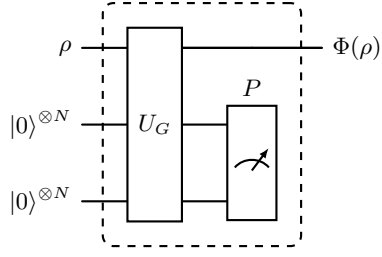

By conditioning on a measurement outcome, we mean that the measurement with POVM $\{P,\id-P\}$ is performed and we proceed only if the first outcome is obtained. $P$ is the postselected POVM element on the $k$ modes that are measured, modeling tensor products of single-mode measurements or multimode measurements. Partition the coordinates for input modes depending on whether they are measured or not, with primed coordinates denoting measured modes. From \cref{wignerkernel}, the integral kernel for the measurement is %
\begin{equation}\label{msmtkernel}
    K_\text{meas}(\bm x_\text{out},\bm p_\text{out};\bm x_\text{in},\bm p_\text{in},\bm x'_\text{in},\bm p'_\text{in})
    =
    \pi^k W_P(\bm x'_\text{in},\bm p'_\text{in}) \delta(\bm x_\text{out}-\bm x_\text{in}) \delta(\bm p_\text{out} - \bm p_\text{in}).
\end{equation}
That is, the Wigner function of the output is given by
\begin{equation}
    W_{\Phi(\rho)}(\bm x,\bm p) 
    = 
    \pi^k \int_{\mathbb{R}^{2k}}  W_P(\bm x',\bm p') W_{U_G (\rho\otimes \dyad{0}^{\otimes 2N}) U_G^\dagger}(\bm x,\bm p,\bm x', \bm p') \; {\rm d}^{k}\bm x' {\rm d}^{k}\bm p'.
\end{equation}
The measurement kernel of the full CPTNI map is then positive if $P$ has a positive Wigner function. 

Although measuring and conditioning on a specific outcome is a CPTNI map, this does not imply that the kernel is substochastic; there may be unphysical test functions whose $L^1$-norm is actually increased by the map. This is in contrast to trace-preserving maps, whose kernels, if positive, are necessarily stochastic~\cite{Upadhyaya2026}. We check the substochasticity condition explicitly starting from \cref{msmtkernel},  
\begin{align}
    \int_{\mathbb{R}^{2N}} K_\text{meas}(\bm x_\text{out},\bm p_\text{out};\bm x_\text{in},\bm p_\text{in},\bm x'_\text{in},\bm p'_\text{in}) \; {\rm d}^{N} \bm x_\text{out} {\rm d^{N}} \bm p_\text{out}
    &= \pi^k W_P(\bm x'_\text{in},\bm p'_\text{in}).
\end{align}
Hence,  the kernel is substochastic if and only if
\begin{equation}\label{msmtstoch}
     W_P(\bm x'_\text{in},\bm p'_\text{in}) \leq \frac{1}{\pi^k} \quad  \forall (\bm x'_\text{in},\bm p'_\text{in}) \in \mathbb{R}^{2k}
     \,.
\end{equation}
This condition can be formulated more intuitively. In our convention, the Wigner function of the identity operator on $k$ modes is $1/\pi^k$. So, by linearity, the conditioning POVM element and its complement satisfy
\begin{equation}
    W_P(\bm x'_\text{in},\bm p'_\text{in}) + W_{\id-P}(\bm x'_\text{in},\bm p'_\text{in}) = \frac{1}{\pi^k}.
\end{equation}
\cref{msmtstoch} is thus equivalent to 
\begin{equation}\label{wigstoch}
    W_{\id-P}(\bm x'_\text{in},\bm p'_\text{in}) \geq 0 \quad  \forall (\bm x'_\text{in},\bm p'_\text{in}) \in \mathbb{R}^{2k}
     \,.
\end{equation}
Thus, the kernel is Wigner substochastic if and only if both POVM elements, $P$ and $\id-P$, are Wigner positive. For example, projecting on vacuum, $P=\ketbra{0}$, is not substochastic, while projecting with a normalization penalty, $P= \ketbra{0}/2$, is. Thus, even the POVM element corresponding to discarding cannot ``hide'' Wigner negativity.

Now, given that the kernel is substochastic, we seek reference distributions $q,q'$ so that the kernel is $(q,q')$-substochastic. We set $q=1$ so that it is a fixed point of the arbitrary Gaussian unitary that is first applied. We can then calculate how this $q$ transforms under the measurement. We have that 
\begin{align}
    K_\text{meas} q &= \pi^k \int_{\mathbb{R}^{2k}} W_P(\bm x',\bm p') \; {\rm d}^{k}\bm x' {\rm d}^{k}\bm p' \\
    &= \pi^k \Tr(P) \label{wigqqs}
\end{align}
So we set $q'$ to also be a uniform distribution, equal everywhere to $\pi^k \Tr(P)$. Thus we need the POVM element to have a finite trace to treat this class of CPTNI maps via majorization. 

For example, performing a heterodyne measurement and postselecting on a region of finite area is ($q,q'$)-substochastic with $q'$ set by the region's area. Decreasing the area generally reduces the maximum conversion probability. On the other hand, conditioning on the click/no-click outcome of a photodetector is not an allowed operation, as some POVM elements are Wigner-negative. Conditioning on a homodyne window is Wigner substochastic (Eq.~\eqref{wigstoch}) but is not $(q,q')$-substochastic (Eq.~\eqref{wigqqs}) because the conditioning POVM element does not have a finite trace. For different choices of $q$ and $q'$ however, such operations are allowed (see e.g. the following recipe).

The rescaling $q' \rightarrow pq'$ for $p<1$ compresses the positive and negative Lorenz curves horizontally (\cref{eq:rescale}). This makes the target Lorenz curves steeper so the success probability upper bound decreases. This is expected as reducing the allowed measurement trace restricts the class of operations. It also corresponds to the fact that POVM elements with smaller trace will in general reduce the conditioning probability.

\subsubsection{Passive linear optical networks}\label{recipelon}

An $N$-mode passive linear optical network (LON) implements a unitary that is composed of layers of passive, i.e. number-preserving, optical elements,
\begin{equation}\label{eq:LON}
    \U = \U_L \circ \dots \circ \U_1 \in {\rm LON} \,,
\end{equation}
where each $\U_i$ is a unitary corresponding to the $i^{th}$ layer.
These optical elements are generated by phase shifters and beam splitters. An arbitrary number of vacuum modes can also be appended.

We consider an $N$-mode unitary LON, as defined in \cref{eq:LON}, followed by conditioning on a restricted class of measurements on WLOG the last $N-M$ modes. We set $q=W_{\ket{0}^{\otimes N}}$ which is a fixed point of the LON unitary. As previously, to ensure substochasticity, we restrict to measurements with $P$ and $\id-P$ Wigner-positive. With the choice $q'=W_{\ket{0}^{\otimes M}}$, \emph{every} such measurement, and thus the overall operation, is $(q,q')$-substochastic. This is in sharp contrast to the previous recipe. The reason is that the measurement acts on the Wigner function as 
\begin{equation}\label{normalizedmsmt}
    W_{\dyad{0}^{\otimes N}}\rightarrow W_{\dyad{0}^{\otimes M}} \Tr(P\dyad{0}^{\otimes N-M})
    \leq W_{\dyad{0}^{\otimes M}}.
\end{equation}
The inequality follows simply as outcome probabilities lie in the unit interval. %
Partial trace and homodyne measurements with conditioning are allowed operations because, as we established above, they are Wigner substochastic. Heterodyne measurement and conditioning on a region is also included in this set.

\subsubsection{Passive linear optical networks with Gaussian noise}
\label{sec:noisyLON}

A noisy LON is modeled by interleaving noise channels between the unitary layers,
\begin{equation}\label{eq:noisyLON}
    \widetilde{\U} \in \left\{ \N_L \circ \U_L \circ \dots \circ \N_1 \circ \U_1 \;|\; \U_1, \dots, \U_L \in {\rm LON} \right\} \,,
\end{equation}
where noise channel $\N_i$ follows the application of the $i^{\text{th}}$ unitary layer.

Photonic platforms predominantly suffer from Gaussian error sources, namely compositions of pure loss and Gaussian random noise~\cite{Weedbrook2012-qu,menicucci2014fault,Bourassa2021blueprint,bombin2021interleaving}.
Photon loss occurs at all optical elements including optical fibers used to connect components, and fiber-to-chip interfaces.
It is described by a linear bosonic channel $\N^{\rm loss}_\eta$ that acts on the first and second moments of a 1-mode state as follows,
\begin{equation}\label{eq:pureloss}
    \bmd \mapsto \sqrt{\eta} \bmd \,, \qquad V \mapsto \eta V + (1-\eta)\id \,,
\end{equation}
where $\eta \in [0,1]$ is the transmissivity of the channel.

Gaussian random noise (GRN) represents fluctuations that displace a state in phase space according to a Gaussian distribution, stemming for example from thermalization.
A Gaussian random noise channel $\N^{\rm GRN}_{\Sigma}$ is called isotropic when the phase-space fluctuations are rotationally invariant in phase space, i.e. they are independent and identically distributed in position and momentum.
In this case, the channel acts as
\begin{equation}\label{eq:isoranddisp}
    \bmd \mapsto \bmd \,, \qquad V \mapsto V + \Sigma \id \,,
\end{equation}
where $\Sigma \ge 0$ is the added covariance-matrix noise, corresponding to quadrature variance $\Sigma/4$.
A general phase-insensitive Gaussian attenuator can be obtained as the composition of these two channels,
\begin{equation}
    \N^{\rm att}_{\eta, \Sigma} \coloneqq \N^{\rm GRN}_{\Sigma} \circ \N^{\rm loss}_\eta
    \,.
\end{equation}
If each channel $\N_k$ is a tensor product of 1-mode attenuating channels of this form, the output reference distribution will depend on all the noise parameters $\{\eta_k, \Sigma_k\}_k$.
In the special case of depth-1 noisy LONs, i.e. $L=1$, the output reference state becomes an $M$-copy thermal state, $\tau' = \rho_{\rm th}(\bar{n})^{\otimes M}$, where
\begin{equation}
    \rho_{\rm th}(\bar{n}) 
    \coloneqq \left( \N^{\rm att}_{\eta_1, \Sigma_1} \circ \U \right) (\ketbra{0})
    = e^{-\beta\omega \hat a^\dag \hat a}/\Tr(e^{-\beta\omega \hat a^\dag \hat a})
    \,,
\end{equation}
where $\eta_1$ has no effect on the vacuum, and the state's mean photon number is
\begin{equation}
    \bar{n} \coloneqq \langle \hat a^\dag \hat a\rangle = \Sigma_1/2
    \,.
\end{equation}
Therefore, the class of allowed operations is $(q,q')$-substochastic with $q=W_{\ket{0}^{\otimes N}}$ and $q'=W_{\rho_{\rm th}(\bar{n})^{\otimes M}}$.

\subsubsection{Amplifying Gaussian channels}

Consider input and output states with the same number of modes and set $(q,q')$ to both be the uniform distribution equal to $1$. Then, as shown in Ref.~\cite{Upadhyaya2026}, the allowed operations include all Gaussian channels with $\abs{\det X}\geq 1$ and Gaussian-dilatable channels satisfying certain constraints including Wigner-positive environment.

\subsubsection{Larger classes of channels}
We can also consider larger classes of channels using kernels for representations other than the Wigner function. A prominent family of quasiprobability distributions is the $s$-ordered distributions~\cite{cahill1969density}, $s \in [-1,1]$, which includes the Wigner function~\cite{Wigner1932} at $s=0$. Phase-space point operators can be defined for these representations similarly to the Wigner function's. Then, the kernel of a CPTNI map where the input is in the $s$-ordering and the output is in the $t$-ordering is~\cite{Rahimi-Keshari2016}
\begin{equation}\label{skernel}
    K^{(t,s)}_{\Phi} (\bm{x},\bm{p},\bm{x}',\bm{p}') \coloneqq \pi^N \Tr( \hat{\Delta}_{t}(\bm{x},\bm{p}) \Phi(\hat{\Delta}_{-s}(\bm{x}',\bm{p}'))  ).
\end{equation}
Expressing the channel $\Phi$ in terms of its Choi matrix~\cite{watrous2018} via
\begin{equation}\label{choiaction}
    \Phi(A) = \Tr_2\left( J(\Phi) (\id_1 \otimes A^{\rm T}_2) \right)
    \,,
\end{equation}
we can rewrite the kernel in \cref{skernel} as 
\begin{align}
    K^{(t,s)}_{\Phi} (\bm{x},\bm{p},\bm{x}',\bm{p}') &= \pi^N \Tr( \hat{\Delta}_{t}(\bm{x},\bm{p}) \Tr_2( J(\Phi) (\id_1 \otimes \hat{\Delta}_{-s}(\bm{x}',\bm{p}')^{\rm T} )  ) )\\
    &=\pi^N \Tr( J(\Phi)(\hat{\Delta}_{t}(\bm{x},\bm{p}) \otimes \hat{\Delta}_{-s}(\bm{x}',-\bm{p}') )  ) 
\end{align}
where in the second line we have used that transposition on phase-space point operators acts as time-reversal, reversing momenta but not positions. Thus, for $t=-s$ we have %
\begin{equation}
    K^{(-s,s)}_{\Phi} (\bm{x},\bm{p},\bm{x}',\bm{p}') = \pi^N W^{(-s)}_{J(\Phi)} (\bm{x},\bm{p},\bm{x}',-\bm{p}') .  
\end{equation}
Decreasing $s$ corresponds to a convolution~\cite{cahill1969density}, so the condition $W_A^{(s)} \geq 0$ implies $W_A^{(s')} \geq 0$ for all $s' \le s$.
Thus, we have that if the $(-s,s)$-kernel of a channel is elementwise nonnegative, so too is the $(-s',s')$-kernel for $s'>s$. 
By choosing different values of $s$, we can thus tune the quantum channels whose kernels are stochastic. In the extreme case of $s=1$, the kernels map the P function of the input state to the Q function of the output state. 
In this case, every channel is stochastic because the kernel is always nonnegative, being the Husimi Q function of the Choi matrix.

\section{Certified majorization comparisons}
\label{sec:robustcurves}

Computing the Lorenz curves of a Wigner function analytically can be challenging, making it important to have a routine to obtain them numerically.
In order to numerically compute the positive and negative Lorenz curves relative to a reference distribution, we sample both functions on a discrete grid with a finite number of points. Since both functions are defined over $\mathbbm{R}^{2N}$, %
the curves we plot are only approximations to the true curves, as is generally the case with numerical methods.
For a specific state and reference distribution, we can change the sampling grid resolution, phase-space cutoff, and sparsity threshold by hand to gauge the convergence of the positive and negative curves to their true shape. 
However, a more principled approach that does not require manual analysis for each case would be desirable. In this section, we derive such an approach. We derive robust Lorenz \emph{envelopes} which are guaranteed to contain the true positive and negative curves. We emphasize that these are not probabilistic error bars, but rigorous deterministic bounds on the true curves that account for numerical error. For vanishing numerical error, the envelopes converge to the true curves.

Given robust envelopes for two pairs of Lorenz curves, to certify the existence of a crossing we need to observe that the entire envelopes interchange which one is on top. It is then guaranteed that the true curves must also cross, while otherwise we cannot conclude whether the true curves cross or not. Both cases are illustrated in \cref{robustbounce}.
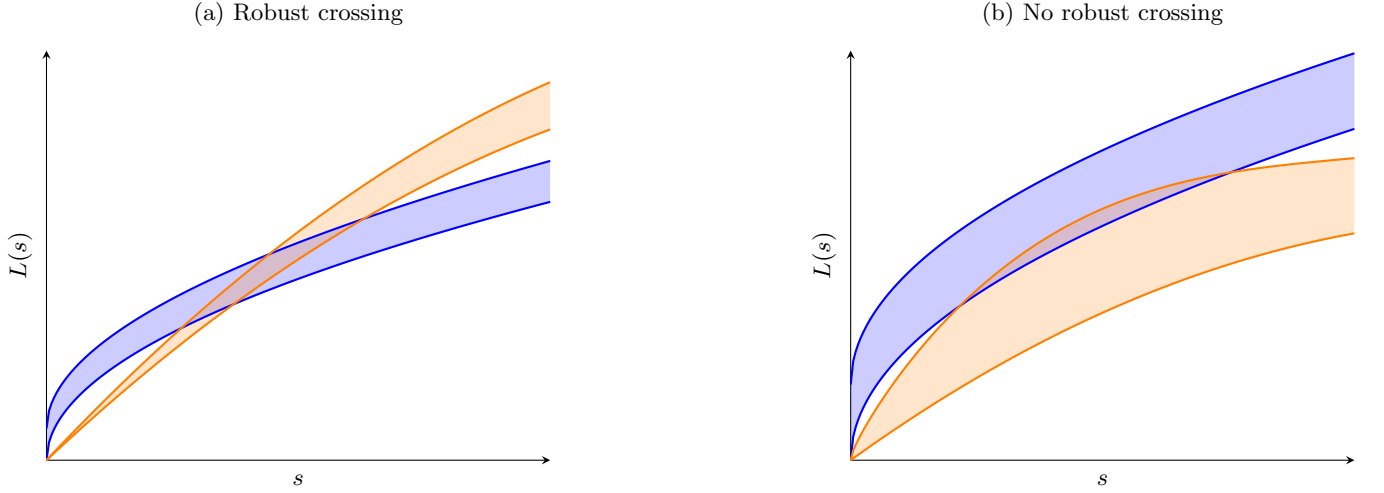
\begin{figure}[htbp]
    \centering
    \begin{tikzpicture}
        \begin{axis}[
            width=0.46\textwidth, height=7cm,
            title={(a) Robust crossing},
            axis lines=left,
            xlabel={$s$},
            ylabel={$L(s)$},
            xtick=\empty, ytick=\empty,
            domain=0:1, samples=200,
            xmin=0, xmax=1,
            ymin=0, ymax=1.3,
        ]
            \def\err{0.4} %

            \addplot[blue, thick, name path=fupper]
                {h(x)+0.05*x+0.1};
            \addplot[blue, thick, name path=flower]
                {h(x)+0.02*x+0};
            \addplot[blue, fill opacity=0.2, draw=none]
                fill between[of=fupper and flower];
            
            \addplot[orange, thick, name path=gupper]
                {1.7*x-0.5*x^2+0.};
            \addplot[orange, thick, name path=glower]
                {1.5*x-0.45*x^2};
            \addplot[orange, fill opacity=0.2, draw=none]
                fill between[of=gupper and glower];
        \end{axis}
    \end{tikzpicture}
    \hfill
    \begin{tikzpicture}
        \begin{axis}[
            width=0.46\textwidth, height=7cm,
            title={(b) No robust crossing},
            axis lines=left,
            xlabel={$s$},
            ylabel={$L(s)$},
            xtick=\empty, ytick=\empty,
            domain=0:1.2, samples=200,
            xmin=0, xmax=1.2,
            ymin=0, ymax=1.3,
        ]
            \def\err{0.04} %
            \addplot[blue, thick, name path=fupper]
                {1.2*(h(x)+0*x+0.2)};
            \addplot[blue, thick, name path=flower]
                {1.2*(h(x)+0)};
            \addplot[blue, fill opacity=0.2, draw=none]
                fill between[of=fupper and flower];
            
            \addplot[orange, thick, name path=gupper]
                {1*(1.75*x-1.5*x^2+0.43*x^3+0.+0.3*h(1.1*x))};
            \addplot[orange, thick, name path=glower]
                {1.2*(0.8*x-0.25*x^2)};
            \addplot[orange, fill opacity=0.2, draw=none]
                fill between[of=gupper and glower];
        \end{axis}
    \end{tikzpicture}
    \caption{Illustration of (a) a certified crossing: no matter where the true curves lie in their respective envelopes, they must cross; and (b) the absence of a robust crossing: the true curves could lie in their envelopes without intersecting, we thus cannot conclude whether they do or do not intersect.}
    \label{robustbounce}
\end{figure}

In more detail, each of the positive and negative robust envelopes consists of an inner bound and an outer bound. We use this terminology instead of lower and upper bounds to highlight the symmetry between the positive and negative curves and minimize confusion with negative signs. Each bound comes from a discrete list $Z_{I/O,P/N}$ whose Lorenz curve is piecewise linear and can thus be computed exactly. As before, the lists' Lorenz curves are extended by a horizontal plateau so that they are defined for all $s$. The Lorenz curves of these lists for positive and negative outer and inner bounds are derived to satisfy
\begin{align}
    \declor{Z_\text{I,P}} \leq &\qdeclor{W_\rho} \leq \declor{Z_\text{O,P}}, \\
    \inclor{Z_\text{I,N}} \geq &\qinclor{W_\rho} \geq \inclor{Z_\text{O,N}}. 
\end{align}

We can now provide a numerically stable version of the statement about Lorenz curves from \cref{physmajor}.
\begin{theorem}\label{certifiedcrossing}
    Let $Z$ and $Z'$ be the discrete envelope lists for $(W_\rho,q)$ and $(W_{\rho'},q')$. The heralded conversion of $\rho$ to $\rho'$ with probability $p$ using any operation from $\O$ is impossible if there exists an $s$ such that 
    \begin{align}
        \declor{Z_\text{O,P}} &< p \declor{Z'_\text{I,P}}, \quad \text{or} \\
        \inclor{Z_\text{O,N}} &>p \inclor{Z'_\text{I,N}}.
    \end{align}    
\end{theorem}

In \cref{app:robust}, we explain how we construct these envelopes for a generic Wigner function and reference distribution. We require some regularity conditions on the Wigner function enforced via its first and second derivatives. The envelopes account for the cutoff in phase space over which we sample and the discreteness of the grid used. They also involve a coarse-graining step that reduces the size of the lists to improve computational efficiency while maintaining the bounds' validity. Our construction also allows numerical floating-point tolerance and function-evaluation error to be taken into account.

\section{Applications to photonic state conversions}
\label{sec:conversion}

In this section, we present four applications of our majorization framework to constraining state conversions.
We start in \cref{sec:squeezing} with a pedagogical example, demonstrating how considering the entire family of Wigner $\alpha$-R\'enyi divergences prohibits exact squeezing concentration under LONs.
Then, we use the Wigner 2-R\'enyi divergence in \cref{sec:cat} to obtain analytical tradeoff relations that constrain heralded, approximate cat state conversions.
In \cref{sec:gkp}, we use robust numerical Lorenz curves to constrain interconversions of GKP states. 
Finally, in \cref{sec:Fock} we recover  tight constraints on Fock state concentration under LONs using analytical forms for portions of Lorenz curves.

\subsection{Squeezing amplification no-go via the family of Wigner R\'enyi divergences}
\label{sec:squeezing}

As a warm-up example, we show purely through majorization that LONs cannot concentrate squeezing.
That is, for any $N \ge 1$, $\ket{r}^{\otimes N} \not\rightarrow \ket{r'} \otimes \ket{0}^{\otimes N-1}$, whenever $|r| < |r'|$, where $\ket{r} \coloneqq \exp(\frac{r}{2}(\hat{a}^{\dagger 2} - \hat{a}^2))\ket{0}$ and similarly for $\ket{r'}$.
This result is clear by noting that passive unitaries (LONs) transform a covariance matrix $V$ to $V' \coloneqq OVO^{\rm T}$ for some orthogonal matrix $O$, thus the input and output covariance matrices $V, V'$ must have equal spectra, which is impossible when $|r| < |r'|$.

We can reach the same by considering the family of Wigner $\alpha$-R\'enyi divergences with $\alpha > 1$.
In \cref{app:wigner_cat}, we calculate that
\begin{equation}\label{eq:squ_renyi}
    D_\alpha(W_{\ket{r}}\|W_{\ket{0}}) 
    = -\frac{1}{2(\alpha-1)}\log\left( 1 - 4\alpha(\alpha-1)\sinh^2(r) \right) \,,
\end{equation}
if $|r| < \coth^{-1}(2\alpha-1)$ and diverges to $+\infty$ otherwise.
\cref{prop:div} requires that $N D_\alpha(W_{\ket{r}}\|W_{\ket{0}}) \ge D_\alpha(W_{\ket{r'}}\|W_{\ket{0}})$, which can be rearranged to the tradeoff relation
\begin{equation}\label{eq:squ_tradeoff}
    \cosh(2r') \le 1 + \frac{1 - \Big(1 - 4\alpha(\alpha-1)\sinh^2(r)\Big)^N}{2\alpha(\alpha-1)} \,.
\end{equation}
In the limit $\alpha \rightarrow 1^+$, the Wigner $\alpha$-R\'enyi divergence is finite for all $|r|$ and approaches $D_1(W_{\ket{r}}\|W_{\ket{0}}) = 2\sinh^2(r)$ which equals twice the energy of $\ket{r}$.
Then, \cref{eq:squ_tradeoff} coincides with the energy bound $\sinh^2(r') \le N\sinh^2(r)$.

Considering instead the entire family of Wigner $\alpha$-R\'enyi divergences, we can conclude that the transformation is impossible: given a target squeezing $|r'| > |r|$, we can choose
\begin{equation}
    \alpha = \frac{1}{2}(\coth|r'|+1) > 1 \,,
\end{equation}
making the Wigner $\alpha$-R\'enyi divergence finite at the input and infinite at the output.
Therefore, \cref{prop:div} guarantees a zero probability of conversion.

\subsection{Tradeoff relations for approximate cat state distillation via robust divergences}
\label{sec:cat}

Here, we study approximate cat state conversions under the class of operations described in \cref{sec:noisyLON}.
Accounting for experimental imperfections, we consider the set $\O$ of depth-1 noisy LONs, which map $q = W_{\ket{0}^{\otimes N}}$ to $q'=W_{\rho_{\rm th}(\bar{n})^{\otimes M}}$ in phase space.

The cat state is a superposition of coherent states with opposite displacement vectors,
\begin{equation}
    \ket{\alpha}_{\rm cat} \coloneqq A (\ket{\alpha} + \ket{-\alpha}) \,, \quad A \coloneqq \frac{1}{\sqrt{2(1+e^{-2|\alpha|^2})}} \,,
\end{equation}
for some $\alpha = X + iP$ where $X,P \in \mathbb{R}$.
Cat states find many important applications in photonic and superconducting platforms. 
A predominant application is in their use as hardware-efficient, biased-noise qubits.
In particular, logical bit flips can be suppressed exponentially in $|\alpha|$ reducing the requirements for quantum error correction by allowing the use of phase-flip repetition codes~\cite{chamberland2022building}.
Cat states are also common in quantum metrology and sensing as they achieve Heisenberg-limited sensitivity~\cite{munro2002weak,joo2011quantum}, and in quantum communication, enabling quantum teleportation~\cite{enk2001entangled,wang2001quantum}.
Finally, they can be thought of as minimal finite-term approximations of GKP states and are in fact used to breed GKP states~\cite{vasconcelos2010all,Weigand2018generating}.

Under the Gaussian attenuating noise model we are considering, a noisy cat state is defined as
\begin{equation}
    \rho(\alpha; \eta, \Sigma) \coloneqq \N^{\rm att}_{\eta, \Sigma}(\ketbra{\alpha}_{\rm cat}) \,,
\end{equation}
for some $\eta \in [0,1]$ and $\Sigma \ge 0$.
In this setting, the Wigner 2-R\'enyi divergence of an attenuated cat state relative to a thermal state is computable analytically as the integrand $W_\rho(x,p)^2 W_\tau(x,p)^{-1}$ contains only real Gaussian terms and interference terms that can be understood as complex Gaussian functions in $x,p$.
We perform this derivation in \cref{app:wigner_cat}, and present the result here.
\begin{lemma}\label[lemma]{cat2renyi}
    Let $\alpha \in \mathbb{C}$, $\eta \in [0,1]$, and $\bar{n}, \Sigma \ge 0$.
    The Wigner 2-R\'enyi divergence of an attenuated cat state $\rho(\alpha; \eta, \Sigma)$ relative to a thermal state $\rho_{\rm th}(\bar{n})$ is
    \begin{equation}\label{eq:D2catnoisy}
        D_2(W_{\rho(\alpha; \eta, \Sigma)} \| W_{\rho_{\rm th}(\bar{n})})
        = \log \frac{K (1+2\bar{n})^2 \Big( \cosh\left(2(1 + 2\eta K (\Sigma - 2\bar{n})) |\alpha|^2\right) \cosh\left(4\eta K (1+2\bar{n}) |\alpha|^2\right) + 1 \Big)}{\cosh(2|\alpha|^2) + 1}
        \,,
    \end{equation}
    with
    \begin{equation}
        K = (1+\Sigma)^{-1}(1+4\bar{n}-\Sigma)^{-1}
        \,,
    \end{equation}
    if $\Sigma < 1+4\bar{n}$ and $D_2(W_{\rho(\alpha; \eta, \Sigma)} \| W_{\rho_{\rm th}(\bar{n})}) = +\infty$ otherwise.
\end{lemma}
In the limit of complete attenuation, $\eta = 0$, the expression in \cref{eq:D2catnoisy} simplifies to the classical 2-R\'enyi divergence of a Gaussian distribution with variance $\frac{1+\Sigma}{4}\id_2$ relative to a Gaussian distribution with variance $\frac{1+2\bar{n}}{4}\id_2$.
In the opposite regime where all noise is absent, $\eta = 1$ and $\bar{n} = \Sigma = 0$, the expression simplifies to
\begin{equation}\label{eq:D2catnoiseless}
    D_2(W_{\dyad{\alpha}_{\rm cat}} \| W_{\dyad{0}})
    = \log\Big(1 + 2\cosh(2|\alpha|^2)\left(\cosh(2|\alpha|^2)-1\right)\Big)
    \,.
\end{equation}

Our goal is to approximately convert $N$ copies of a noisy cat state $\rho(\alpha; \eta, \Sigma)^{\otimes N}$ to $M$ copies of another noisy cat state $\rho(\alpha'; \eta', \Sigma')^{\otimes M}$, using depth-1 noisy LONs.
Combining the expression derived in \cref{cat2renyi} with \cref{corollaryprobbound_extended}, we obtain in \cref{app:wigner_cat} the following analytical tradeoff relations between the parameters of interest.
\begin{corollary}\label{thm:catstates}
    Let $\bar{n} > 0$ and consider the set of operations $\O$ with $q = W_{\ket{0}^{\otimes N}}$ and $q'=W_{\rho_{\rm th}(\bar{n})^{\otimes M}}$ (see \cref{sec:noisyLON}).
    Let $\rho(\alpha; \eta, \Sigma)$ and  $\rho(\alpha'; \eta', \Sigma')$ with $\eta, \eta' \in [0,1], \alpha,\alpha' \in \mathbb{C}$ and $\Sigma, \Sigma' \ge 0$ be the input and target noisy cat states, respectively.
    
    Suppose $\Sigma < 1$ and $\Sigma' < 1 + 4\bar{n}$, and define the constant
    \begin{equation}
        c
        = \frac{1 + 4\bar{n} - \Sigma'}{2(2\bar{n}+1)}
        \in (0, 1)
        \,.
    \end{equation}
    Then, there exists a constant $\Delta > 0$, such that:
    If an operation $\E \in \O$ achieves the probabilistic approximate conversion
    \begin{equation}\label{cat_conversion}
        \E\left( \rho(\alpha; \eta, \Sigma)^{\otimes N} \right) = p\tilde{\rho}
        \,, \qquad
        \frac{1}{2} \|\tilde{\rho} - \rho(\alpha'; \eta', \Sigma')^{\otimes M} \|_1 \le \epsilon
        \,,
    \end{equation}
    with target precision $\epsilon \in [0, \Delta^{-1/c})$, the following bound must be satisfied
    \begin{equation}\label{pbound_cat}
        p^2 (1-\Delta\epsilon^c) \le \exp( N D_2(W_{\rho(\alpha; \eta, \Sigma)} \| W_{\ket{0}}) - M D_2(W_{\rho(\alpha'; \eta', \Sigma')} \| W_{\rho_{\rm th}(\bar{n})}) )
        \,.
    \end{equation}
    Equivalently, whenever the bound in \cref{pbound_cat} is violated, there exists no $\E \in \O$ achieving the conversion in \cref{cat_conversion}.
\end{corollary}
We emphasize that the constants $\Delta$ and $c$ are independent of the perturbed state $\tilde{\rho}$ and error rate $\epsilon$.
An expression for $\Delta$ in terms of the quantities involved can be found  in \cref{app:wigner_cat}.

While the general bound in \cref{corollaryprobbound_extended} can in principle be optimized over the R\'enyi order $\alpha$, such an optimization is analytically and computationally prohibitive.
Even computing the R\'enyi divergence generally requires expensive integration over a high-dimensional phase space.
This highlights the importance of our analytical bound in \cref{pbound_cat}.

\subsection{GKP state conversion no-goes via robust numerical Lorenz curves}
\label{sec:gkp}

In this section, we consider interconversions between GKP code states~\cite{Gottesman2001}. GKP codes offer a promising approach to encoding a qudit in a harmonic oscillator and protecting it against small displacement errors. They have led to the first demonstration of quantum error correction beyond the break-even point with a superconducting device~\cite{sivak2023real} and have been recently experimentally demonstrated in photonics~\cite{konno2024logical,larsen2025integrated}. GKP codes are experimentally attractive as logical Clifford unitaries can be implemented using Gaussian unitaries. Given a supply of GKP-encoded stabilizer states, Gaussian operations suffice for universal quantum computation~\cite{Baragiola_2019}. These features underpin architectures for scalable, fault-tolerant photonic quantum computing~\cite{Bourassa2021blueprint}. %
Here we focus on a recent protocol~\cite{Hosseinynejad2026} that uses GKP sensor states to prepare GKP encodings of magic states.

Ideal GKP states are not normalizable. We use one of the standard models for a non-ideal GKP state which is damping by an energy envelope~\cite{Matsuura2020}, 
\begin{equation}
\ket{\psi}_\text{physical} \propto e^{-\epsilon_\text{damp}  \hat{n}}  \ket{\psi}_\text{ideal},
\end{equation}
where $\psi$ is the encoded qudit state.

The Wigner functions of ideal GKP states are linear combinations of Dirac delta functions. Let a generic GKP state encoding the qudit state $\rho$ be represented as
\begin{equation}\label{eq:idealgkpwigner}
    W_{\rho,\text{ideal}}(\mathbf r)
    =\sum_{j\in J}c_j\,
    \delta^{(2)}(\mathbf r-\boldsymbol\mu_j).
\end{equation}
See Refs.~\cite{PhysRevLett.128.210502,PRXQuantum.6.010330} for the precise definition of $J$, $c_j$, and $\mu_j$.
Then, the Wigner function of the damped and renormalized state is a linear combination of finite-width Gaussians~\cite{Bourassa2021}:
\begin{align}\label{eq:dampedgkpwigner}
    W_{\rho,\text{physical}}(\mathbf r)
    &=\sum_{j\in J}c_{j,\epsilon_\text{damp}}\,
    G_{\epsilon_\text{damp}}\!\left(
    \mathbf r-\boldsymbol\mu_j \sech \epsilon_\text{damp} \right),
\end{align}
where 
\begin{align}
G_{\epsilon_\text{damp}}(\mathbf r')&\coloneqq\frac{2}{\pi \tanh\epsilon_\text{damp}}
    \exp\!\left(-\frac{2|\mathbf r'|^2}{\tanh\epsilon_\text{damp}}\right) \quad \text{and} \quad  
    c_{j,\epsilon_\text{damp}}
    \coloneqq\frac{c_j
    \exp\!\left(-2 \tanh\epsilon_\text{damp} |\boldsymbol\mu_j|^2\right)}
    {\sum_{j\in J}c_j
    \exp\!\left(-2 \tanh\epsilon_\text{damp} |\boldsymbol\mu_j|^2\right)}.
\end{align}
The damping broadens each peak, contracts the peak centers, and suppresses the weights of peaks far from the origin.

We list three specific states that will arise in our examples. The sensor or qunaught state is 
\begin{equation}\label{eq:sensor-ket}
    \ket{\varnothing}_\text{ideal}
    \propto\sum_{j\in\mathbb Z}\ket{x=j\sqrt{\pi}},
\end{equation}
and its Wigner function is
\begin{equation}\label{eq:sensorwigner}
    W_{\ket{\varnothing},\text{ideal}}(\mathbf r)
    \propto\sum_{m,n\in\mathbb Z}(-1)^{mn}\,
    \delta^{(2)}\!\left(\mathbf r-\frac{\sqrt{\pi}}{2}
    \begin{pmatrix} m \\ n \end{pmatrix}\right).
\end{equation}
The logical zero and one states ($i=0,1$) are
\begin{equation}
    \ket{i}_{\text{ideal}}
    \propto\sum_{j\in\mathbb Z}\ket{x=j\sqrt{2\pi}+ i\sqrt{\frac{\pi}{2}}}, 
\end{equation}
with Wigner functions
\begin{equation}\label{eq:logicalwigner}
    W_{\ket{i},\text{ideal}}(\mathbf r)
    \propto\sum_{m,n\in\mathbb Z}(-1)^{n(m+i)}\,
    \delta^{(2)}\!\left(\mathbf r-\frac12\sqrt{\frac{\pi}{2}}
    \begin{pmatrix} 2m \\ n \end{pmatrix}\right).
\end{equation}
Numerically, we truncate the sums and accordingly estimate the tolerance for pointwise accuracy.

\subsubsection{State generation from GKP sensor states}
Approximate conversions from damped sensor states to damped magic GKP states using linear optics and postselection are considered in Ref.~\cite{Hosseinynejad2026}. Based on their setup, we consider a scenario where exact, probabilistic conversion is possible.
The input state is two copies of a damped sensor state. The desired output is an H magic GKP state at the same damping. Protocols to obtain GKP encodings of magic states are of interest as deterministic generation from stabilizer GKP states requires demanding non-Gaussian operations~\cite{PRXQuantum.6.010330}. We use the convention $\ket{H}_{\text{ideal}}= \sin \pi/8 \ket{0}_{\text{ideal}} - \cos \pi/8 \ket{1}_{\text{ideal}}$. A simple protocol achieves nonzero conversion probability. First, interfering both input states on a balanced beamsplitter gives the two-mode state 
\begin{equation}
 e^{-\epsilon_\text{damp} (\hat{n}_1+\hat{n}_2)} (\ket{00}_{\text{ideal}} + \ket{11}_{\text{ideal}}),
\end{equation}
where we have omitted the local rotations as compared to Eq.~1 of~\cite{Hosseinynejad2026}. %
Then, consider measuring the second mode with a coherent state POVM element $\ketbra{\gamma}/2$, with $\gamma=e^{\epsilon_\text{damp}} \sqrt{\frac{\pi}{2}} (1-i)$. This is within the class of allowed operations in \cref{recipelon}. The coherent-state amplitude is chosen so that the damping is canceled by the first factor and the second factor satisfies~\cite{Rymarz2021}
\begin{align}
    \sin \frac{\pi}{8} \braket{\sqrt{\frac{\pi}{2}} (1-i)}{1}_{\text{ideal}} = -\cos \frac{\pi}{8} \braket{\sqrt{\frac{\pi}{2}} (1-i)}{0}_{\text{ideal}}.
\end{align}
This ensures that the conditional state on the first mode is, up to an overall normalization and phase, $\ket{H}_{\text{physical}}$.

We follow the recipe in \cref{recipelon}, so compute Lorenz curves relative to copies of the Wigner function of vacuum.
Setting the damping strength at $\epsilon_\text{damp}=0.5$, in \cref{fig:gkpsensormagic} we plot the robust Lorenz curve of both states. We sweep the success probability downwards from $p=1$, which corresponds to vertically rescaling the output Lorenz envelopes, until the curves no longer robustly intersect. This gives an upper bound of $p<0.472$. %

\begin{figure}[h!]
    \centering
    \includegraphics[width=0.5\linewidth]{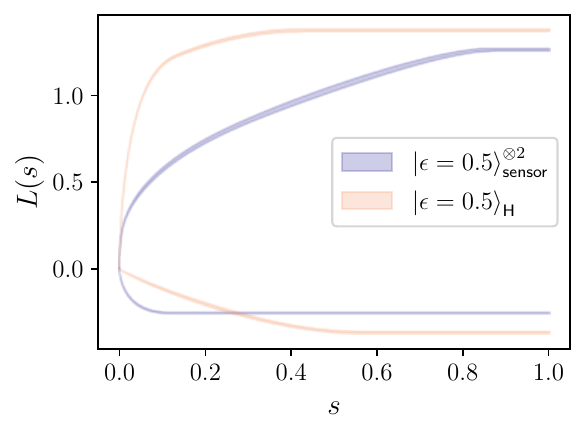}
\captionof{figure}{Positive and negative robust Wigner Lorenz curves for damped sensor and H GKP states. As the output curve is vertically scaled by $p\leq 1$, the envelopes continue to intersect robustly, until for $p<0.472$ they no longer do.}
\label{fig:gkpsensormagic}
\end{figure}

Modifying the choice of coherent-state measurement in the previous protocol allows for the generation of other GKP states. For example, $\gamma=e^{\epsilon_\text{damp}} \sqrt{\frac{\pi}{2}} (1+i/2)$ conditionally generates $\ket{1}_{\text{physical}}$. We can derive upper bounds under a different class of operations by using regular majorization. This is in line with the recipe in \cref{recipegaussian} but without ancillary vacuum modes at the input. We need to specify the maximum allowed trace of the conditioning POVM element. We use $q'=\pi/2$, which includes our explicit protocol that has $\Tr(P)=1/2$.
The robust curves are shown in \cref{fig:gkpsensorstab} for $\epsilon_\text{damp}=0.4$. The probability sweep gives a bound $p<0.296$.

\begin{figure}[h!]
    \centering
    \includegraphics[width=0.5\linewidth]{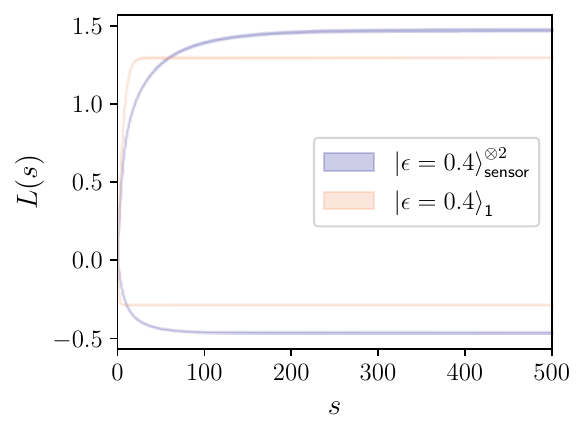}
\captionof{figure}{Positive and negative robust Wigner Lorenz curves for damped sensor and logical $\ket{1}_{\text{physical}}$ GKP states. As the output curve is vertically scaled by $p\leq 1$, the envelopes continue to intersect robustly, until for $p<0.296$ they no longer do.}
\label{fig:gkpsensorstab}
\end{figure}

\subsection{Tight Fock state concentration bounds via analytical Lorenz curves near the origin}
\label{sec:Fock}

In this section, we provide an analytical upper bound on the probability of concentrating $n$ single-photon Fock states to one Fock state $\ket{n}$ using linear optics. We match the tight bound derived in Ref.~\cite{Aaronson2012} even though we allow for a larger class of operations. Our bound follows by considering positive Lorenz curves of Fock states relative to vacuum. In particular, we derive analytical forms for the curves at small $s$.

\subsubsection{Allowed operations}

The allowed operations in Ref.~\cite{Aaronson2012} consist of appending arbitrary copies of vacuum, applying LON unitaries, and then measuring all the modes in the Fock basis. 
The probability of getting $n$ photons in, say, the first mode is shown to be $p\leq \frac{n!}{n^n}$. 
Here, we show that majorization gives the same tight bound under a strictly larger class of operations, i.e. under the majorization recipe in \cref{recipelon}. 
Since the associated set of operations does not include heralding on photon number measurements, we first explain what quantity we bound and how it constrains the conversion as set up in Ref.~\cite{Aaronson2012}.

Let the multimode state prior to measurement be $\ket{\theta}$. After dephasing in the Fock basis, the state becomes
\begin{align}
    \Phi_{\rm deph}(\dyad{\theta}) &= \sum_{\bf{n}} p_{\bf{n}} \ketbra{\bf{n}}
\end{align}
where $\bf n$ indexes multimode Fock states. By conservation of total photon number, only terms with $\bf n$ satisfying $\sum n_i=n$ are nonzero. Dephasing in the Fock basis doesn't affect the measurement: the probability $p$ that is bounded in Ref.~\cite{Aaronson2012} is exactly $p_{n,0,...,0}$. 
Next, consider the reduced state on the first mode,
\begin{align}\label{targetfock}
    \Tr_{2...}(\Phi_{\rm deph}(\dyad{\theta})) &= \sum_{m} p'_m \ketbra{m}.
\end{align}
Again by photon number conservation, the only way to have $n$ photons in the first mode is if all other modes have zero photons, so $p'_n=p_{n,0,...,0}$; the only nonzero terms are those with $m\leq n$. 
Thus, for any $\theta$, the heralded success probability considered in Ref.~\cite{Aaronson2012} is equal to the weight of $\ketbra{n}$ on the first mode after dephasing.
In the majorization recipe we are using, partial trace and Fock dephasing are allowed operations~\cite{Upadhyaya2026}. Hence, we can WLOG set \cref{targetfock} as the target state. Our result is then an upper bound on $p'_n$. This argument illustrates that majorization can provide constraints even in settings where the measurements are not Wigner-substochastic.

\subsubsection{Fock Lorenz curves}

We first find an analytical formula for the positive Lorenz curves of single-mode Fock states at small $s$. Starting from the definition of the distribution function, we approximate the Laguerre polynomials in the Wigner functions of Fock states by their highest order term. We can then derive an expression for the leading-order term of the positive Lorenz curve at small $s$. We prove in \cref{app:fock} that
\begin{align}\label{singlemodefock}
    \qqdeclor{W_m}{W_0}&\approx 2^m  \frac{s(-\log s)^m}{m!}.
\end{align}
The target state is a mixture of Fock states up to $n$ so its small-$s$ positive Lorenz curve is just $p'_n \qqdeclor{W_n}{W_0} $. That is, the highest power of $\log s$ dominates as $s\rightarrow 0$.

Next, we derive an expression for the positive Lorenz curve relative to vacuum of $n$ single photon states. Again we find an approximate expression for the distribution function and obtain the Lorenz curve to leading order:
\begin{align}\label{singlephotonfock}
    \qqdeclor{W_1^{\otimes n}}{W_0^{\otimes n}}&\approx\frac{2^n}{n^n} s (-\log s)^n.
\end{align}
The proof is in \cref{app:fock}.

With both curves in hand, we can then upper bound $p_n$ by the ratio of Lorenz curves at any point $s$, in particular as $s\rightarrow 0$. In this limit, we can directly substitute our approximate expressions as they become arbitrarily accurate at smaller $s$. Using the leading order expressions for both Lorenz curves from \cref{singlemodefock} and \cref{singlephotonfock}, we obtain the bound
\begin{align}
    p_n &\leq \lim_{s\rightarrow 0} \frac{\declor{W_1^{\otimes n}|W_0^{\otimes n}}}{\declor{W_n|W_0}} \\
    &= \lim_{s\rightarrow 0} \frac{\frac{2^n}{n^n} s (-\log s)^n}{ \frac{2^n}{n!} s(-\log s)^n} \\
    &= \frac{n!}{n^n}.
\end{align}
Thus, we exactly match the known tight bound from~\cite{Aaronson2012} for Fock concentration using linear optics. Our bound is valid even under the larger class of operations that are Wigner-substochastic and preserve vacuum. This includes for example Fock dephasing and other non-Gaussian operations. At $n=2$, the success probability is $1/2$ and is achieved by a symmetric beamsplitter, i.e. the Hong-Ou-Mandel effect~\cite{Hong1987}.

\section{Outlook}
\label{sec:conclusion}
In this work, we introduce a highly flexible framework for assessing photonic state conversion in various operational scenarios. We do so by extending majorization theory to weak, fully relative majorization for phase-space distributions in bounded and unbounded spaces. Using Lorenz curves and divergences, we show that this framework can place constraints on heralded state conversions, including approximate ones, under different classes of allowed operations. Known protocols for photonic state generation and conversion are usually probabilistic and approximate, making our framework especially suitable for their analysis. We also provide a robust Lorenz curve construction for computing certified constraints even in the presence of numerical errors. Focusing on the Wigner function for bosonic systems, we illustrate analytical and numerical bounds for several state conversion problems involving cat, GKP, and Fock states. The flexibility of the majorization machinery allows us to tailor our framework to different operational scenarios. 

A promising next step is to use our majorization framework to further analyze technologically relevant protocols and scenarios, such as NOON state preparation. Given the substantial success of majorization theory in establishing sufficient conditions for state conversion in entanglement theory~\cite{nielsen1999conditions}, finding analogous conditions for photonic systems, even for restricted sets of states, would be desirable.

\section*{Acknowledgments}
We thank Zacharie Van Herstraeten and Jack Davis for valuable discussions. T.U.\ acknowledges the support of the Joint Center for Quantum Information and Computer Science and the MathQuantum RTG (DMS-2231533).
O.H.\ acknowledges support from JST CREST Grant Number JPMJCR23I3.
U.C.\ acknowledges funding from the European Union’s Horizon Europe Framework Programme (EIC Pathfinder Challenge project Veriqub) under Grant Agreement No.~101114899.
N.K.\ acknowledges support from the NSF QLCI grant OMA-2120757.
Exchanges with LLM tools including ChatGPT-6 assisted with the proof of \cref{thm:semicontinuity}, the generation of \cref{fig:summary}, and proofreading. The authors take full responsibility for the correctness of the manuscript.  

\newpage

\appendix

\begin{center}
    {\Huge Appendix}
\end{center}

\section{Statements and proofs of weak majorization equivalences}
\label{app:majproofs}

We prove here our general results for regular, weak majorization and then extend them to the setting of relative majorization with different reference distributions. 

\subsection{Definitions and preliminaries for regular majorization}

Let $(X,\mathcal {A},\mu)$ be a $\sigma$-finite measure space, with $X$ a set, $\mathcal{A}$ a $\sigma$-algebra, and $\mu$ a measure on it. For example, $X = \mathbb{R}^{2n}$ with $\mu$ as the standard Lebesgue measure gives the bosonic phase space on $n$ modes.
As usual, let $L^1(X,\mu)$ denote the Lebesgue-integrable real-valued functions on $X$. 
For our purposes, a quasiprobability distribution $f$ is a function in $L^1(X,\mu)$ normalized so that $\int_X f {\rm d} \mu =1$.

Let $f \in L^1(X,\mu)$ be a distribution.
We define the distribution function $D_f$ and codistribution function $C_f$ of $f$ as
\begin{align}
    D_f(t) &\coloneqq \mu\{x: f(x) > t \} 
    \,, \\
    C_f(t) &\coloneqq \mu\{x: f(x) < t \}
    \,.
\end{align}
They respectively allow us to rigorously define the decreasing and increasing rearrangements of $f$,
\begin{align}
    f^\downarrow(u) &\coloneqq \inf\{ t: D_f(t)\leq u \}=\sup \{ t: D_f(t)> u \} 
    \,, \\
    f^\uparrow(u) &\coloneqq \sup\{ t: C_f(t)\leq u \}=\inf \{ t: C_f(t)> u \}
    \,.
\end{align}
The positive and negative parts of $f$ are denoted by $f^+\coloneqq\max\{f,0\}$ and $f^-\coloneqq\min\{f,0\}$.
The positive and negative Lorenz curves of $f$ are
\begin{align}
    L^{\Dsh}_f(s)&\coloneqq \int_0^s (f^+)^\downarrow {\rm d}u \, , \\
    L^{\Ish}_f(s)&\coloneqq  \int_0^s (f^-)^\uparrow {\rm d}u \,,
\end{align}
where ${\rm d} u$ indicates integration with respect to Lebesgue measure on $\mathbb{R}$. $s$ is in the range $[0,\mu(X)]$ when $\mu(X)$ is finite and $[0,\mu(X))$ when $\mu(X)$ is infinite. For $s>\mu(X)$, we define the Lorenz curve to be equal to its endpoint, i.e. extend it by a horizontal plateau. This becomes relevant for certain state-conversion comparisons. 

A linear operator $S: L^1(X,\mu) \rightarrow L^1(X',\mu')$ is an integral operator if its action can be expressed via a real-valued measurable function $S(x',x)$ on $X' \cross X$, called the kernel of $S$. Namely, for $g=Sf$,
\begin{equation}\label{eq:intgral_operator_kernel_def}
    g(x') = \int_{X} S(x',x) f(x) {\rm d}\mu(x) \,.
\end{equation}
An integral operator $S$ is positive
if $S(x',x) \geq 0$ for all $(x',x) \in X' \cross X$.
A positive integral operator $S$ is substochastic if $\int_{X'} S(x',x) {\rm d}\mu'(x') \le 1$ for almost all $x\in X$.
It is stochastic if the same holds with equality. A substochastic operator is doubly substochastic if $\int_{X} S(x',x) {\rm d}\mu(x) \le 1$ for almost all $x'\in X'$.

In order to consider atomic measure spaces, which arise for example for qudit Wigner functions, on the same footing as nonatomic ones, we first establish a lemma, similar to Lemma~5 of~\cite{Upadhyaya2026}, that lets us consider nonatomic measure spaces without loss of generality. For any space $X$, let $\bar{X}\coloneqq X \cross [0,1]$ and for any measure $\mu$, let $\bar{\mu}\coloneqq \mu \cross u$ with $u$ the Lebesgue measure on the unit interval.
For $f\in L^1(X,\mu)$, define $\bar{f}\in L^1(\bar{X},\bar{\mu})$ by $\bar{f}(x_1,x_2)= f(x_1)$. The (co)distribution functions satisfy $D_{\bar{f}}=D_f, C_{\bar{f}}=C_f$.
Analogous constructions apply to $(X', \mu')$ and $g$.
\begin{lemma}\label[lemma]{atomicspacestochastic}
    If there is a sequence of integral doubly substochastic operators $\bar{S}_n: L^1(\bar{X},\bar{\mu}) \rightarrow L^1(\bar{X'},\bar{\mu'})$ such that $\bar{S}_n \bar{f}$ converges to $\bar{g}$ in $L^1(\bar{X'},\bar{\mu'})$, then there is a sequence of integral doubly substochastic operators $S_n: L^1(X,\mu) \rightarrow L^1(X',\mu')$ such that $S_n f$ converges to $g$ in $L^1(X',\mu')$.
\end{lemma}

\begin{proof}

Consider the linear operator $R$ defined by
\begin{equation}
    (R h)(x_1,x_2) = h(x_1).
\end{equation}
Its kernel representation is a Dirac delta so it is a doubly stochastic operator, and can be approximated arbitrarily well by a sequence of integral doubly stochastic operators. The operator $T$ that takes us back to the original space, 
\begin{equation}
    (T h)(x_1) = \int_{[0,1]} h(x_1,x_2) \; {\rm d}u(x_2).
\end{equation}
is similarly doubly stochastic and admits a sequence of arbitrarily good integral approximations. Then, we can consider a new sequence of operators  $T\bar{S}_n R$ mapping $f$ to $g$. The composition of doubly substochastic operators is itself doubly substochastic, each element of the sequence can be approximated by integral operators, and the convergence in $L^1$ is preserved since $R,T$ are contractions.
\end{proof}

\subsection{Proof and statement for regular majorization}

\begin{theorem}\label{weakmajor}
Given $f\in L^1(X,\mu) $, $g\in L^1(X',\mu')$ with $\max\{\mu(X),\mu'(X')\}=M$ we say that $f$ weakly majorizes $g$, and write $f \succ\!\succ g$, if any of the following equivalent statements holds.
\begin{enumerate}
\item $L^{\Dsh}_f(s) \geq L^{\Dsh}_g(s)$ and $L^{\Ish}_f(s) \leq L^{\Ish}_g(s)$ for all $s\in [0,M]$;%
\item There exists a sequence of doubly substochastic integral operators $(S_n)_{n \in \mathbbm{N}}$ such that $S_n f$ converges to $g$ in $L^1(X',\mu')$;
\item $\int_X \phi(f) {\rm d}\mu \ge \int_{X'} \phi(g) {\rm d}\mu'$ for all nonnegative convex functions $\phi: \mathbbm{R} \rightarrow \mathbbm{R}_{\geq 0}$ satisfying $\phi(0)=0$ such that the integrals converge;%
\item $ \int_X (f-z)^+ {\rm d}\mu \ge \int_{X'} (g-z)^+ {\rm d}\mu'$ and $ \int_X (f+z)^- {\rm d}\mu \le \int_{X'} (g+z)^- {\rm d}\mu' $ for all real numbers $z \ge 0$.%
\end{enumerate}
\end{theorem}
Our proof is similar to the proof of Theorem~1 in Ref.~\cite{Upadhyaya2026}. Interestingly, our proof is actually simpler despite being more general. The reason is that splitting the Lorenz curves into positive and negative parts naturally leads to weak majorization since the curves may individually have different normalizations even when the overall functions have the same normalization.
\begin{proof}
Assume first that $\mu(X)=\mu'(X')$.\\
\noindent$\bm{1 \implies 2:}$ 
Let $F_P,G_P$ and $F_N,G_N$ be the asymptotic values of the positive and negative Lorenz curves. By \cref{atomicspacestochastic}, we can assume $f$ and $g$ are defined on a nonatomic measure space.
Following a similar construction as in Ref.~\cite{Upadhyaya2026}, define a reduced function $f_{\text{red}}=f(x) \chi_{A}$. The region $A$, which may not be of finite measure, is selected so that the positive and negative Lorenz curves of $f_{\text{red}}$ match those of $f$ up to $G_P$ and $G_N$. That is, we do not include the tails of $f$ in the region. There is a simple doubly substochastic operator $T$ taking $f$ to $f_{\text{red}}$, defined by 
\begin{equation}
    T h \coloneqq  \id(h \chi_{A}).
\end{equation}

By construction, $f_{\text{red}}$ majorizes $g$ and $f_{\text{red}}$ and $g$ have the same normalization. Therefore by Theorem~1 of~\cite{Upadhyaya2026}, there exists a sequence of semidoubly stochastic operators $S_n$ such that $S_n f_{\text{red}} \rightarrow g$ in $L^1(X',\mu')$. Strictly speaking, that theorem only applies when the two functions being compared are on the same measure space. We can however generalize to functions on two different measure spaces with the same total measure. This is accomplished by using Theorem 2.2 of~\cite{Chong1976} in place of Theorem 4.9 of~\cite{Day1973}. Once this is done, $S_n T $ is a sequence of doubly substochastic integral operators satisfying the desired equation.

\noindent$\bm{2\implies 3:}$
The measure defined by $\nu_{n,x'}(A) = \int_A S_n(x',x) {\rm d}\mu(x)$ satisfies $\nu_{n,x'}(X)\leq 1$, for almost all $x'$, because $S_n$ is doubly substochastic. 
Assuming $\nu_{n,x'}(X)>0$, we apply Jensen's inequality~\cite{jensen1906fonctions} to the renormalized measure, 
\begin{align}
    \phi( S_n f )(x')&= \phi\left( \int_X S_n(x',x) f(x) {\rm d}\mu(x)  \right)\\
    &= \phi \left( \int_X \nu_{n,x'}(X) f(x) \ \frac{1}{\nu_{n,x'}(X)}{\rm d}\nu_{n,x'}(x) \right)\\
    &\leq  \int_X \phi \left[ \nu_{n,x'}(X) f(x)\right] \ \frac{1}{\nu_{n,x'}(X)}{\rm d}\nu_{n,x'}(x). 
\end{align}
By assumption $\phi$ is convex and vanishes at the origin. Thus, 
\begin{equation}
    \phi \left[ \nu_{n,x'}(X) f(x)\right] \leq \nu_{n,x'}(X) \phi \left[ f(x)\right] + \left[1-\nu_{n,x'}(X)\right] \phi(0) = \nu_{n,x'}(X) \phi  \left[ f(x)\right] \,,
\end{equation}
leading to the inequality,
\begin{equation}
    0\leq\phi( S_n f )(x') \leq \int_X \phi \left( f(x)\right) {\rm d}\nu_{n,x'}(x) \eqqcolon h_n(x'), \label{convexub}
\end{equation}
since $\phi$ is nonnegative. If $\nu_{n,x'}(X)=0$, the same inequality follows as both sides are zero.
We compute
\begin{align}
    \int_{X'} h_n(x') {\rm d}\mu'(x')&= \int_{X'} \int_X \phi \left( f(x)\right) \ S_n(x',x) {\rm d}\mu(x) {\rm d}\mu'(x')\\
    &=  \int_{X} \phi \left( f(x)\right) \ \int_{X'} S_n(x',x)  {\rm d}\mu'(x') {\rm d}\mu(x)\\
    &\leq  \int_X \phi \left( f(x)\right) {\rm d}\mu(x), \label{jensenfubini}
\end{align}
where the second line follows from Fubini's theorem~\cite{fubini1907sugli,Rudin1987}, the third line follows from substochasticity of $S_n$, and the final integral exists by assumption.

Since $S_n f$ converges to $g$ in $L^1$, there exists a subsequence such that $S_{n_i} f$ converges pointwise to $g$ almost everywhere.
Since $\phi$ is convex, it is also continuous on $\mathbb{R}$ so $\phi(S_{n_i} f)$ converges pointwise to $\phi(g)$ almost everywhere. Since the limit exists almost everywhere the limit inferior equals it. Then, we have that
\begin{align}
    \int_{X'} \phi(g) {\rm d}\mu' &=\int_{X'} \liminf_{i\rightarrow \infty} \phi(S_{n_i} f) {\rm d} \mu'\\
    &\leq  \liminf_{i\rightarrow \infty} \int_{X'} \phi(S_{n_i} f) {\rm d} \mu'\\
    &\leq  \int_X \phi(f) {\rm d} \mu,
\end{align}
where the second line follows by Fatou's lemma~\cite{fatou1906series} and the third line by \cref{convexub} and \cref{jensenfubini}.

\noindent$\bm{3\implies 4:}$
The functions $s \mapsto (s-z)^+$ and $s \mapsto -(s+z)^-$ for $z\geq0$ are convex and satisfy the requirements on the assumption. We then get immediately that
\begin{align}
    \int_X (f(x)-z)^+ {\rm d}\mu(x) &\geq \int_{X'} (g(x')-z)^+ {\rm d}\mu'(x'),\\
    \int_X (f(x)+z)^- {\rm d}\mu(x) &\leq \int_{X'} (g(x')+z)^- {\rm d}\mu'(x')
\end{align}
for all $z\geq0$. All integrals converge because $f,g$ are integrable. %

\noindent$\bm{4\implies 1:}$
$\int_{X'} (g-z)^+ {\rm d}\mu' \leq \int_X (f-z)^+ {\rm d}\mu$ by assumption. By Lemma~3 of~\cite{Upadhyaya2026} which extends Corollaries 1.2 and 1.3 of~\cite{Chong1974}, this is equivalent to $\int_z^\infty D_{g^+} \leq \int_z^\infty D_{f^+} $. For any $s\in (0,\mu(X))$, the choice $z=(f^+)^\downarrow(s)\geq 0$ ensures $D_{f^+}(z) \leq s \leq D_{f^+}(z^-)$. Then, by Lemma~4 of~\cite{Upadhyaya2026}, which is a slight generalization of Lemma~1.5 of~\cite{Chong1974}, we have
\begin{align}
    \int_0^s ((g^+)^\downarrow -z) {\rm d}u &\leq  \int_z^\infty D_{g^+}(t) {\rm d} t\\
    &\leq \int_z^\infty D_{f^+}(t) {\rm d} t\\
    &=\int_0^s ((f^+)^\downarrow -z) {\rm d}u \,.
\end{align}

Similarly, we have $\int_{X'} (g+z)^- {\rm d}\mu' \geq \int_X (f+z)^- {\rm d}\mu$ if and only if $ - \int_{-\infty}^{-z} C_g \geq - \int_{-\infty}^{-z} C_f$. Setting $z=-(f^-)^\uparrow(s)\geq 0$, we guarantee that $C_{f^-}(-z) \leq s \leq C_{f^-}((-z)^+)$, leading to
\begin{align}
    \int_0^s ((g^-)^\uparrow + z) {\rm d}u &\geq  - \int_{-\infty}^{-z} C_{g^-}(t) \  {\rm d}t\\
    &\geq - \int_{-\infty}^{-z} C_{f^-}(t) \  {\rm d}t \\
    &=\int_0^s ((f^-)^\uparrow +z) {\rm d}u \,.
\end{align}
In summary we obtain $L^{\Dsh}_f(s) \geq L^{\Dsh}_g(s)$ and $L^{\Ish}_f(s) \leq L^{\Ish}_g(s)$ for all $s$.

If $\mu'(X')>\mu(X)$, define a new measure space as follows: $\tilde X \coloneqq X\sqcup[0,U)$ for $U=\mu'(X')-\mu(X)$ and $ \tilde \mu \coloneqq \mu\oplus u$ for $u$ denoting the Lebesgue measure on the real line. Formally extend $f$ to this larger measure space, $\tilde f \coloneqq f\oplus0$. If instead, $\mu'(X')<\mu(X)$, extend $g$ analogously. The extended measure spaces are $\sigma$-finite and the extended functions are clearly $L^1$ and have the same Lorenz curves as their unextended counterparts, up to extensions by horizontal plateaus, as accounted for in our formal definition of the Lorenz curves. The integrals in statements 3 and 4 are unchanged since the functionals in the integrands all take the value zero at zero. A doubly substochastic operator on the original space can be trivially padded to a doubly substochastic operator on the extended spaces. Conversely, restricting a doubly substochastic operator on the extended spaces to one on the original space preserves its double substochasticity. Thus, we can WLOG apply the version of the theorem with equal measures.

\end{proof}

\subsection{Definitions and preliminaries for relative majorization}
Let $q: X \rightarrow \mathbb{R}_{>0}$ be any strictly positive, measurable function. The above definitions can be generalized in a manner that is necessary for relative majorization. The key idea is to introduce a new measure $\nu$~\cite{Joe1990}. Define $\nu(A)= \int_A q {\rm d} \mu$. Roughly speaking, this is multiplying the original measure pointwise by $q$.
The relative distribution function $D^q_f$ and codistribution function $C^q_f$ are
\begin{align}
    D^q_f(t)&\coloneqq\nu\{x: f(x)> tq(x) \} \, \\
    C^q_f(t)&\coloneqq\nu\{x: f(x)< tq(x) \} \,.
\end{align}
Then, the decreasing and increasing rearrangements of $f$ relative to $q$ are
\begin{align}
    f^{\downarrow,q}(u)&\coloneqq\inf \{ t: D^q_f(t)\leq u \}, \\
    f^{\uparrow,q}(u)&\coloneqq\sup \{ t: C^q_f(t)\leq u \} \,.
\end{align}
Then, the positive and negative Lorenz curves of $f$ relative to $q$ are
\begin{align}
    \qdeclor{f}\coloneqq \int_0^s (f^+)^{\downarrow,q} \ {\rm d}u \,, \\ \qinclor{f}\coloneqq \int_0^s (f^-)^{\uparrow,q} \ {\rm d}u \,.
\end{align}
The range of $s$ is again $[0,\nu(X)]$ or $[0,\nu(X))$ depending on whether $\nu(X)$ is finite. And again, we formally define the Lorenz curve for $s>\nu(X)$ to equal its endpoint value.

Generalizing the notion of double substochasticity, a substochastic operator is $(q,q')$-substochastic if it satisfies $\int_{X} S(x',x) q(x) {\rm d}\mu(x) \le q'(x')$ for almost all $x'\in X'$.

\subsection{Proof and statement for relative majorization}

Next we give the extension to full relative weak majorization and prove the equivalences by reduction to regular weak majorization.

\begin{theorem}\label{relmajor}
    Consider two integrable functions $f\in L^1(X,\mu) $, $g\in L^1(X',\mu')$ on $\sigma$-finite measure spaces and two strictly positive, measurable reference functions $q: X \rightarrow \mathbbm{R}_{>0} $, $ q': X'\rightarrow \mathbbm{R}_{>0} $ such that $\max\{\int_X q {\rm d}\mu,\int_{X'} q' {\rm d}\mu'\}= M$. 
    We say that $f$ weakly majorizes $g$ relative to $(q,q')$, and write $(f,q) \weak (g,q')$, if any of the following equivalent statements holds.
    \begin{enumerate}
    \item $\qdeclor{f} \geq \qpdeclor{g}$ and $\qinclor{f} \leq \qpinclor{g}$ for all $s\in [0,M]$;
    \item There exists a sequence of substochastic integral operators $(S_n)_{n \in \mathbbm{N}}$ such that $S_n f$ converges to $g$ in $L^1(X',\mu')$ and $(S_n q)(x') \le q'(x')$  for all $n \in \mathbbm{N}$ and for almost all $x'\in X'$;
    \item  $\int_X q \phi\left(\frac{f}{q}\right) {\rm d}\mu \ge \int_{X'} q' \phi\left(\frac{g}{q'}\right) {\rm d}\mu'$ for all nonnegative convex functions $\phi: \mathbbm{R} \rightarrow \mathbbm{R}_{\geq 0}$ satisfying $\phi(0)=0$ such that the integrals converge;
    \item $ \int_X (f-zq)^+ {\rm d}\mu \ge \int_{X'} (g-zq')^+ {\rm d}\mu' $ and $ \int_X (f+zq)^- {\rm d}\mu \le \int_{X'} (g+zq')^- {\rm d}\mu'  $ for all real numbers $z \ge 0$.
    \end{enumerate}
\end{theorem}

\begin{proof}
We follow the approach of Ref.~\cite{Joe1990} of reducing relative majorization to regular majorization via a change of measure.
Define the measure $\nu$ via
\begin{equation}\label{eq:nu}
     \nu(A) \coloneqq \int_X \chi_A \, q(x) \, {\rm d}\mu(x) \,,\quad \text{for all measurable } A \subseteq X \,.
\end{equation}
and analogously define $\nu'$.

We observe simply that $ \nu\{x: f(x)> tq(x) \}=\nu\{x: f(x)/q(x)> t \}$, and similarly for the other (co)distribution functions. Thus, Statement 1 can be rewritten as 
\begin{equation}
    \declor{\frac{f}{q}} \geq \declor{\frac{g}{q'}} \text{ and } \inclor{\frac{f}{q}}\leq \inclor{\frac{g}{q'}} \,,\quad \text{for all } s \in [0,M] \,,
\end{equation}
Consider substochastic integral operators $S_n:L^1(X,\mu)\rightarrow L^1(X',\mu')$ satisfying $S_n q \leq q'$.
Define a new sequence of integral operators via the kernels $T_n(x',x) = S_n(x',x)/q'(x')$. These operators are positive and well-defined by virtue of $q'$ being strictly positive. The operator $T_n$ is doubly substochastic on $L^1(X,\nu)\rightarrow L^1(X',\nu')$ because 
\begin{align}
    \int_{X'} T_n(x',x) {\rm d}\nu'(x') &= \int_{X'} S_n(x',x) {\rm d}\mu'(x') \leq 1 \,,\quad \text{and}\\
    \int_X T_n(x',x) {\rm d}\nu(x) &= \frac{1}{q'(x')} \int_X S_n(x',x) q(x) {\rm d}\mu(x) \leq 1 \,,
\end{align}
where in each line the first equality follows from the definition of the rescaled measure and the second by definition of $(q,q')$-substochastic operators.
Additionally, 
\begin{align}
    \int_X T_n(x',x) \frac{f(x)}{q(x)} {\rm d}\nu(x) = \frac{1}{q'(x')} \int_X S_n(x',x) f(x) {\rm d}\mu(x) \rightarrow \frac{g(x')}{q'(x')} \,.
\end{align}
Therefore, statement 2 is equivalent to the existence of doubly substochastic $T_n$ satisfying $T_n \frac{f}{q}\rightarrow \frac{g}{q'}$ in $L^1(X', \nu')$. 
By the definition of the rescaled measure, Statement 3 is equivalent to $\int_X \phi\left(\frac{f}{q}\right) {\rm d}\nu \geq \int_{X'} \phi\left(\frac{g}{q'}\right) {\rm d}\nu'$. 
Statement 4 is equivalent to $ \int_X \left(\frac{f}{q}-z\right)^+ {\rm d}\nu \geq \int_{X'} \left(\frac{g}{q'}-z\right)^+ {\rm d}\nu' $ and $\int_{X'} \left(\frac{g}{q'}+z\right)^- {\rm d}\nu' \geq \int_X \left(\frac{f}{q}+z\right)^- {\rm d}\nu$, for all $z \geq 0$. This follows from the definition of the rescaled measure and the fact that $q,q'$ are strictly positive.
The equivalences between statements now follow directly from \cref{weakmajor} applied to $\frac{f}{q}$ and $\frac{g}{q'}$ as integrable functions in $L^1(X,\nu)$ and $L^1(X',\nu')$.
\end{proof}

\subsection{Comparison with other definitions of weak majorization}
For nonnegative, discrete distributions with $q,q'=1$, our definition of weak majorization reduces to the standard one for weak submajorization~\cite[Definition 1.A.2]{marshall2011inequalities}. 
However, for signed distributions, our definition does not reduce either to weak submajorization ($\succ_w$) or to weak supermajorization ($\succ^w$).
As identified in Ref.~\cite{marshall2011inequalities}, attempting to use the standard definition of weak submajorization with doubly substochastic matrices trivializes to regular majorization. Their stochastic-operator equivalence is thus restricted to nonnegative distributions~\cite[Theorem 2.C.4]{marshall2011inequalities}. However, they establish equivalences for signed distributions based on convex functions, with the requirement that the convex function be increasing for weak submajorization and decreasing for weak supermajorization~\cite[4.B.2]{marshall2011inequalities}. 
Our definition of weak majorization allows us to prove an equivalence between Lorenz-curve and doubly-substochastic-operator conditions for signed distributions. It can be viewed as a hybrid definition that combines weak submajorization of the positive parts of a distribution with weak supermajorization of the negative parts. Our convex-function equivalence combines both requirements in that the convex function must be increasing over the positive real line and decreasing over the negative real line.

As a simple example to illustrate the different preorders, consider the subnormalized quasiprobability distributions
\begin{align}
    f = \left(\frac{1}{2},\frac{1}{3},\frac{1}{4},\frac{-1}{12}\right)
    \,, \qquad
    g = \left(\frac{1}{3},\frac{1}{3},0,\frac{-1}{2}\right)
    \,, \qquad
    h = \left(\frac{1}{2},\frac{1}{3},\frac{1}{4},{-1}\right)
    \,.
\end{align}
We have that $f \succ_w g$, $f \not \weak g$, and $g \succ^w f$. We also have $h \not \succ_w g$, $h \weak g$, and $h \succ^w g$.

\section{Wigner--R\'enyi divergences}
\label{app:wigner}

\cref{app:wigner1} establishes some useful notation for this section and derives standard preliminary results on the continuity of the Wigner function.
\cref{app:wigner2} proves the first, qualitative part of \cref{thm:semicontinuity} on the lower semicontinuity of the Wigner $\alpha$-R\'enyi divergence, and \cref{app:wigner3} proves the second quantitative part of \cref{thm:semicontinuity} on the H\"older lower semicontinuity, giving the explicit form of the constants that appear in the theorem statement.
Finally, \cref{app:wigner4} establishes that lower semicontinuity cannot be relaxed to continuity.

\subsection{Notation and Wigner function continuity}
\label{app:wigner1}

The $L^\infty$-norm between quasiprobability distributions $w, w'$ on $\mathbb{R}^{2N}$ is
\begin{equation}
    \|w - w'\|_{\infty} 
    \coloneqq \sup_{\bm{r} \in \mathbb{R}^{2N}} |w(\bm{r}) - w'(\bm{r})|
    \,.
\end{equation}
We extend \cref{def:WignerRenyi} of the Wigner $\alpha$-R\'enyi divergence to the case where the reference distribution is positive almost everywhere, but not necessarily the Wigner function of a state. 
Specifically, for any quasiprobability $w$ and reference $q > 0$, the $\alpha$-R\'enyi divergence is
\begin{equation}
    D_\alpha(w \| q) = \frac{1}{\alpha-1} \log J_{\alpha}(w \| q) \,, \quad \alpha \in (1, \infty) \,,
\end{equation}
where we have introduced the exponential form
\begin{equation}\label{eq:Jdef}
    J_{\alpha}(w \| q) 
    \coloneqq \int_{\mathbb{R}^{2N}} \abs{w(\bm{r})}^\alpha q(\bm{r})^{1-\alpha} {\rm d}^{2N}\bm{r}
    = \exp((\alpha-1)D_\alpha(w \| q))
    \,.
\end{equation}

We establish uniform convergence for the Wigner function and pointwise convergence for powers of its absolute value in the trace-norm topology.
\begin{lemma}\label[lemma]{lem:Wbounds}
    Let $\epsilon \ge 0$ and $\rho \in \S(\H_N)$.
    For any $\sigma \in \S(\H_N)$ that satisfies $\frac{1}{2} \| \sigma - \rho \|_1 \le \epsilon$, the following bounds hold,
    \begin{align}
        \|W_\rho - W_{\sigma}\|_{\infty} 
        &\le 2\left(\frac{2}{\pi}\right)^N \epsilon
        \,, \label{eq:boundW1}
        \\
        |W_\rho(\bm{r})|^\alpha - |W_{\sigma}(\bm{r})|^\alpha  
        &\le 2\left(\frac{2}{\pi}\right)^N \alpha |W_\rho(\bm{r})|^{\alpha - 1} \; \epsilon
        \,, \qquad \forall \alpha > 1
        \,. \label{eq:boundW2}
    \end{align}
\end{lemma}
\begin{proof}
The trace distance guarantee directly implies the bound on the $L^\infty$-norm,
\begin{align}
    \|W_\rho - W_{\sigma}\|_{\infty} 
    &= \sup_{\bm{r} \in \mathbb{R}^{2N}} \left| \Tr\left[(\rho-\sigma) \hat{\Delta}_{\rm W}(\bm{r}) \right] \right|
    \nonumber\\
    &\overset{(1)}{\le} \|\rho-\sigma\|_1 \sup_{\bm{r} \in \mathbb{R}^{2N}} \|\hat{\Delta}_{\rm W}(\bm{r})\|_{\infty}
    \nonumber\\
    &\overset{(2)}{=} \left(\frac{2}{\pi}\right)^N \|\rho-\sigma\|_1
    \nonumber\\
    &\le 2\left(\frac{2}{\pi}\right)^N \epsilon
    \,,
\end{align}
where $(1)$ is H\"older's inequality, and $(2)$ follows from our definition of the displaced parity operators in \cref{eq:Delta}.

Using this bound, we can also obtain the second bound,
\begin{align}
    |W_\rho(\bm{r})|^\alpha - |W_{\sigma}(\bm{r})|^\alpha
    &\overset{(1)}{\le} \alpha |W_\rho(\bm{r})|^{\alpha - 1} \; \big( |W_{\rho}(\bm{r})| - |W_{\sigma}(\bm{r})| \big)
    \nonumber\\
    &\overset{(2)}{\le} \alpha |W_\rho(\bm{r})|^{\alpha - 1} \big| |W_{\rho}(\bm{r})| - |W_{\sigma}(\bm{r})| \big|
    \nonumber\\
    &\overset{(3)}{\le} \alpha |W_\rho(\bm{r})|^{\alpha - 1} \big| W_{\rho}(\bm{r}) - W_{\sigma}(\bm{r}) \big|
    \nonumber\\
    &\overset{(4)}{\le} \alpha |W_\rho(\bm{r})|^{\alpha - 1} \| W_{\rho} - W_{\sigma} \|_{\infty}
    \nonumber\\
    &\overset{(5)}{\le} 2\left(\frac{2}{\pi}\right)^N \alpha |W_\rho(\bm{r})|^{\alpha - 1} \; \epsilon
    \,, \label{eq:convexity|W|}
\end{align}
where $(1)$ follows from the first-order characterization of convexity $f(b) - f(a) \le f'(b)(b-a)$ applied to the convex function $f(t) = t^\alpha$ with $\alpha > 1$, $(2)$ from $b-a \le |b-a|$,  $(3)$ is the reverse triangle inequality, $(4)$ follows from the definition of the $L^\infty$-norm, and $(5)$ from \cref{eq:boundW1}.
\end{proof}

\subsection{Proof of lower semicontinuity}
\label{app:wigner2}

We prove the lower semicontinuity of the Wigner $\alpha$-R\'enyi divergence with $\alpha > 1$, as stated in part \emph{(i)} of \cref{thm:semicontinuity}, which we restate here for convenience.
\begingroup
\renewcommand{\thetheorem}{2(i)}
\begin{theorem}\label{thm:semicontinuity1}
Let $\tau \in \S(\H_N)$ be a reference state such that $W_\tau > 0$.
Then, the functional $\rho \mapsto D_\alpha(W_\rho \| W_\tau)$ is lower semicontinuous on $\S(\H_N)$ in the trace-norm topology.
In particular, for any sequence of states $(\rho_n)_{n \in \mathbb{N}} \subset \S(\H_N)$ converging to $\rho \in \S(\H_N)$ in trace norm, i.e.
\begin{equation}
    \lim_{n \to \infty} \frac{1}{2} \|\rho_n - \rho\|_1 = 0
    \,,
\end{equation}
we have
\begin{equation}\label{eq:liming_app}
    \liminf_{n \to \infty} D_\alpha(W_{\rho_n} \| W_\tau) \ge D_\alpha(W_\rho \| W_\tau)
    \,.
\end{equation}
\end{theorem}
\addtocounter{theorem}{-1} %
\endgroup
\begin{proof}
Since $W_\tau > 0$, we can define the integrand sequence $(f_n(\bm{r}))_{n \in \mathbb{N}}$, where
\begin{equation}
    f_n(\bm{r}) \coloneqq |W_{\rho_n}(\bm{r})|^\alpha W_\tau(\bm{r})^{1-\alpha}
    \ge 0
    \,,
\end{equation}
for all $\bm{r} \in \mathbb{R}^{2N}$.
\cref{lem:Wbounds} directly implies the pointwise convergence
\begin{equation}
    \lim_{n \to \infty} f_n(\bm{r}) 
    = f(\bm{r}) 
    \coloneqq |W_\rho(\bm{r})|^\alpha W_\tau(\bm{r})^{1-\alpha}
    \ge 0
    \,,
\end{equation}
for all $\bm{r} \in \mathbb{R}^{2N}$.
By Fatou's lemma~\cite{fatou1906series},
\begin{align}
    \liminf_{n \to \infty} J_\alpha(W_{\rho_n} \parallel W_\tau)
    &= \liminf_{n \to \infty} \int_{\mathbb{R}^{2N}}  f_n(\bm{r}) {\rm d}^{2N}\bm{r}
    \nonumber\\
    &\ge \int_{\mathbb{R}^{2N}}  \liminf_{n \to \infty} f_n(\bm{r}) {\rm d}^{2N}\bm{r}
    = J_\alpha(W_\rho \parallel W_\tau)
    \,. \label{eq:fatouJ}
\end{align}
By Jensen's inequality~\cite{jensen1906fonctions} on $\mathbb{R}^{2N}$ with measure $W_\tau(\bm{r}) {\rm d}\bm{r}$, we have
\begin{equation}
    J_\alpha(W_\rho \parallel W_\tau)
    = \int_{\mathbb{R}^{2N}} W_\tau(\bm{r})  \left(\frac{|W_\rho(\bm{r})|}{W_\tau(\bm{r})}\right)^\alpha {\rm d}^{2N}\bm{r}
    \ge \left(\int_{\mathbb{R}^{2N}}  |W_\rho(\bm{r})| {\rm d}^{2N}\bm{r} \right)^\alpha 
    \ge 1
    \,.
\end{equation}
For $\alpha > 1$, the function $t \mapsto \frac{1}{\alpha-1} \log t$ is strictly increasing and continuous on $[1, \infty)$, so it commutes with the limit inferior.
Therefore, applying the function on both sides of \cref{eq:fatouJ} gives \cref{eq:liming_app}.
\end{proof}

In fact, any generalized divergence formed according to our majorization equivalences is lower semicontinuous by an application of Fatou's lemma. This includes Wigner negativity (the relevant quantity in the $\alpha\rightarrow 1$ limit) as well as power divergences~\cite{Joe1990}.

\subsection{Proof of one-sided H\"older continuity bound}
\label{app:wigner3}

We derive a one-sided H\"older continuity bound for Wigner $\alpha$-R\'enyi divergences with $\alpha > 1$ under additional mild assumptions. 
This constitutes part \emph{(ii)} of \cref{thm:semicontinuity}, which we restate here, and include the explicit form of the constant $\Delta$.
\begingroup
\renewcommand{\thetheorem}{2(ii)}
\begin{theorem}\label{thm:semicontinuity2}
    Let $\tau \in \S(\H_N)$ be a Gaussian state with displacement vector $\bmd$ and covariance matrix $V$, and $\rho \in \S(\H_N)$ be a state that satisfies the finite tail-margin assumption
    \begin{equation}\label{Dkappa_assumption}
        D_{\alpha}(W_\rho \| q_\kappa) < \infty
        \,, \qquad
        q_\kappa(\bm{r}) \coloneqq \exp(-2\kappa \|\bm{r} - \bmd\|_2^2) \; W_\tau(\bm{r})
        \,,
    \end{equation}
    for some $\kappa > 0$.
    Then, every quantum state $\tilde{\rho}$ satisfying
    $\frac{1}{2} \|\tilde{\rho}-\rho\|_1 \le \epsilon$ obeys
    \begin{equation}\label{eq:semicontinuity_app}
        D_\alpha(W_{\tilde{\rho}}\|W_\tau)
        \ge
        D_\alpha(W_\rho\|W_\tau) - \frac{\log(1/(1-\Delta\epsilon^c))}{\alpha-1}
        \,, \qquad
        \forall \epsilon \in \left[ 0, \Delta^{-\frac{1}{c}} \right)
        \,,
    \end{equation}
    where
    \begin{align}
        c &\coloneqq \frac{\kappa \alpha}{\kappa \alpha + \|V^{-1}\|_{\infty}}
        \in (0, 1)
        \,, \label{eq:c}\\
        \Delta 
        &\coloneqq \left( \frac{2\alpha(\det V)^{\frac{\alpha-1}{2\alpha}}}{(\alpha-1)^{\frac{2N}{\alpha}} \kappa^{\frac{N}{\alpha}}} \right)^c \frac{J_{\alpha}(W_\rho \| q_\kappa)^{1-\frac{c}{\alpha}}}{J_{\alpha}(W_\rho \| W_\tau)}
        > 0
        \,. \label{eq:Delta_const}
    \end{align}
    In particular, for any $E \in (0,1)$, the Wigner $\alpha$-R\'enyi divergence is H\"older lower semicontinuous at $W_\rho$,
    \begin{equation}\label{eq:semicontinuity2_app}
        D_\alpha(W_{\tilde{\rho}}\|W_\tau)
        \ge
        D_\alpha(W_\rho\|W_\tau) - \frac{\log(1/(1-E))}{(\alpha-1)E}\Delta \epsilon^c
        \,, \qquad
        \forall \epsilon \in \left[ 0, \left(E\Delta^{-1}\right)^{\frac{1}{c}} \right]
        \,.
    \end{equation}
    Furthermore, in the limit $\kappa \to 0^+$, we have $\Delta \to 1$, $c \to 0^+$, and $\Delta^{-\frac{1}{c}} \to 0^+$.
\end{theorem}
\addtocounter{theorem}{-1}
\endgroup
\begin{proof}
We seek to control how much $D_\alpha(W_{\tilde{\rho}} \| W_{\tau})$ can fall below $D_\alpha(W_\rho \| W_{\tau})$.
We will instead work with the exponential forms of the divergences defined in \cref{eq:Jdef}.
Then, the assumption in \cref{Dkappa_assumption} becomes $J_{\alpha}(W_\rho \| q_\kappa) < \infty$ for some $\kappa > 0$.
Letting $c \in (0,1)$, we obtain the bound
\begin{align}
    J_{\alpha}(W_{\tilde{\rho}} \| W_\tau) - J_{\alpha}(W_\rho \| W_\tau)
    &= \int_{\mathbb{R}^{2N}}   \Big( \abs{W_{\tilde{\rho}}(\bm{r})}^\alpha - \abs{W_\rho(\bm{r})}^\alpha \Big) \; W_{\tau}(\bm{r})^{1-\alpha} {\rm d}^{2N}\bm{r}
    \nonumber\\
    &\overset{(1)}{\ge} -\int_{\mathbb{R}^{2N}}  \min\left\{ \abs{W_\rho(\bm{r})}^\alpha \,,\ 2\left(\frac{2}{\pi}\right)^N \alpha |W_\rho(\bm{r})|^{\alpha - 1} \; \epsilon \right\} \; W_{\tau}(\bm{r})^{1-\alpha} {\rm d}^{2N}\bm{r}
    \nonumber\\
    &\overset{(2)}{\ge} -\left(\frac{2}{\pi}\right)^{Nc} (2\alpha\epsilon)^c \int_{\mathbb{R}^{2N}}  \abs{W_\rho(\bm{r})}^{\alpha-c} \; W_{\tau}(\bm{r})^{1-\alpha} {\rm d}^{2N}\bm{r}
    \nonumber\\
    &\overset{(3)}{\ge} -\left(\frac{2}{\pi}\right)^{Nc} (2\alpha\epsilon)^c J_{\alpha}(W_\rho \| q_\kappa)^{1-\frac{c}{\alpha}} \left( \int_{\mathbb{R}^{2N}}  e^{2\kappa (1-\alpha)\left(\alpha c^{-1} - 1\right) \|\bm{r} - \bmd\|_2^2} \; W_{\tau}(\bm{r})^{1-\alpha} {\rm d}^{2N}\bm{r} \right)^{\frac{c}{\alpha}} 
    \,, \label{eq:Jbound_temp}
\end{align}
where $(1)$ follows from \cref{lem:Wbounds}, $(2)$ from $\min\{u,v\} \le u^{1-c}v^c$, which holds for all $u,v\ge0$ and $c \in (0,1)$, and $(3)$ is H\"older's inequality for the conjugate exponents $\alpha/(\alpha-c)$ and $\alpha/c$.

To bound the remaining integral, we recall that the Wigner function of Gaussian state $\tau$ takes the form
\begin{equation}
    W_\tau(\bm{r}) = \left(\frac{2}{\pi}\right)^N (\det V)^{-\frac{1}{2}} \exp(-2(\bm{r} - \bmd)^{\rm T} V^{-1} (\bm{r} - \bmd))
    \,.
\end{equation}
Then, using our choice for $c$ in \cref{eq:c}, we obtain

\begin{align}
    \int_{\mathbb{R}^{2N}}  e^{2\kappa (1-\alpha)\left(\alpha c^{-1} - 1\right) \|\bm{r} - \bmd\|_2^2} \; W_{\tau}(\bm{r})^{1-\alpha} {\rm d}^{2N}\bm{r}
    &\overset{(1)}{\le}
    \left(\frac{\pi}{2}\right)^{N(\alpha-1)} (\det V)^{\frac{\alpha-1}{2}} \int_{\mathbb{R}^{2N}}  e^{-2(\alpha-1)^2 \kappa \|\bm{r} - \bmd\|_2^2} {\rm d}^{2N}\bm{r} 
    \nonumber\\
    &\overset{(2)}{=}
    \left( \frac{\pi}{2} \right)^{N\alpha} \frac{(\det V)^{\frac{\alpha-1}{2}}}{(\alpha-1)^{2N}\kappa^N}
    \,,
\end{align}
where $(1)$ follows from the bound $\bm{x}^{\rm T} V^{-1} \bm{x} \le \|V^{-1}\|_{\infty} \|\bm{x}\|_2^2$ for all $\bm{x} \in \mathbb{R}^{2N}$, and in $(2)$ we explicitly evaluated the Gaussian integral.
Substituting this into \cref{eq:Jbound_temp} and using our choice for $\Delta$ in \cref{eq:Delta_const} leads to
\begin{equation}
    J_{\alpha}(W_{\tilde{\rho}} \| W_\tau) \ge J_{\alpha}(W_\rho \| W_\tau) \left( 1 - \Delta \epsilon^c \right)
    \,.
\end{equation}
When $\epsilon < \Delta^{-\frac{1}{c}}$, we can take logarithms to arrive at \cref{eq:semicontinuity_app}.
\cref{eq:semicontinuity2_app} follows from the secant inequality
\begin{equation}\label{eq:secant}
    \log(1-x) \ge \frac{\log(1-E)}{E} x
    \,, \qquad 
    E \in (0,1)
    \,,
\end{equation}
which holds for all $x \in [0, E]$.

Finally, in the limit of $\kappa \to 0^+$, it is clear from \cref{eq:c} that $c \to 0^+$.
Then, in \cref{eq:Delta_const} we get
\begin{equation}
    \kappa^{-Nc/\alpha} 
    = \exp(-\frac{N \kappa \log(\kappa)}{\kappa \alpha + \|V^{-1}\|_{\infty}}) 
    \to 1
    \,,
\end{equation}
and $J_{\alpha}(W_\rho \| q_\kappa)^{1-\frac{c}{\alpha}} \to J_{\alpha}(W_\rho \| W_\tau)$, resulting in $\Delta \to 1$.
The limit $\Delta^{-\frac{1}{c}}$ is of indeterminate form, so we need to use the explicit forms from \cref{eq:c,eq:Delta_const} to evaluate it, giving
\begin{equation}
    \Delta^{-\frac{1}{c}}
    = \frac{(\alpha-1)^{\frac{2N}{\alpha}}}{2\alpha(\det V)^{\frac{\alpha-1}{2\alpha}}} \frac{J_{\alpha}(W_\rho \| W_\tau)^{\frac{1}{c}}}{J_{\alpha}(W_\rho \| q_\kappa)^{\frac{1}{c}-\frac{1}{\alpha}}}  \kappa^{\frac{N}{\alpha}}
    \to 0^+
    \,,
\end{equation}
To see this, we consider the first-order expansion
\begin{equation}
    J_{\alpha}(W_\rho \| q_\kappa) 
    = J_{\alpha}(W_\rho \| W_\tau) + 2(\alpha-1) J_0 \kappa + {\cal O}(\kappa^2)
    \,,
\end{equation}
where we have defined
\begin{equation}
    J_0 \coloneqq \int_{\mathbb{R}^{2N}}  |W_\rho(\bm{r})|^\alpha W_\tau(\bm{r})^{1-\alpha} \|\bm{r} - \bmd\|_2^2 {\rm d}^{2N}\bm{r} < \infty
    \,.
\end{equation}
Then, using $1/c = 1 + \frac{\|V^{-1}\|_\infty}{\kappa\alpha}$ and the Euler limit, the middle factor tends to a positive constant,
\begin{align}
    \lim_{\kappa \to 0^+} \frac{J_{\alpha}(W_\rho \| W_\tau)^{\frac{1}{c}}}{J_{\alpha}(W_\rho \| q_\kappa)^{\frac{1}{c}-\frac{1}{\alpha}}}
    &= \lim_{\kappa \to 0^+} J_{\alpha}(W_\rho \| q_\kappa)^{\frac{1}{\alpha}} \times \lim_{\kappa \to 0^+} \left( \frac{J_{\alpha}(W_\rho \| W_\tau)}{J_{\alpha}(W_\rho \| W_\tau) + 2(\alpha-1) J_0 \kappa} + {\cal O}(\kappa^2) \right)^{\frac{1}{c}}
    \nonumber\\
    &= J_{\alpha}(W_\rho \| W_\tau)^{\frac{1}{\alpha}} \times \lim_{\kappa \to 0^+} \left( 1 - 2(\alpha-1) \frac{J_0}{J_{\alpha}(W_\rho \| W_\tau)} \kappa + {\cal O}(\kappa^2) \right)^{1 + \frac{\|V^{-1}\|_\infty}{\kappa\alpha}}
    \nonumber\\
    &= J_{\alpha}(W_\rho \| W_\tau)^{\frac{1}{\alpha}} \exp(-\frac{2(\alpha-1) \|V^{-1}\|_\infty J_0}{\alpha J_{\alpha}(W_\rho \| W_\tau)})
    > 0
    \,.
\end{align}
Hence, $\Delta^{-\frac{1}{c}} \to 0^+$ because of the vanishing factor $\kappa^{\frac{N}{\alpha}}$.
\end{proof}

\subsection{Counterexample to continuity}
\label{app:wigner4}

Here, we establish that lower semicontinuity cannot be relaxed to a continuity statement, even if we restrict to sequences of states with bounded moments.
\begin{lemma}\label[lemma]{thm:discontinuity}
    Consider a 1-mode state $\rho \in \S(\H_1)$ with bounded moments, $\Tr\left( \rho (a^\dag a)^M \right) < C_M$ for all $M \in \mathbb{N}$. Suppose that for some $\alpha \in (1, \infty)$ it has finite Wigner $\alpha$-R\'enyi divergence $D_\alpha(W_{\rho} \| W_{\tau}) < \infty$ relative to a thermal state $\tau \in \S(\H_1)$.
    Then, there exists a sequence of states $(\rho_n)_{n \in \mathbb{N}} \subset \S(\H_1)$ that satisfy the following:
    \begin{enumerate}
        \item All statistical moments are bounded, $\sup_{n \in \mathbb{N}} \Tr\left(\rho_n (a^\dag a)^M\right) < \infty$ for all $M \in \mathbb{N}$;
        \item The trace distance with $\rho$ vanishes, $\lim_{n \to \infty} \frac{1}{2} \| \rho_n - \rho \|_1 = 0$;
        \item The Wigner $\alpha$-R\'enyi divergence relative to the thermal state $\tau$ diverges,
        \begin{equation}\label{eq:divdiv}
            \lim_{n \to \infty} \left| D_\alpha(W_{\rho_n} \| W_{\tau}) - D_\alpha(W_{\rho} \| W_{\tau}) \right| = \infty
            \,.
        \end{equation}
    \end{enumerate}
\end{lemma}
\begin{proof}
    Let the mean photon number of state $\tau$ be $\bar{n}$, so that its Wigner function is
    \begin{equation}
         W_{\tau}(x,p) 
         = \frac{2\lambda}{\pi} \exp(-2\lambda(x^2 + p^2)) 
         \,,
    \end{equation}
    where $\lambda \coloneqq (2\bar{n}+1)^{-1}$.
    Consider a sequence of new thermal states $(\tau_n)_{n \in \mathbb{N}} \subset \S(\H_1)$ with
    \begin{equation}
        \tau_n = \sum_{k=0}^\infty \frac{\bar{n}_n^k}{(\bar{n}_n + 1)^{k+1}} \ketbra{k} 
        \,, \quad 
        \bar{n}_n = \bar{n} + \frac{1}{2\lambda(\alpha-1)} \left(1 - \frac{1}{\sqrt{n}}\right) \ge 0
        \,.
    \end{equation}
    Setting $\lambda_n \coloneqq (2\bar{n}_n+1)^{-1}$, its Wigner function is
    \begin{equation}
         W_{\tau_n}(x,p) 
         = \frac{2\lambda_n}{\pi} \exp(-2\lambda_n(x^2 + p^2)) 
         \,.
    \end{equation}

    We now construct the sequence $(\rho_n)_{n \in \mathbb{N}} \subset \S(\H_1)$ by defining
    \begin{equation}
        \rho_n = (1-w_n) \rho + w_n \tau_n
        \,, \quad
        w_n = n^{-c}
        \,, \quad
        c \in \left(0, \frac{1}{2\alpha}\right)
        \,.
    \end{equation}
    Its moments are strictly bounded as, for all $n, M \in \mathbb{N}$ and $\alpha \in (1,\infty)$, %
    \begin{equation}
        \Tr\left(\rho_n (a^\dag a)^M\right) 
        = (1-w_n) \Tr\left(\rho (a^\dag a)^M \right) + w_n \Tr\left(\tau_n (a^\dag a)^M\right)
        < C_M + M! (\bar{n}_\star + 1)^M
        < \infty
        \,,
    \end{equation}
    where $\bar{n}_\star \coloneqq \bar{n} + 1/(2\lambda(\alpha-1))$.
    The trace distance with $\rho$ vanishes,
    \begin{equation}
        \frac{1}{2} \| \rho_n - \rho \|_1 = w_n \frac{1}{2} \| \tau_n - \rho \|_1 \le w_n = n^{-c}
        \,.
    \end{equation}
    To prove the statement regarding the Wigner $\alpha$-R\'enyi divergence, we relate it to the weighted norm on the space of Wigner functions with measure ${\rm d}\mu \coloneqq W_{\tau}(x,p)^{1-\alpha} \; {\rm d}x \; {\rm d}p$ defined as follows,
    \begin{equation}
        \|W\|_{\alpha, \mu} \coloneqq \left( \int_{\mathbb{R}^2} |W(x,p)|^\alpha \; W_{\tau}(x,p)^{1-\alpha} \; {\rm d} x \; {\rm d} p \right)^{1/\alpha}
        = \exp(\frac{\alpha-1}{\alpha} D_\alpha( W \| W_{\tau}))
        \,.
    \end{equation}
    Using the triangle inequality on this space, we obtain
    \begin{equation}\label{eq:wsigman}
        \|W_{\rho_n}\|_{\alpha, \mu} = \|(1-w_n)W_\rho + w_n W_{\tau_n}\|_{\alpha, \mu} \ge w_n \|W_{\tau_n}\|_{\alpha, \mu} - (1-w_n) \|W_\rho\|_{\alpha, \mu}\
        \,.
    \end{equation}
    Next, we evaluate
    \begin{align}
        \|W_{\tau_n}\|_{\alpha, \mu}^\alpha 
        &= \int_{\mathbb{R}^2} |W_{\tau_n}(x,p)|^\alpha \; W_{\tau}(x,p)^{1-\alpha} \; {\rm d} x \; {\rm d} p
        \nonumber\\
        &= \frac{2\lambda}{\pi} \left(\frac{\lambda_n}{\lambda}\right)^\alpha \int_{\mathbb{R}^2} \exp(-2(\alpha\lambda_n + (1 - \alpha)\lambda)(x^2 + p^2)) {\rm d} x {\rm d} p
        \nonumber\\
        &= \frac{\lambda}{\alpha \lambda_n + (1 - \alpha)\lambda} \left(\frac{\lambda_n}{\lambda}\right)^\alpha
        \nonumber\\
        &= \left( \frac{\alpha-1}{\alpha - n^{-1/2}} \right)^{\alpha-1} n^{1/2}
        \,, \label{eq:wtaun}
    \end{align}
    where the last two equalities follow from the Gaussian integral $\frac{2}{\pi} \int_{\mathbb{R}^2} e^{-2A (x^2 + p^2)} \; {\rm d}x \; {\rm d}p = \frac{1}{A}$ for all $A > 0$, where here
    \begin{equation}
        A = \alpha\lambda_n + (1 - \alpha)\lambda = \frac{\alpha - 1}{\alpha\sqrt{n} - 1}\lambda > 0
        \,,
    \end{equation}
    Then, combining \cref{eq:wsigman,eq:wtaun}, we obtain
    \begin{equation}
        \|W_{\rho_n}\|_{\alpha, \mu} 
        \ge w_n \|W_{\tau_n}\|_{\alpha, \mu} - (1 - w_n) \|W_\rho\|_{\alpha, \mu}
        = \left( \frac{\alpha-1}{\alpha - n^{-1/2}} \right)^{\frac{\alpha-1}{\alpha}} n^{\frac{1}{2\alpha} - c} - (1 - n^{-c}) \|W_\rho\|_{\alpha, \mu}
        \,,
    \end{equation}
    which diverges as $n \to \infty$ for all $\alpha \in (1, \infty)$ because of the assumptions $c \in \big(0, 1/(2\alpha)\big)$ and $D_\alpha(W_{\rho} \| W_{\tau}) < \infty$ which implies $\|W_\rho\|_{\alpha, \mu} < \infty$.
    Consequently, $D_\alpha(W_{\rho_n} \| W_{\tau})$ diverges, and, since $D_\alpha(W_{\rho} \| W_{\tau}) < \infty$, we obtain \cref{eq:divdiv}, concluding the proof.
\end{proof}

\section{Derivation of certified Lorenz envelopes}
\label{app:robust}

In this appendix we detail various sources of numerical error and explain how the envelopes account for them. To compute robust curves, we require additional regularity information beyond just the values of the Wigner function. Our procedure is agnostic to how that additional information is computed. We provide some generic ways to compute this information using higher-order derivatives, but different approaches may give better results for certain Wigner functions. Any method is compatible with our construction (see also, consistency bounds below). A trick we use repeatedly is pointwise and cellwise dominance: Given two functions $f,g$ such that ${\rm sgn}(f) = {\rm sgn}(g)$ and $\abs{f}\geq\abs{g}$, the area enclosed by the positive and negative Lorenz curves of $f$ contains the Lorenz curves of $g$. %
And if a piecewise constant function majorizes or is majorized by $f$ on each cell, then its piecewise linear Lorenz curves similarly bound those of $f$. We first derive the error region for a single-mode function, then show how to extend this to the multimode case. Our envelopes are the Lorenz curves of piecewise constant functions. Hence, they are exactly piecewise linear.
We use $\concat$ to denote concatenation to the list on the left-hand side.

\subsection{Phase-space cutoff}
First consider the error from integrating over a finite region of phase space, discarding the function outside.
In our numerics, we integrate over a square centered at the origin. Let $f_\text{c}$ be the cutoff function that equals $f$ in the region and zero elsewhere and let the absolute values of the dropped positive and negative masses be $w^+_{f}$ and $w_{f}^-$.
The cutoff function's Lorenz curves are automatically inner bounds because signed mass is dropped from $f$.
In the worst case, the missing masses are on a region with arbitrarily small measure, so that the true Lorenz curves are, in the limit of smaller and smaller measure, shifted vertically from the cutoff ones by $w_f^{\pm}$. Note this argument holds regardless of what $q$ is. In summary the Lorenz envelopes due to cutoff are
\begin{align}
     \qdeclor{f_\text{c}} \leq \qdeclor{f}\leq \qdeclor{f_\text{c}}+w_f^+\\
     \qinclor{f_\text{c}} \geq \qinclor{f}\geq \qinclor{f_\text{c}}-w_f^-.
\end{align}
Recall our formalism is restricted to functions with finite negative volume so both masses $w_f^{\pm}$ are finite. However, they could in principle be arbitrarily large. A generic bound, even e.g. assuming Wigner functions of states with finite energy, is generally not possible. %
Instead a case-by-case analysis is needed to bound the discarded weights for a specific $f$.

\subsection{Integration error}
Now we work with the cutoff function and account for the error from a finite sampling resolution. In our numerics, we sample over a square grid.

Within each cell, we first compute an upper and lower bound on the function values and on its integral over the cell. This can be done according to any procedure, and for a specific $f$ a specialized approach may give the best results. We present here a generic method using second-order information about $f$. For the functions and numerical parameters we typically consider, using second-order information is necessary to get reasonable error bars.

Assume we have a function for $f$, for the gradient of $f$, and for the second derivatives of $f$. Each square cell $C_i$ has side length $\Delta$ and midpoint $a_i$.
We first bound how much $f$ can change. Let $E^i_{xx},E^i_{xp},E^i_{pp}$ be upper bounds on $\abs{f_{xx},f_{xp},f_{pp}}$ respectively in cell $i$. By the second-order multivariate Taylor's theorem, 
\begin{equation}
    \abs{f(x)-f(a_i)} \leq \frac{\Delta}{\sqrt{2}} \norm{\grad{f}|_{a_i}}_2 + \frac{\Delta^2}{2} \max \{ E^i_{xx},E^i_{xp},E^i_{pp} \},
\end{equation}
for all $x\in C_i$, where the notation in the first term indicates the $2$-norm of $\grad f$ evaluated at $a_i$. %
Next, we use the midpoint rule estimate to obtain a second-order estimate on the integration error
\begin{equation}
    \abs{\int_{C_i}f {\rm d}x{\rm d}p -\Delta^2 f(a_i)} \leq \Delta^2 \frac{\Delta^2}{24} (E^i_{xx}+E^i_{pp}).
\end{equation}
In addition to the bounds on $f$, bounding the integral helps us refine our envelopes. We summarize the lists of bounds for each cell,
\begin{align}
    \fmax&= f(a_i)+\frac{\Delta}{\sqrt{2}} \norm{\grad{f}|_{a_i}}_2 + \frac{\Delta^2}{2} \max \{ \abs{E^i_{xx}},\abs{E^i_{xp}},\abs{E^i_{pp} }\}\\
    \fmin&= f(a_i) - \frac{\Delta}{\sqrt{2}} \norm{\grad{f}|_{a_i}}_2 - \frac{\Delta^2}{2}\max \{ \abs{E^i_{xx}},\abs{E^i_{xp}},\abs{E^i_{pp} }\}\\
    \fu&=\min\{\Delta^2 f(a_i)+ \Delta^2 \frac{\Delta^2}{24} (\abs{E^i_{xx}}+\abs{E^i_{pp}}),\fmax \Delta^2\}\\
    \fl&=\max\{\Delta^2 f(a_i)- \Delta^2 \frac{\Delta^2}{24} (\abs{E^i_{xx}}+\abs{E^i_{pp}}),\fmin \Delta^2 \}.
\end{align}
By taking the maximum or minimum, we have imposed a consistency constraint so that the integral bounds are at least as good as those derived from the function value bounds. This is strictly unnecessary based on the specific bounds we're using, but if a different procedure were used to derive pointwise and area bounds, then this consistency condition is assumed to have been enforced in the subsequent discussion. It is also implicit that $\fmax\geq\fmin$ and  $\fu \geq \fl $, but to be clear, this needs to be enforced if using some other procedure to derive bounds.

\subsubsection{Regular majorization}

Now assuming we have consistent upper and lower bounds on the function and its integral in each cell, derived in any manner, we derive the envelopes for regular majorization. Consider nonnegative $f$ for illustration. The most ordered function respecting the function-value constraints takes the value $\fmax$ over as large an area as possible (the specific region is irrelevant for computing its Lorenz curve, only the total area is important) and then stays at $\fmin$ until the total mass limit is hit. By construction this function weakly majorizes any other function respecting the upper and lower bounds on its values and average over the cell. %
The least ordered function is one that is constant. In this case, by the constraint we can set it to take the value $\fl/\Delta^2$ over the whole cell (in this case of nonnegative $f$, $\fmin$ can be taken to be nonnegative and this forces $\fl$ to also be nonnegative). %

Now consider possibly negative $f$. If $\fmax$ is positive, then the function with the most ordered positive Lorenz curve is as described above, but the fact that $\fmin$ could be negative means the most ordered function could achieve its maximum value over a larger region with the mass canceled out by negative values elsewhere. So we need to compute what that maximum area is. Define areas $A_1$ and $A_2$ where the bounding function achieves $\fmax$ and $\fmin$ respectively, which must obey
\begin{align}\label{areasum}
    A_1+A_2 = \Delta^2.
\end{align}
The areas are also constrained by
\begin{align}\label{areaweightpos}
    A_1 \fmax + A_2 \fmin \leq \fu.
\end{align}
Substituting the first constraint into the second one, we can upper bound $A_1$ by 
\begin{align}
A_1 (\fmax - \fmin) + (A_2+A_1) \fmin  \leq \fu &\implies A_1 (\fmax - \fmin) \leq \fu - \Delta^2 \fmin\\
&\implies A_1  \leq \frac{\fu - \Delta^2 \fmin}{\fmax - \fmin} \eqqcolon \amaxp,\label{a1max}
\end{align}
where if $\fmax=\fmin$, then $\fu=\Delta^2 \fmax$ and the indeterminate fraction should be read as $\Delta^2$. Because $\Delta^2 \fmin \leq \fu \leq \Delta^2 \fmax$, the bound is always in $[0,\Delta^2]$, i.e. never worse than the total cell area.
So the list of points that leads to the outer positive bound gets two contributions from each cell,
\begin{align}
Z_\text{O,P}&=\{ \left(A_1^\text{max} \max\{\fmax ,0\} , A_1^\text{max} \right), \left( [\Delta^2-A_1^\text{max}] \max\{\fmin ,0\} , [\Delta^2-A_1^\text{max}] \right)\}.
\end{align}
The $\max$ simply ensures that $\fmax, \fmin$ don't contribute to the positive bound if they are negative. 
For the outer negative list, we repeat an analogous analysis but now upper bound $A_2$. Assume $\fmin$ is negative, if not there is no contribution to the negative list. The relevant constraint is now
\begin{equation}\label{areaweightneg}
    A_1 \fmax + A_2 \fmin \geq \fl.
\end{equation}
The bound on $A_2$ is then
\begin{align}
    A_2(\fmin - \fmax) + \Delta^2 \fmax \geq \fl &\implies A_2(\fmin - \fmax) \geq \fl- \Delta^2 \fmax\\
    &\implies A_2 \leq \frac{\fl- \Delta^2 \fmax}{\fmin - \fmax}=\frac{\Delta^2 \fmax-\fl}{\fmax-\fmin}\eqqcolon \amaxn.
\end{align}
Again let this bound be $\Delta^2$ when $\fmax=\fmin$ and denote it by $A_2^\text{max}$. The outer negative list gets two contributions for each cell,
\begin{align}
     Z_\text{O,N}&=\{ \left(A_2^\text{max} \min\{\fmin,0\},A_2^\text{max}\right), \left([\Delta^2-A_2^\text{max}] \min\{\fmax,0\},[\Delta^2-A_2^\text{max}]\right)   \}. 
\end{align}
Corresponding to a constant bounding function taking the minimum possible value, the inner positive list is
\begin{equation}
    Z_\text{I,P}= \{ \left(\max\{\fl,0\}, \Delta^2 \right) \}. 
\end{equation}
Similarly, the inner negative list is 
\begin{equation}
    Z_\text{I,N} = \{ \left(\min\{\fu,0\}, \Delta^2\right) \}. 
\end{equation}

In summary, we have defined piecewise constant functions which majorize or are majorized by the positive and negative parts of $f$. This guarantees that their Lorenz curves bound the curves of the cutoff function: 
\begin{align}
\declor{Z_\text{I,P}} \leq &\declor{f_\text{c}} \leq \declor{Z_\text{O,P}} \\
\inclor{Z_\text{I,N}} \geq &\inclor{f_\text{c}} \geq \inclor{Z_\text{O,N}}. 
\end{align}

\subsubsection{Relative majorization}

Now consider relative majorization and the associated errors from sampling $q$ too over the discrete grid. Assume in a similar manner that we have consistent nonzero bounds $\qmax,\qmin,\qu,\ql>0$. For example, these bounds can be derived using the first and second-derivative procedure as described for $f$. We will use that relative majorization reduces to regular majorization of $f/q$ with the measure ${\rm d}\nu = q {\rm d}\mu$. Pointwise upper and lower bounds on $f/q$ are $\fmax/q$ and $\fmin/q$. For relative majorization, the outer bounds are formed by concatenating the lists generated in three cases.

\vspace{10pt}

\textbf{Case 1.} $\fmin\geq 0$.

By the same logic as above, the most ordered function takes the maximum value over some portion of the cell and the minimum value elsewhere. These values are $\fmax/\qmin$ and $\fmin/\qmax$. The most-ordered function accumulates its mass over the smallest measure, so we use the lower bound $\ql$. We then solve for the maximum area over which the bounding function can be at its maximum value without the total mass exceeding $\fu$. Define this value as
\begin{equation}
Q^* = \frac{\fu - \frac{\fmin}{\qmax} \ql}{\frac{\fmax}{\qmin} - \frac{\fmin}{\qmax}}   
\end{equation}
Then, the outer bound list is
\begin{align}
Z_\text{O,P}&\concat\left\{ \left(\frac{\fmax}{\qmin}Q^*,Q^*\right) , \left( \frac{\fmin}{\qmax} (\ql - Q^*),\ql-Q^*  \right)\right\}.
\end{align}
Nothing is added to the $Z_\text{O,N}$ list.

\vspace{10pt}

\textbf{Case 2.} $\fmax>0, \fmin< 0$.

We find a bounding function in two steps. Again, $f_c^+$ is weakly majorized (relative to $q$) by the function which equals $\fmax/q$ over as large a region as possible and $\fmin/q$ elsewhere. Let these two regions be $R_1,R_2$. $R_1$ is constrained by
\begin{equation}
    \fu \geq \int_{R_1} \frac{\fmax}{q} q {\rm d}x{\rm d}p + \int_{R_2} \frac{\fmin}{q} q {\rm d}x{\rm d}p = A_1 \fmax + A_2 \fmin .
\end{equation}
And we have that $A_1+A_2=\Delta^2$ so \cref{a1max} is still an upper bound on $A_1$. Next, this function is in turn relatively majorized by one where $q$ is replaced by $\qmin$. %
So by transitivity of relative majorization, the piecewise constant function taking values $\fmax/\qmin$ over $\amaxp$ and $\fmin/\qmin$ elsewhere weakly majorizes $f^+$ relative to $q$ in the cell. 

Thus, the contribution to the outer positive list is
\begin{align}
Z_\text{O,P}&\concat\{ \left(A_1^\text{max} \fmax ,  A_1^\text{max} \qmin \right)\}.
\end{align}
We also get a contribution to the outer negative list by similar reasoning as for the positive case,
\begin{align}
Z_\text{O,N}&\concat \{ \left(A_2^\text{max} \fmin, A_2^\text{max}\qmin \right) \}. 
\end{align}

\vspace{10pt}

\textbf{Case 3.} $\fmax \le 0$.

This is just a reversal of the first case. The steeper slope is $\fmin/\qmin$ and the shallower one is $\fmax/\qmax$. So the outer bound intersection point is at 
\begin{equation}
Q^* = \frac{\fl - \frac{\fmax}{\qmax} \ql}{\frac{\fmin}{\qmin} - \frac{\fmax}{\qmax}}.   
\end{equation}
The contribution to the negative list is then 
\begin{align}
Z_\text{O,N}&\concat\left\{ \left(\frac{\fmin}{\qmin}Q^*,Q^*\right) , \left( \frac{\fmax}{\qmax} (\ql - Q^*),\ql-Q^*  \right)\right\}.
\end{align}

In all three cases, we have that the positive mass accumulated is at least $\fl$, and that the shallowest Lorenz curve segment is if this is accumulated uniformly over the largest area $\qu$. With similar logic for the negative mass, the inner bounds are given by
\begin{equation}
    Z_\text{I,P}= \{ (\max\{\fl,0\}, \qu) \},
\end{equation}
and
\begin{equation}
    Z_\text{I,N} = \{ (\min\{\fu,0\}, \qu) \}. 
\end{equation}

With these choices of lists, we are guaranteed that
\begin{align}
    \declor{Z_\text{I,P}} \leq &\qdeclor{f_\text{c}} \leq \declor{Z_\text{O,P}} \\
    \inclor{Z_\text{I,N}} \geq &\qinclor{f_\text{c}} \geq \inclor{Z_\text{O,N}}. 
\end{align}

We assumed that the lower bounds on $q$ are always strictly positive. In cells where this is not the case, the corresponding function masses can be added to separate `asymptote' lists for positive and negative sign. Then, they contribute to a vertical shift of the Lorenz curve at the origin for the outer bounds (see below) while their contribution to the inner bounds is as above. The shift models the function accumulating those values over a region of arbitrarily small measure $\nu$. %

\subsection{Coarse graining}
Our second-order method works well in practice for single-mode states. However, the length of the discrete lists grows inverse polynomially in $\Delta$, which can be large for practical purposes. In particular, the multimode case becomes prohibitively slow in the step of tensoring the lists together. To remedy this, we observe that the curves can have many consecutive segments which only differ slightly in slope; our solution is to carefully combine them, while still maintaining the rigorous envelope guarantees. This leaves the curves essentially unchanged but the number of points greatly reduced. We use that the majorization preorder is preserved under direct sums. In practice, we use a histogram log-binning to choose which consecutive slopes get combined into a single segment. The following procedure we describe is agnostic to how the binning is performed.

So assume that some binning of the lists has been performed. For the rest of this section, let the index run only over the points in one group that are to be binned together. The process described is then repeated for each group. We assume the lists are ordered as they would be for computing the Lorenz curves and that the groups are contiguous.

Let a single such group be the set $\{(z_{k}^f,z_{k}^q) \}_{k=1}^m$ and let $\mathcal{I}$ be an index over all the groups. For the positive inner bound, the straight line segment connecting the endpoints of the portion of the curve arising from this group is an inner bound by concavity. For the negative inner bound, the same holds by convexity. Mathematically, each group is replaced by a single point,
\begin{align}
Z_\text{I,P}^\text{coarse} &= (\sum_k z_{k,\text{I,P}}^f, \sum_k z_{k,\text{I,P}}^q)\\
Z_\text{I,N}^\text{coarse}&= (\sum_k z_{k,\text{I,N}}^f, \sum_k z_{k,\text{I,N}}^q)
\end{align}

For the outer curves, extending the steepest slope within the group over the whole interval is guaranteed to remain an outer bound again by concavity or convexity. This gives the following coarse-grained lists,
\begin{align}
Z_\text{O,P}^\text{coarse}&=\left( \frac{z_{1,\text{O,P}}^f}{z_{1,\text{O,P}}^q} \sum_k z_{k,\text{O,P}}^q  , \sum_k z_{k,\text{O,P}}^q  \right)    \\
Z_\text{O,N}^\text{coarse}&=\left( \frac{z_{1,\text{O,N}}^f}{z_{1,\text{O,N}}^q} \sum_k z_{k,\text{O,N}}^q  , \sum_k z_{k,\text{O,N}}^q  \right) 
\end{align}

\subsection{Final processing}
To model the cutoff weights, which shift the outer curves vertically, we add an entry to each outer list,
\begin{align}
Z_\text{O,P}^\text{coarse}&\concat\{(w_f^+,0)\} \\
Z_\text{O,N}^\text{coarse}&\concat\{(-w_f^-,0)\}.
\end{align}
These entries are guaranteed to sort to the beginning of the list and shift the curves as required. If certain cell values have been added to an `asymptote' list, their corresponding positive or negative weights are also added at this stage to the shift.
Finally, we concatenate the outer and inner bound lists. 
\begin{align}
    Z_\text{O} &\coloneqq \{ Z_\text{O,P}^\text{coarse},Z_\text{O,N}^\text{coarse} \}\\
    Z_\text{I} &\coloneqq \{ Z_\text{I,P}^\text{coarse},Z_\text{I,N}^\text{coarse} \}
\end{align}
Interpreting these lists as the Lorenz curves of two piecewise constant functions, $f_\text{O}$ relative to $q_\text{O}$ and $f_\text{I}$ relative to $q_\text{I}$, by construction we are guaranteed that $(f_\text{O},q_\text{O}) \weak (f,q) \weak (f_\text{I},q_\text{I})$.

\subsection{Multimode}
For a tensor-product state, we simply tensor the outer lists $\bigotimes_i Z^i_\text{O}$ and the inner lists $\bigotimes_i Z^i_\text{I}$. Here, the multiplication used in the tensor product is
\begin{equation}
    (z_1^f,z_1^q) \times (z_2^f,z_2^q) = (z_1^f z_2^f,z_1^q z_2^q).
\end{equation}
The positive and negative Lorenz curves obtained from the lists are then rigorous outer and inner bounds because the majorization preorder is preserved under tensor products. Recall that the positive and negative Lorenz curves do not depend on terms of the opposite sign.

For a specific implementation, one could also account for errors in function evaluation or even errors from floating point arithmetic. For our implementation, we expect these to be negligible in comparison to the errors we have accounted for. We focus on Wigner functions built from standard functions such as polynomials and Gaussians. So, for each single mode term, we add (for upper bounds) or subtract (for lower bounds) a small tolerance from all the lists. In various parts of our derivation, we have specialized to functions over $\mathbbm{R}^2$ and square grids, but our approach can be applied to any measure space and arbitrary sampling grids. %

\section{Wigner--R\'enyi divergences of squeezed and attenuated cat states}
\label{app:wigner_cat}

In this section, we derive the Wigner--R\'enyi divergences of squeezed and cat states that are used in the main text.
We recall the Gaussian integral~\cite{zinn-justin_2010_PathInteg}, which we will use repeatedly in the derivations,
\begin{equation}\label{eq:gauss_integral}
    \int_{\mathbb{R}} {\rm d}z \; e^{-az^2+bz+c} = \sqrt{\frac{\pi}{a}}\exp(\frac{b^2}{4a}+c) \,, \quad a,b,c \in \mathbb{C} \,,\quad {\rm Re}(a) > 0 \,.
\end{equation}

\subsection{Wigner--R\'enyi divergences for squeezed states relative to vacuum}

Recall the Wigner function of the vacuum state $\ket{0}$ is
\begin{equation}\label{eq:Wvac}
    W_{\ket{0}}(x,p) = \frac{2}{\pi} e^{-2(x^2+p^2)} \,.
\end{equation}
The Wigner function of the squeezed state $\ket{r}$ is obtained from \cref{eq:Wvac} by a symplectic coordinate transformation,
\begin{equation}\label{eq:Wsqu}
    W_{\ket{r}}(x,p) = \frac{2}{\pi} \exp( -2\left( e^{-2r} x^2 + e^{2r} p^2 \right) ) \,.
\end{equation}
The Wigner $\alpha$-R\'enyi divergence of $\ket{r}$ relative to $\ket{0}$ is thus
\begin{align}
    D_\alpha(W_{\ket{r}}\|W_{\ket{0}})
    &= \frac{1}{\alpha-1} \log \frac{2}{\pi} \int_{\mathbb{R}^2} \exp(-2\alpha\left( e^{-2r} x^2 + e^{2r} p^2 \right) -2(1-\alpha)(x^2+p^2)) \;  {\rm d}x{\rm d}p \,, \nonumber\\
    &= -\frac{1}{2(\alpha-1)}\log\left( 1 - 4\alpha(\alpha-1) \sinh^2(r) \right) \,, \label{eq:Dr0}
\end{align}
for $\alpha \in (1, \infty)$.
In the first line, we substituted the Wigner functions from \cref{eq:Wvac,eq:Wsqu}, and in the second line, we integrated using \cref{eq:gauss_integral} which is well-defined precisely when the argument of the logarithm is greater than 0, or, equivalently, when
\begin{equation}
    |r| < \sinh^{-1}\left(\frac{1}{2\sqrt{\alpha(\alpha-1)}}\right) = \coth^{-1}(2\alpha-1) \,.
\end{equation}
Otherwise, $D_\alpha(W_{\ket{r}}\|W_{\ket{0}})$ diverges to $+\infty$.
As the R\'enyi order decreases, $\alpha \rightarrow 1^+$, the bound on admissible input squeezing diverges, $\coth^{-1}(2\alpha-1) \rightarrow +\infty$.
Moreover, \cref{eq:Dr0} tends to the KL divergence
\begin{equation}
    D_1(W_{\ket{r}}\|W_{\ket{0}}) 
    = 2\sinh^2(r) \,,
\end{equation}
which follows for example by applying L'H\^opital's rule in \cref{eq:Dr0}.

\subsection{Reference transformation for noisy LONs}

We consider depth-1 noisy LONs on $N$ modes described by one layer of a unitary LON followed by one layer of attenuating noise $(\N^{{\rm att}}_{\eta, \Sigma})^{\otimes N}$ as described in \cref{eq:noisyLON} of the main text.
We show that this set of operations maps copies of vacuum to possibly subnormalized copies of thermal states with mean photon number $\Sigma/2$.

The Gaussian noise channel $\N^{\rm att}_{\eta, \Sigma} = \N^{\rm GRN}_{\Sigma} \circ \N^{\rm loss}_\eta$ acts in phase space according to
\begin{equation}\label{eq:attenuator}
    \bmd \mapsto \sqrt{\eta} \bmd \,, \qquad V \mapsto \eta V + \left(1-\eta + \Sigma\right) \id \,.
\end{equation}
Hence, the Wigner kernel of this channel is
\begin{equation}\label{eq:attkernel}
    G_{\eta, \Sigma}(x,p \;|\; x',p') = \frac{2}{\pi(1-\eta + \Sigma)}\exp(-\frac{2}{1-\eta + \Sigma}\, \Big((x - \sqrt{\eta}x')^2 + (p - \sqrt{\eta}p')^2\Big) ) \,.
\end{equation}
It is readily verified that this is a stochastic integral operator as it is positive for all $(x,p) \in \mathbb{R}^2$ and normalized, i.e. $\int_{\mathbb{R}^2} {\rm d}x{\rm d}p\ G_{\eta, \Sigma}(x,p \;|\; x',p') = 1$ for all $(x',p') \in \mathbb{R}^2$, which follows directly from the Gaussian integral in \cref{eq:gauss_integral}.

Consider the thermal state $\rho_{\rm th}(\bar{n}) = e^{-\beta\omega \hat a^\dag \hat a}/\Tr(e^{-\beta\omega \hat a^\dag \hat a})$ with mean photon number
\begin{equation}
    \bar{n} \coloneqq \langle \hat a^\dag \hat a\rangle = (e^{\beta\omega} - 1)^{-1} \,.
\end{equation}
This state can be modeled as an attenuated vacuum state with $\Sigma = 2\bar{n}$ and arbitrary $\eta \in [0,1]$,
\begin{equation}
    \rho_{\rm th}(\bar{n}) 
    = \sum_{n=0}^\infty \frac{\bar{n}^n}{(1+\bar{n})^{n+1}} \ketbra{n}
    = \N^{\rm GRN}_{2\bar{n}} (\ketbra{0})
    \,.
\end{equation}
Therefore, its Wigner function is obtained by applying the kernel from \cref{eq:attkernel} on the vacuum Wigner function,
\begin{align}
    W_{\rho_{\rm th}(\bar{n})}(x,p) 
    &= \int_{\mathbb{R}^2} \ G_{\eta, \Sigma}(x,p \;|\; x',p')\ W_{\ketbra{0}}(x',p') {\rm d}x' {\rm d}p' \nonumber\\
    &= \frac{4}{\pi^2(1-\eta + \Sigma)} \int_{\mathbb{R}} \ \exp(-2x'^2-\frac{2}{1-\eta + \Sigma}(x - \sqrt{\eta}x')^2) {\rm d}x' \int_{\mathbb{R}} \ \exp(-2p'^2-\frac{2}{1-\eta + \Sigma}(p - \sqrt{\eta}p')^2) {\rm d}p' \nonumber\\
    &= \frac{2}{\pi(2\bar{n}+1)} \exp(-\frac{2}{2\bar{n}+1}(x^2 + p^2)) \,, \label{eq:Wthermal}
\end{align}
where in the second line we have factored out the integrals over $x'$ and $p'$ into a product, and the third line follows from two applications of the Gaussian integral in \cref{eq:gauss_integral}.
We also note that isotropic Gaussian-random noise is additive in the sense that applying $\N^{\rm GRN}_{\Sigma}$ on a thermal state $\rho_{\rm th}(\bar{n})$ results in a new thermal state $\rho_{\rm th}(\bar{n}')$ with mean photon number $\bar{n}' = \bar{n} + \Sigma/2$.

In the following, we will often replace $\Sigma$ with
\begin{equation}
    \lambda \coloneqq \frac{1}{1+\Sigma} = \frac{1}{2\bar{n}+1} = \frac{e^{\beta\omega}-1}{e^{\beta\omega}+1} \,,
\end{equation}
for readability.
Thus, $2\bar{n}+1 = \lambda^{-1}$, and we can rewrite \cref{eq:Wthermal} as
\begin{equation}\label{eq:Wthermal2}
     W_{\rho_{\rm th}\left(\bar{n}\right)}(x,p) 
     = \frac{2\lambda}{\pi} \exp(-2\lambda(x^2 + p^2)) \,.
\end{equation}

\subsection{Wigner function of attenuated cat states}

Recall the (even) cat state
\begin{equation}\label{eq:cat}
    \ket{\alpha}_{\rm cat} \coloneqq A (\ket{\alpha} + \ket{-\alpha}) \,, \quad A \coloneqq \frac{1}{\sqrt{2(1+e^{-2|\alpha|^2})}} \,,
\end{equation}
for some displacement $\alpha = X + iP$ where $X,P \in \mathbb{R}$.
The literature often distinguishes between even cat states $\ket{\alpha}_{\rm cat} \propto \ket{\alpha} + \ket{-\alpha}$ and odd cat states $\ket{\alpha}_{\rm cat} \propto \ket{\alpha} - \ket{-\alpha}$.
Here, we focus on even cat states, but we note that our analysis is directly applicable to odd cat state conversions as well.

Following our conventions, the Wigner function of the cat state in \cref{eq:cat} is
\begin{equation}\label{eq:Wcat}
    W_{\ketbra{\alpha}_{\rm cat}}(x,p) 
    = \frac{1}{\pi\left(1+e^{-2|\alpha|^2}\right)}\Big( e^{-2(x - X)^2 -2(p - P)^2} + e^{-2(x + X)^2 -2(p + P)^2} + 2e^{-2(x^2+p^2)}\cos(4(pX-xP)) \Big) \,.
\end{equation}
We provide a standard derivation (see also~\cite{kim1992}) for the Wigner function of an attenuated cat state via convolution of the Wigner function by the attenuating kernel in \cref{eq:attkernel}.
\begin{lemma}\label[lemma]{attcatW}
    Consider the attenuated cat state
    \begin{equation}
        \rho(\alpha; \eta, \Sigma) \coloneqq \N^{\rm att}_{\eta, \Sigma}(\ketbra{\alpha}_{\rm cat}) \,.
    \end{equation}
    Then, the Wigner function of $\rho(\alpha; \eta, \Sigma)$ is
    \begin{align}
        W_{\rho(\alpha; \eta, \Sigma)}(x,p) &= \frac{\lambda}{\pi\left(1+e^{-2|\alpha|^2}\right)}\Big( \exp(-2\lambda(x - \sqrt{\eta}X)^2 -2\lambda(p - \sqrt{\eta}P)^2) + \exp(-2\lambda(x + \sqrt{\eta}X)^2 -2\lambda(p + \sqrt{\eta}P)^2) \nonumber\\
        &\hspace{76pt}+ 2\exp(-2\lambda(x^2+p^2) -2(1-\lambda\eta)|\alpha|^2)\ \cos(4\lambda\sqrt{\eta}(pX-xP)) \Big) 
        \,, \label{eq:Wcatnoisy}
    \end{align}
    where $\lambda = 1/(1+\Sigma)$.
\end{lemma}
\noindent
As sanity checks, we can evaluate the Wigner function in \cref{eq:Wcatnoisy} for certain parameter values. For $\eta = 1$ and $\Sigma = 0$ ($\lambda = 1$), we obtain the Wigner function of a pure cat state from \cref{eq:Wcat}.
For $\eta = 0$ or $\alpha = 0$, we obtain the Wigner function of a thermal state from \cref{eq:Wthermal2}. 
\begin{proof}
    The interference term in \cref{eq:Wcat} can be written using the identity
    \begin{equation}\label{eq:cosdoubleangle}
        \cos\left(4(pX - xP)\right) = \operatorname{Re}\left(e^{i 4 P x}\right)\operatorname{Re}\left(e^{i 4 X p}\right) + \operatorname{Im}\left(e^{i 4 P x}\right)\operatorname{Im}\left(e^{i 4 X p}\right) \,.
    \end{equation}
    Substituting $1 - \eta + \Sigma = \frac{1-\lambda\eta}{\lambda}$ in the attenuating kernel from \cref{eq:attkernel}, and using \cref{eq:cosdoubleangle} allows us to write the Wigner function as
    \begin{align}
    W_{\rho(\alpha; \eta, \Sigma)}(x,p)
    &= \int_{\mathbb{R}^2} {\rm d}x' {\rm d}p'\ G_{\eta, \Sigma}(x,p \;|\; x',p')\ W_{\ketbra{\alpha}_{\rm cat}}(x',p') \nonumber\\
    &= \frac{2\lambda}{\pi^2(1 - \lambda\eta)\left(1+e^{-2|\alpha|^2}\right)} \int_{\mathbb{R}^2} {\rm d}x' {\rm d}p'\ \Bigg( \nonumber\\
    &\hspace{25pt} \exp(-\frac{2\lambda}{1 - \lambda\eta}(x - \sqrt{\eta}x')^2 -2(x' - X)^2)\ \exp(-\frac{2\lambda}{1 - \lambda\eta}(p - \sqrt{\eta}p')^2 -2(p' - P)^2) \nonumber\\
    &\hspace{15pt}+ \exp(-\frac{2\lambda}{1 - \lambda\eta}(x - \sqrt{\eta}x')^2 -2(x' + X)^2)\ \exp(-\frac{2\lambda}{1 - \lambda\eta}(p - \sqrt{\eta}p')^2 -2(p' + P)^2) \nonumber\\
    &\hspace{15pt}+ 2\operatorname{Re}\left(\exp(-\frac{2\lambda}{1 - \lambda\eta}(x - \sqrt{\eta}x')^2 -2x'^2 +i4Px')\right)\ \operatorname{Re}\left(\exp(-\frac{2\lambda}{1 - \lambda\eta}(p - \sqrt{\eta}p')^2 -2p'^2 +i4Xp')\right) \nonumber\\
    &\hspace{15pt}+ 2\operatorname{Im}\left(\exp(-\frac{2\lambda}{1 - \lambda\eta}(x - \sqrt{\eta}x')^2 -2x'^2 +i4Px')\right)\ \operatorname{Im}\left(\exp(-\frac{2\lambda}{1 - \lambda\eta}(p - \sqrt{\eta}p')^2 -2p'^2 +i4Xp')\right) \nonumber\\
    &\Bigg) \,.
    \end{align}
    Then, we evaluate each term in the expression with the help of \cref{eq:gauss_integral} to obtain
    \begin{align}
    W_{\rho(\alpha; \eta, \Sigma)}(x,p)
    &= \frac{2\lambda}{\pi^2(1 - \lambda\eta)\left(1+e^{-2|\alpha|^2}\right)} \Bigg( \nonumber\\
    &\hspace{25pt} \frac{\pi(1 - \lambda\eta)}{2}\ \exp(-2\lambda(x - \sqrt{\eta}X)^2)\ \exp(-2\lambda(p - \sqrt{\eta}P)^2) \nonumber\\
    &\hspace{15pt}+ \frac{\pi(1 - \lambda\eta)}{2}\ \exp(-2\lambda(x + \sqrt{\eta}X)^2)\ \exp(-2\lambda(p + \sqrt{\eta}P)^2) \nonumber\\
    &\hspace{15pt}+ \pi(1 - \lambda\eta)\ \exp(-2\lambda x^2 -2 (1-\lambda\eta) P^2)\ \cos(4\lambda\sqrt{\eta}xP)\ \exp(-2\lambda p^2 -2 (1-\lambda\eta) X^2)\ \cos(4\lambda\sqrt{\eta}pX) \nonumber\\
    &\hspace{15pt}+ \pi(1 - \lambda\eta)\ \exp(-2\lambda x^2 -2 (1-\lambda\eta) P^2)\ \sin(4\lambda\sqrt{\eta}xP)\ \exp(-2\lambda p^2 -2 (1-\lambda\eta) X^2)\ \sin(4\lambda\sqrt{\eta}pX) \nonumber\\
    &\Bigg) \,.
    \end{align}
    Collecting terms, we obtain \cref{eq:Wcatnoisy}.
\end{proof}

The Wigner function in \cref{eq:Wcatnoisy} can also be expressed as
\begin{equation}\label{eq:Wcatnoisy2}
    W_{\rho(\alpha; \eta, \Sigma)}(x, p) 
    = 
    \frac{2\lambda e^{-2\lambda(x^2+p^2)}}{\pi\left(1+e^{-2|\alpha|^2}\right)} \left( e^{-2\lambda\eta|\alpha|^2} \cosh\Big(4\lambda\sqrt{\eta}(xX+pP)\Big) + e^{-2(1-\lambda\eta)|\alpha|^2} \cos\Big(4\lambda\sqrt{\eta}(pX-xP)\Big) \right)
    \,.
\end{equation}
This expression makes apparent when the Wigner function becomes non-negative in the non-trivial cases when parameters $\alpha, \lambda, \eta$ are non-zero.
The prefactor is strictly positive, so we need to check the minimum of the function
\begin{equation}
    f(u,v) \coloneqq e^{-2\lambda\eta|\alpha|^2} \cosh\Big(4\lambda\sqrt{\eta}u\Big) + e^{-2(1-\lambda\eta)|\alpha|^2} \cos\Big(4\lambda\sqrt{\eta}v\Big)
    \,,
\end{equation}
where the lines $u(x,p) = xX + pP$ and $v(x,p) = pX - xP$ are perpendicular in phase space.
Hence, to find the minimum of $f(u,v)$, we can find the minimum of $\cosh(4\lambda\sqrt{\eta} u)$ and $\cos(4\lambda\sqrt{\eta}v)$ independently.
Choosing
\begin{equation}\label{eq:uvmin}
    u = 0 
    \,, \qquad
    v = \frac{\pi}{4\lambda\sqrt{\eta}}
    \,,
\end{equation}
we obtain
\begin{equation}
    \min_{(x,p) \in \mathbb{R}^2} f(u(x,p),v(x,p))
    = e^{-2\lambda\eta|\alpha|^2} - e^{-2(1-\lambda\eta)|\alpha|^2}
    \,.
\end{equation}
Therefore, the Wigner function is non-negative everywhere if and only if $e^{-2\lambda\eta|\alpha|^2} \ge e^{-2(1-\lambda\eta)|\alpha|^2}$, or, equivalently
\begin{equation}
    \frac{\eta}{1+\Sigma} \le \frac{1}{2}
    \,.
\end{equation}

\subsection{Wigner 2-R\'enyi divergence for attenuated cat states relative to a thermal state}

We prove \cref{cat2renyi} which derives the Wigner 2-R\'enyi divergence of an attenuated cat state $\N^{\rm att}_{\eta, \Sigma}(\ketbra{\alpha}_{\rm cat})$ relative to a thermal state $\rho_{\rm th}(\bar{n})$.
    
\begingroup
\allowdisplaybreaks[1]
\begin{proof}
    Let $\lambda = 1/(1+\Sigma)$ and $\lambda_{\rm env} = 1/(1+\Sigma_{\rm env})$, where $\Sigma_{\rm env} = 2\bar{n}$.
    Plugging the Wigner functions that we calculated in \cref{eq:Wthermal2} and \cref{eq:Wcatnoisy} into the 2-R\'enyi divergence, and using the cosine identity from \cref{eq:cosdoubleangle}, we obtain
    \begin{align}
        D_2(W_{\rho(\alpha; \eta, \Sigma)} \| W_{\rho_{\rm th}(\bar{n})})
        &= \log\int  \frac{W^2_{\rho(\alpha; \eta, \Sigma)}(x,p)}{W_{\rho_{\rm th}(\bar{n})}(x,p)} {\rm d}x{\rm d}p \nonumber\\
        &= \log \frac{\lambda^2\lambda_{\rm env}^{-1}}{2\pi\left(1+e^{-2|\alpha|^2}\right)^2} + \log\int_{\mathbb{R}^2}  e^{2\lambda_{\rm env}(x^2 + p^2)} \Bigg( e^{-2\lambda(x - \sqrt{\eta}X)^2 -2\lambda(p - \sqrt{\eta}P)^2} \nonumber\\
        &\hspace{20pt}+ e^{-2\lambda(x + \sqrt{\eta}X)^2 -2\lambda(p + \sqrt{\eta}P)^2} + 2e^{-2\lambda(x^2+p^2) -2(1-\lambda\eta)|\alpha|^2} \cos(4\lambda\sqrt{\eta}(pX-xP)) \Bigg)^2 {\rm d}x{\rm d}p 
        \,.
    \end{align}
    Expanding the square, and using the identity $2\cos^2(x) = 1+\cos(2x)$, we get
    \begin{align}
        D_2(W_{\rho(\alpha; \eta, \Sigma)} \| W_{\rho_{\rm th}(\bar{n})})
        &= \log\frac{\lambda^2\lambda_{\rm env}^{-1}}{2\pi\left(1+e^{-2|\alpha|^2}\right)^2} + \log \int_{\mathbb{R}^2} {\rm d}x{\rm d}p\ \Bigg( \nonumber\\
        &\hspace{30pt} e^{2\lambda_{\rm env} x^2 - 4\lambda(x - \sqrt{\eta}X)^2}\ e^{2\lambda_{\rm env} p^2 - 4\lambda(p - \sqrt{\eta}P)^2} \nonumber\\
        &\hspace{20pt}+ e^{2\lambda_{\rm env} x^2 - 4\lambda(x + \sqrt{\eta}X)^2}\ e^{2\lambda_{\rm env} p^2 - 4\lambda(p + \sqrt{\eta}P)^2} \nonumber\\
        &\hspace{20pt}+ 2e^{2\lambda_{\rm env} x^2 - 2\lambda(x - \sqrt{\eta}X)^2 - 2\lambda(x + \sqrt{\eta}X)^2}\ e^{2\lambda_{\rm env} p^2 - 2\lambda(p - \sqrt{\eta}P)^2 - 2\lambda(p + \sqrt{\eta}P)^2} \nonumber\\
        &\hspace{20pt}+ 4e^{-2(1-\lambda\eta)|\alpha|^2}\ e^{2(\lambda_{\rm env}-\lambda) x^2 - 2\lambda(x - \sqrt{\eta}X)^2}\ e^{2(\lambda_{\rm env}-\lambda) p^2 - 2\lambda(p - \sqrt{\eta}P)^2}\cos(4\lambda\sqrt{\eta}(pX-xP)) \nonumber\\
        &\hspace{20pt}+ 4e^{-2(1-\lambda\eta)|\alpha|^2}\ e^{2(\lambda_{\rm env}-\lambda) x^2 - 2\lambda(x + \sqrt{\eta}X)^2}\ e^{2(\lambda_{\rm env}-\lambda) p^2 - 2\lambda(p + \sqrt{\eta}P)^2} \cos(4\lambda\sqrt{\eta}(pX-xP)) \nonumber\\
        &\hspace{20pt}+ 2e^{-4(1-\lambda\eta)|\alpha|^2}\ e^{2(\lambda_{\rm env}-2\lambda) x^2}\ e^{2(\lambda_{\rm env}-2\lambda) p^2} \nonumber\\
        &\hspace{20pt}+ 2e^{-4(1-\lambda\eta)|\alpha|^2}\ e^{2(\lambda_{\rm env}-2\lambda) x^2}\ e^{2(\lambda_{\rm env}-2\lambda) p^2} \cos(8\lambda\sqrt{\eta}(pX-xP)) \nonumber\\
        &\Bigg)
        \,.
    \end{align}
    
    Then, we use the cosine identity from \cref{eq:cosdoubleangle} and bring the exponents to standard quadratic form. 
    We further denote $\bar{\lambda} \coloneqq 2\lambda - \lambda_{\rm env}$, leading to
    \begin{align}
        D_2(W_{\rho(\alpha; \eta, \Sigma)} \| W_{\rho_{\rm th}(\bar{n})}) 
        &= \log\frac{\lambda^2(2\lambda-\bar{\lambda})^{-1}}{2\pi\left(1+e^{-2|\alpha|^2}\right)^2} + \log \int_{\mathbb{R}^2}  \Bigg( \nonumber\\
        & e^{-2\bar{\lambda} x^2 + 8\lambda\sqrt{\eta}X x - 4\lambda\eta X^2}\ e^{-2\bar{\lambda} p^2 + 8\lambda\sqrt{\eta}P p - 4\lambda\eta P^2} \nonumber\\
        &\hspace{-10pt}+ e^{-2\bar{\lambda} x^2 - 8\lambda\sqrt{\eta}X x - 4\lambda\eta X^2}\ e^{-2\bar{\lambda} p^2 - 8\lambda\sqrt{\eta}P p - 4\lambda\eta P^2} \nonumber\\
        &\hspace{-10pt}+ 2e^{-2\bar{\lambda} x^2 - 4\lambda\eta X^2}\ e^{-2\bar{\lambda} p^2 - 4\lambda\eta P^2} \nonumber\\
        &\hspace{-10pt}+ 4e^{-2(1-\lambda\eta)|\alpha|^2}\operatorname{Re}\left(e^{-2\bar{\lambda} x^2 + 4\lambda\sqrt{\eta}X x + i4\lambda\sqrt{\eta}P x - 2\lambda\eta X^2}\right)\ \operatorname{Re}\left(e^{-2\bar{\lambda} p^2 + 4\lambda\sqrt{\eta}P p + i4\lambda\sqrt{\eta}X p - 2\lambda\eta P^2}\right) \nonumber\\
        &\hspace{-10pt}+ 4e^{-2(1-\lambda\eta)|\alpha|^2}\operatorname{Im}\left(e^{-2\bar{\lambda} x^2 + 4\lambda\sqrt{\eta}X x + i4\lambda\sqrt{\eta}P x - 2\lambda\eta X^2}\right)\ \operatorname{Im}\left(e^{-2\bar{\lambda} p^2 + 4\lambda\sqrt{\eta}P p + i4\lambda\sqrt{\eta}X p - 2\lambda\eta P^2}\right) \nonumber\\
        &\hspace{-10pt}+ 4e^{-2(1-\lambda\eta)|\alpha|^2}\operatorname{Re}\left(e^{-2\bar{\lambda} x^2 - 4\lambda\sqrt{\eta}X x + i4\lambda\sqrt{\eta}P x - 2\lambda\eta X^2}\right)\ \operatorname{Re}\left(e^{-2\bar{\lambda} p^2 - 4\lambda\sqrt{\eta}P p + i4\lambda\sqrt{\eta}X p - 2\lambda\eta P^2}\right) \nonumber\\
        &\hspace{-10pt}+ 4e^{-2(1-\lambda\eta)|\alpha|^2}\operatorname{Im}\left(e^{-2\bar{\lambda} x^2 - 4\lambda\sqrt{\eta}X x + i4\lambda\sqrt{\eta}P x - 2\lambda\eta X^2}\right)\ \operatorname{Im}\left(e^{-2\bar{\lambda} p^2 - 4\lambda\sqrt{\eta}P p + i4\lambda\sqrt{\eta}X p - 2\lambda\eta P^2}\right) \nonumber\\
        &\hspace{-10pt}+ 2e^{-4(1-\lambda\eta)|\alpha|^2}e^{-2\bar{\lambda} x^2}\ e^{-2\bar{\lambda} p^2} \nonumber\\
        &\hspace{-10pt}+ 2e^{-4(1-\lambda\eta)|\alpha|^2}\operatorname{Re}\left(e^{-2\bar{\lambda} x^2 + i8\lambda\sqrt{\eta}P x}\right)\ \operatorname{Re}\left(e^{-2\bar{\lambda} p^2 + i8\lambda\sqrt{\eta}X p}\right) \nonumber\\
        &\hspace{-10pt}+ 2e^{-4(1-\lambda\eta)|\alpha|^2}\operatorname{Im}\left(e^{-2\bar{\lambda} x^2 + i8\lambda\sqrt{\eta}P x}\right)\ \operatorname{Im}\left(e^{-2\bar{\lambda} p^2 + i8\lambda\sqrt{\eta}X p}\right) \nonumber\\
        &\Bigg){\rm d}x{\rm d}p
        \,.
    \end{align}
    The next step is to evaluate the expression term by term with the help of the Gaussian integral from \cref{eq:gauss_integral}.
    This integral blows up to $+\infty$ unless $\bar{\lambda} > 0$, or equivalently $\Sigma < 1+2\Sigma_{\rm env}$, in which case it instead results in
    \begin{align}
        D_2(W_{\rho(\alpha; \eta, \Sigma)} \| W_{\rho_{\rm th}(\bar{n})}) 
        &= \log\frac{\lambda^2(2\lambda - \bar{\lambda})^{-1}}{2\pi\left(1+e^{-2|\alpha|^2}\right)^2} + \log \Bigg( \nonumber\\
        &\hspace{30pt} \frac{\pi \bar{\lambda}^{-1}}{2}\ \exp\left(\frac{4\lambda(2\lambda-\bar{\lambda})\eta |\alpha|^2}{\bar{\lambda}}\right) \nonumber\\
        &\hspace{20pt}+ \frac{\pi \bar{\lambda}^{-1}}{2}\ \exp\left(\frac{4\lambda(2\lambda-\bar{\lambda})\eta |\alpha|^2}{\bar{\lambda}}\right) \nonumber\\
        &\hspace{20pt}+ \pi \bar{\lambda}^{-1}\ e^{-4\lambda\eta |\alpha|^2} \nonumber\\
        &\hspace{20pt}+ 2\pi \bar{\lambda}^{-1} e^{-2|\alpha|^2}\cos^2\left(\frac{4\lambda^2\eta XP}{\bar{\lambda}}\right) \nonumber\\
        &\hspace{20pt}+ 2\pi \bar{\lambda}^{-1} e^{-2|\alpha|^2}\sin^2\left(\frac{4\lambda^2\eta XP}{\bar{\lambda}}\right) \nonumber\\
        &\hspace{20pt}+ 2\pi \bar{\lambda}^{-1} e^{-2|\alpha|^2}\cos^2\left(\frac{4\lambda^2\eta XP}{\bar{\lambda}}\right) \nonumber\\
        &\hspace{20pt}+ 2\pi \bar{\lambda}^{-1} e^{-2|\alpha|^2}\sin^2\left(\frac{4\lambda^2\eta XP}{\bar{\lambda}}\right) \nonumber\\
        &\hspace{20pt}+ \pi \bar{\lambda}^{-1} e^{-4(1-\lambda\eta)|\alpha|^2} \nonumber\\
        &\hspace{20pt}+ \pi \bar{\lambda}^{-1} e^{-4(1-\lambda\eta)|\alpha|^2}\ \exp\left(-\frac{8\lambda^2\eta |\alpha|^2}{\bar{\lambda}}\right) \nonumber\\
        &\hspace{20pt}+ 0 \nonumber\\
        &\Bigg) \nonumber\\
        &\hspace{-4cm}= \log\left( \frac{\zeta^2(2\zeta-1)^{-1}}{2(1+e^{-2|\alpha|^2})^2} \left( e^{4\lambda(2\zeta-1)\eta |\alpha|^2} + e^{-4\lambda\eta |\alpha|^2} + 4e^{-2|\alpha|^2} + e^{-4(1-\lambda\eta) |\alpha|^2}  + e^{-4(1-\lambda\eta) |\alpha|^2} e^{-8\lambda\zeta\eta |\alpha|^2}\right) \right) \nonumber\\
        &= \log\left( \frac{\zeta^2(2\zeta-1)^{-1}\left( \cosh\left(2|\alpha|^2 + 4\lambda(\zeta-1)\eta |\alpha|^2\right) \cosh\left(4\lambda\zeta\eta |\alpha|^2\right) + 1 \right)}{\cosh(2|\alpha|^2) + 1} \right)
        \,,
    \end{align}
    where we introduced $\zeta \coloneqq \lambda / \bar{\lambda}$.
    Substituting $\zeta = \lambda^{-1} (1+2\bar{n}) K$, this is equivalent to \cref{eq:D2catnoisy}.
\end{proof}
\endgroup

\subsection{Bound on cat state conversions}

Here, we prove \cref{thm:catstates}.
\begin{proof}
The thermal state $\rho_{\rm th}(\bar{n})$ has displacement vector $\bmd = \bmo$ and covariance matrix $V_{\rm th} = (2\bar{n} + 1)\id_2$.
Let
\begin{equation}
    q_{\kappa, M}(\bm{r}) 
    = \prod_{i=1}^M e^{-2\kappa \|\bm{r}_i\|_2^2} \; W_{\rho_{\rm th}(\bar{n})}(\bm{r}_i)
    \,,
\end{equation}
and $\rho' \coloneqq \rho(\alpha'; \eta', \Sigma')$.
Then, the finiteness condition becomes
\begin{equation}
    J_2(W_{\rho'^{\otimes M}} \| q_{\kappa, M}) = J_2(W_{\rho'} \| q_{\kappa, 1})^M
    =
    \left( \int_{\mathbb{R}^{2}} W_{\rho'}(\bm{r})^2 \; e^{2\kappa \|\bm{r}\|_2^2} \; W_{\rho_{\rm th}(\bar{n})}(\bm{r})^{-1} \; {\rm d}^2\bm{r} \right)^M
    <
    \infty
    \,.
\end{equation}
The Wigner functions of the thermal state and of the attenuated cat state are given in \cref{eq:Wthermal2} and in \cref{attcatW}, respectively.
Examining their forms, we see that each term in the integrand of $J_2(W_{\rho'} \| q_\kappa)$ is a complex exponential in $\bm{r}$ with the quadratic prefactor equal to $-4(1 + \Sigma')^{-1} + 2\kappa + 2(1 + 2\bar{n})^{-1}$.
Hence, the finiteness condition is satisfied if and only if this prefactor is negative, namely when
\begin{equation}
    \kappa < \frac{2}{1 + \Sigma'} - \frac{1}{1 + 2\bar{n}}
    \,.
\end{equation}
We set the tail margin to half this maximum value,
\begin{equation}
    \kappa = \kappa_\star \coloneqq \frac{1 + 4\bar{n} - \Sigma'}{2(1+\Sigma')(1+2\bar{n})}
    \,,
\end{equation}
and calculate the exponent at $\kappa = \kappa_\star$ using \cref{thm:semicontinuity},
\begin{equation}
    c 
    = \frac{2\kappa_\star}{2\kappa_\star + \|V^{-1}\|_{\infty}}
    = \frac{1 + 4\bar{n} - \Sigma'}{2(2\bar{n}+1)}
    \,.
\end{equation}
\cref{thm:semicontinuity} also gives us the form of the prefactor at $\kappa = \kappa_\star$,
\begin{align}
    \Delta 
    &= \left( \frac{4\det\left(V_{\rm th}^{\oplus M}\right)^{1/4}}{\kappa_\star^{M/2}} \right)^c \frac{J_2\left(W_{\rho'^{\otimes M}} \| q_{\kappa_\star, M}\right)^{1-\frac{c}{2}}}{J_2\left(W_{\rho'^{\otimes M}} \| W_{\rho_{\rm th}(\bar{n})^{\otimes M}}\right)}
    \nonumber\\
    &= 4^c \left( \left(\frac{2\bar{n}+1}{\kappa_\star}\right)^{\frac{c}{2}} \frac{J_2(W_{\rho'} \parallel q_{\kappa_\star, 1})^{1 - \frac{c}{2}}}{J_2(W_{\rho'} \parallel W_{\rho_{\rm th}(\bar{n})})} \right)^M
    > 0
    \,.\label{eq:Delta_intermediate}
\end{align}
Finally, the bound on $p^2 (1-\Delta\epsilon^c)$ follows from \cref{corollaryprobbound_extended} with $t=1$ and the additivity of the Wigner 2-R\'enyi divergence.
\end{proof}
We note that $\Delta$ can be further calculated explicitly in terms of the physical parameters by evaluating the integrals $J_2(W_{\rho'} \parallel q_{\kappa_\star, 1}) \,, J_2(W_{\rho'} \parallel W_{\rho_{\rm th}(\bar{n})})$ via \cref{cat2renyi}.

\section{Fock Lorenz curves}\label{app:fock}

Here we derive the small-$s$ expressions for the Fock Wigner Lorenz curves presented in the main text.
\subsection{Multiphoton single-mode curves}
Recall the distribution function is $D^q_f(t)=\nu\{f>qt\}$. In polar coordinates, the Fock Wigner functions are $W_{\ket{n}}=\frac{2}{\pi} e^{-2r^2} (-1)^n L_n(4r^2)$ where $L_n$ is the $n^\text{th}$ Laguerre polynomial. The vacuum measure is $\nu=4r e^{-2r^2}dr$. The distribution function is then
\begin{equation}\label{fockdistr}
    D^{W_{\ket0}}_{W_{\ket n}}(t)=\nu \{ (-1)^n L_n(4r^2)>t \}.
\end{equation}
For large $t$, we approximate the Laguerre polynomial with its leading order term, $(-1)^n L_n(x)\approx x^n/n!$. Then, the inequality in \cref{fockdistr} simplifies,
\begin{align}\label{fockrlimit}
    (4r^2)^n/n!>t \iff r>\frac{1}{2} (tn!)^{1/2n}.
\end{align}
It should be understood that all subsequent statements assume the large-$t$ limit. Substituting \cref{fockrlimit} into \cref{fockdistr} and integrating, we obtain an explicit expression for the distribution function,
\begin{align}
 D^{W_{\ket0}}_{W_{\ket n}}(t)&\approx\nu \left\{ r>\frac{1}{2} (tn!)^{1/2n} \right\}\\
&=\int_{\frac{1}{2} (tn!)^{1/2n}}^\infty 4r e^{-2r^2} \mathrm{d}r\\
&=\exp{-\frac{1}{2} (tn!)^{1/n} }.
\end{align}
Setting this equal to $u$ and solving for $t$ gives the decreasing rearrangement at small $u$,
\begin{align}
    (W_{\ket n}^+)^{\downarrow,W_{\ket 0}}(u) \approx \frac{1}{n!} (-2 \log u)^n.
\end{align}
The positive Lorenz curve at small $s$ is then
\begin{align}
    \qqdeclor{W_{\ket n}}{W_{\ket 0}}&= \int_0^s (W_{\ket n}^+)^{\downarrow,W_{\ket 0}}(u) {\rm d}u\\
    &\approx\int_0^s \frac{1}{n!} (- 2 \log u)^n {\rm d}u\\
    &\approx 2^n  \frac{s(-\log s)^n}{n!}.
\end{align}

\subsection{Multimode single-photon curves}

Letting $r_k$ be the radial coordinate for the $k^{th}$ copy of $\ket{1}$, change variables to 
$y_k=t^{-1/n}  (4r_k^2-1)$.  
Then, 
\begin{align}
    D^{W_{\ket 0}^{\otimes n}}_{W_{\ket 1}^{\otimes n}}(t) &= \nu \{ y_1\ldots y_n  \geq 1\}\\
    &= t e^{-n/2}  \int_{-t^{-1/n}}^\infty \frac{1}{2} e^{- t^{1/n} y_1/2} \ldots \int_{-t^{-1/n}}^\infty \frac{1}{2} e^{- t^{1/n} y_n/2} \   \mathcal{I}_{y_1\ldots y_n \geq 1} {\rm d}y_1 \ldots {\rm d} y_n\label{indicatorfnfock}
\end{align}

A lower bound on the distribution function follows by observing that the indicator constraint is satisfied if all $y_i\geq 1$,
\begin{align}
    D^{W_{\ket 0}^{\otimes n}}_{W_{\ket 1}^{\otimes n}}(t) &\geq t e^{-n/2}  \int_{1}^\infty \frac{1}{2} e^{- t^{1/n} y_1/2} \ldots \int_{1}^\infty \frac{1}{2} e^{- t^{1/n} y_n/2}  {\rm d} y_1 \ldots {\rm d} y_n\\
    &= t e^{-n/2} \left(\int_{1}^\infty \frac{1}{2} e^{- t^{1/n} y/2}{\rm d} y\right)^n\\
    &= e^{-n({t^{1/n}}+1)/2}\\
    &\approx e^{-n{t^{1/n}}/2}.
\end{align}
In fact, this expression is tight at large $t$ as we now show. To do so, we need to show the small negative values of the integration variables do not contribute. Observe that the indicator function in \cref{indicatorfnfock} is only satisfied when an even number of variables are negative, WLOG the first $2m$. Their magnitude is bounded by $\abs{y_k}<1/t^{1/n}$, so for the constraint the other variables must satisfy $\prod_{i>2m} y_i>t^{2m/n}$. By the arithmetic-geometric-mean inequality, $\sum_{i>2m} y_i > (n-2m)  t^{2m/(n(n-2m))} $. But this means the exponent of the integrand in \cref{indicatorfnfock} is $t^{1/(n-2m)}$ so this contribution is not leading order compared to the $m=0$ term. Thus, for our large-$t$ calculation, we can restrict to nonnegative variables and simplify the integral,
\begin{align}
    D^{W_{\ket 0}^{\otimes n}}_{W_{\ket 1}^{\otimes n}}(t) &\approx t e^{-n/2}  \int_{0}^\infty \frac{1}{2} e^{- t^{1/n} y_1/2} \ldots \int_{0}^\infty \frac{1}{2} e^{- t^{1/n} y_n/2} \   \mathcal{I}_{y_1\ldots y_n \geq 1} \; {\rm d}y_1 \ldots {\rm d}y_n\\
    &= t e^{-n/2}  \int_{0}^\infty \frac{1}{2} e^{- t^{1/n} y_1/2} ... \int_{ \frac{1}{\prod_1^{n-1} y_i } }^\infty \frac{1}{2} e^{- t^{1/n} y_n/2} \; {\rm d} y_1  \ldots {\rm d} y_n \\
    &= t^{1-1/n} e^{-n/2}  \int_{0}^\infty \frac{1}{2} e^{- t^{1/n} y_1/2} ... \int_{0}^\infty \frac{1}{2} e^{-t^{1/n} y_{n-1}/2} \; \exp( -\frac{t^{1/n}}{2\prod_{i=1}^{n-1} y_i } ) \; {\rm d}y_1 \ldots {\rm d}y_{n-1} \\
    &=  t^{1-1/n} e^{-n/2} \frac{1}{2^{n-1}}  \int_{0}^\infty  ... \int_{0}^\infty  \exp( -\frac{t^{1/n}}{2} \left(\sum_{i=1}^{n-1} y_i+ \frac{1}{\prod_{i=1}^{n-1} y_i }\right) ) \; {\rm d}y_1 \ldots {\rm d}y_{n-1}
\end{align}
where in the second line we implement the constraint by changing the lower limit of integration of $y_n$ and in the third line we integrate over $y_n$. We can now approximate this integral by Laplace's method. At large $t$, the integral is dominated where the parenthetical term in the integrand's exponent is the smallest. This occurs when all $y_i=1$, as can be seen from the arithmetic-geometric mean inequality or by taking partial derivatives. The Hessian at this point is the $(n-1)\cross(n-1)$ matrix
\begin{equation}
    \begin{pmatrix}
        2 & 1 & 1 &\ldots\\
        1 & 2 & 1 &\ldots \\
        1 & 1 & 2 &\ldots\\
        \vdots & \vdots & \vdots &\ddots
    \end{pmatrix}=\begin{pmatrix}
        1 & 1 & 1 &\ldots\\
        1 & 1 & 1 &\ldots \\
        1 & 1 & 1 &\ldots\\
        \vdots & \vdots & \vdots &\ddots
    \end{pmatrix}+
    \begin{pmatrix}
        1 & 0 & 0 &\ldots\\
        0 & 1 & 0 &\ldots \\
        0 & 0 & 1 &\ldots\\
        \vdots & \vdots & \vdots &\ddots
    \end{pmatrix}.
\end{equation}
The decomposition makes it clear that the Hessian has one eigenvalue $n$ and the rest unity, so its determinant is $n$. By Laplace's formula, this then gives 
\begin{align}
    D^{W_{\ket 0}^{\otimes n}}_{W_{\ket 1}^{\otimes n}}(t) &\approx t^{1-1/n} e^{-n/2} \frac{1}{2^{n-1}} \left(\frac{4\pi}{t^{1/n}}\right)^\frac{n-1}{2}  \frac{e^{-n \frac{t^{1/n}}{2} } }{\sqrt{n}}\\
    &= t^\frac{1-1/n}{2} \frac{e^{-n/2}} {\sqrt{n}} \pi^\frac{n-1}{2}  e^{-n \frac{t^{1/n}}{2} }.
\end{align}
Therefore, to leading order, $\log D^{W_{\ket 0}^{\otimes n}}_{W_{\ket 1}^{\otimes n}}(t) \approx -\frac{n}{2} t^{1/n}$.

To find the decreasing rearrangement, we set the distribution function equal to $u=e^{-n{t^{1/n}}/2}$ and invert,
\begin{align}
    \left(\frac{-2}{n}\log u\right)^n&=t,
\end{align}
so 
\begin{equation}
    ({W_{\ket 1}^{\otimes n}}^+)^{\downarrow,W_{\ket0}^{\otimes n}}(u) \approx \left(\frac{-2}{n}\log u\right)^n
\end{equation}
at small $u$.
Then, keeping only the leading order term as $s\rightarrow 0$,
\begin{align}
    \qqdeclor{W_{\ket 1}^{\otimes n}}{W_{\ket0}^{\otimes n}}&\approx \int_0^s  \left(\frac{-2}{n}\log u\right)^n {\rm d}u \\
    &\approx s\left(\frac{-2}{n}\log s\right)^n \\
    &=\frac{2^n}{n^n} s (-\log s)^n.
\end{align}

\bibliography{bib}

\end{document}